\documentclass[UTF8]{ctexart}

\usepackage{booktabs}
\usepackage{pifont}

\usepackage[a4paper, margin = 2.33cm, centering, includefoot]{geometry}

\usepackage[svgnames, dvipsnames]{xcolor} 

\usepackage[unicode=true, bookmarksnumbered=true, colorlinks=true]{hyperref}
\colorlet{linkblue}{blue!45!black}

\hypersetup{
    citecolor=magenta, 
    linkcolor=linkblue, 
    urlcolor=gray
}

\usepackage{latexsym,mathtools,mathalfa,amssymb,amsmath,amsbsy,amsopn,amstext,amsthm,amsxtra,bm, bbm,calc,ifpdf,cancel,extarrows,manfnt}
\usepackage{enumerate,graphicx,wrapfig,subfig,diagbox,listings,multicol,multirow,makeidx,fontspec,marvosym}

\newtheorem{theorem}{Theorem}[section]
\newtheorem{corollary}{Corollary}[section] 

\newtheorem{definition}{Definition}[section]
\newtheorem{remark}{Remark}[section]
\newtheorem{lemma}{Lemma}[section]
\newtheorem{assumption}{Assumption}[section]
\newtheorem{proposition}{Proposition}[section]
 
\newtheorem{example}{Implementation}

\usepackage{fancyhdr}
\renewcommand{\sectionmark}[1]{}

\usepackage{caption}
\usepackage{lmodern}
\usepackage[T1]{fontenc}
\usepackage{anyfontsize}
\usepackage{tikz}
\usetikzlibrary{decorations.pathreplacing, calligraphy, patterns, arrows.meta, positioning}

\ctexset{
 section = {
    format = \Large\bfseries\raggedright,
    numberformat = {},
    nameformat = {}  
  }, 
  abstractname = {Abstract},
  contentsname = {Contents},
  part = {
    name   = {Part~},
    number = \Roman{part}
  },
  bibname = {References},
  figurename = {Figure},  
  tablename  = {Table}   
}

\usepackage[authoryear,round, sort&compress]{natbib} 

\newcommand{\E}{\mathbb{E}}					
\newcommand{\Var}{\operatorname{Var}}		
\newcommand{\Cov}{\operatorname{Cov}}		
\newcommand{\Bias}{\operatorname{Bias}}
\newcommand{\MSE}{\operatorname{MSE}}		

\usepackage[en-US]{datetime2} 

\title{Bias-Variance Tradeoff of Matching Prior to Difference-in-Differences}
\author{Mingxuan Ge\thanks{Department of Business Economics and Public Policy, Michigan Ross School of Business: \href{mkgg@umich.edu}{mkgg@umich.edu}}, \ Dae Woong Ham\thanks{Department of Technology and Operations, Ross School of Business: \href{daewoong@umich.edu}{daewoong@umich.edu}}}
\date{This draft version: \today}

\begin{document}
\maketitle

\begin{abstract}
Quasi-experimental causal inference methods have become central in empirical operations management (OM) for guiding managerial decisions. Among the relevant methods, empiricists utilize the Difference-in-Differences (DiD) estimator, which is built on the ``parallel trends” assumption. To increase the plausibility of parallel trends holding, researchers often match treated and control units before applying DiD, motivated by the intuition that matched groups are more likely to have the same post-treatment outcome absent treatment. Although this practice has been studied previously, the focus was solely on the bias. Yet bias is only one side to the story; the tradeoff in variance, and hence mean squared error (MSE), is left unexplored. In this work, we not only generalize earlier bias results under weaker assumptions but also fill this gap by analyzing properties of variance and MSE, a practically relevant metric for managerial decision making. We show, under a linear structural model with unobserved time-varying confounders, that variance results contrast with established bias insights: when considering variance, matching on observed covariates prior to DiD is not always recommended over the classic (unmatched) DiD due to a sample size tradeoff. Furthermore, matching additionally on pre-treatment outcome(s) is always beneficial as such tradeoff no longer exists once matching is performed. Given our novel contrasting results to the established literature, we characterize full bias-variance tradeoffs among three estimators and recommend incorporating MSE as an additional metric if applied researchers weigh bias and variance equally. We further give practitioner-friendly guidelines with theoretical guarantees that determine whether applied researchers should match (and on what variables) for their applications. As an illustration, we apply these guidelines to re-examine a recent empirical study that matches prior to DiD to study how the introduction of monetary incentives by a knowledge-sharing platform affects general engagement on the platform. Our results show that the authors' decision was both warranted and critical to produce a credible causal estimate.  
\end{abstract}

\section{Introduction}
\noindent Over the past two decades, empirical work in operations management (OM) and related fields has increased significantly and become integral (see survey papers \cite{gupta2006empirical}; \cite{ho2017om}; \cite{terwiesch2020om}; \cite{yilmaz2024causal}). Many empirical questions in the OM field revolve around the effect of adopting a program or implementing a policy (i.e., treatment) on various outcomes across different contexts such as manufacturing, healthcare and digital platforms. To guide these managerial decisions credibly, applied researchers must establish causal relationships between variables of interest. Although randomized experiments are often regarded as the gold standard for causal inference, in many applied settings, applied researchers must instead rely on observational (quasi-experimental) data, where concerns of endogeneity and selection bias complicate credible causal inference. In such settings, the Difference-in-Differences (DiD) estimator has become a widely adopted empirical tool due to its ease of implementation and natural fit with panel data. The DiD estimator is built on the parallel trends (PT) assumption: the treatment group, absent treatment, would change similarly to the control group over time. A central concern with this assumption is whether the chosen comparison group is suitable so that PT holds reliably.

To account for possible violations of parallel trends between the treatment and control group, a common practice is to apply matching to select a subset of a large number of potential comparison units to resemble the treatment group. Such matching prior to DiD (referred to as ``M-DiD'') strategy formally relies on the conditional parallel trends assumption: once units are matched on observed characteristics, the treated and control groups would have followed similar trends in the absence of treatment. 

This approach has been extensively employed as a main empirical strategy across a range of applied settings, dating back to program evaluation contexts in the economics literature in the late twentieth century (\cite{heckman1997matching}; \cite{heckman1998characterizing}), and it has become especially prominent in modern managerial and platform empirical research. For example, \cite{singh2011recruiting} study how firms exploit the prior inventions of new hires and 
employ M-DiD to compare premove versus postmove citation rates for the recruits’ prior patents and corresponding matched-pair control patents; \cite{wang2022monetary} use M-DiD to examine how monetary incentives influence knowledge sharing and spillovers on non-paid activities and other general user engagement on the platform; Using M-DiD, \cite{deshmane2023come} investigate how platform design choices shape user engagement and \cite{gao2025pitfalls} analyze the unintended consequences of review solicitation policies on consumer review sites. In addition, different matching schemes are frequently employed in addition to DiD as robustness checks to strengthen identification credibility, as in \cite{liu2023nudging}, \cite{farronato2024dog}, \cite{kircher2024ban} and many others. Collectively, these studies highlight the importance and wide-spread popularity of M-DiD methods for evaluating causal managerial decisions in complex environments where randomization is rarely feasible.

However, this improvement in comparability through matching comes at a cost. \cite{daw2018matching} and \cite{lindner2019difference} find that when unconditional parallel trends hold, i.e., parallel trends hold without any matching, matching on pre-treatment outcomes can actually induce bias into a perfectly unbiased estimator. Applied researchers thus face a dilemma: while matching can help create a more plausible control group, the act of matching could potentially harm the estimator. Due to this dilemma, previous studies have examined the impact of matching prior to a difference-in-differences analysis and find that matching can either improve or worsen analysis depending on the specific context. The existing literature has primarily concentrated on analyzing the theoretical properties of bias of the estimators (\cite{chabe2015analysis}; \cite{chabe2017should}; \cite{ding2019bracketing}; \cite{ham2024benefits}), and only explored other properties such as variance and mean squared error (MSE) through simulations (\cite{daw2018matching}; \cite{zeldow2021confounding}). The absence of exact mathematical characterizations of variance and MSE of these estimators leaves an incomplete understanding of the overall performance of matching-DiD estimators, especially for inference. Specifically,  with current guidelines focused on bias, applied researchers and business managers are uncertain about whether to match or not if the final objective aims to minimize the width of confidence intervals.

\subsection{Our Contributions}
\noindent To fill the above gap, we build on \cite{ham2024benefits}. Specifically, we analyze three common DiD-based estimators under a linear structural model that allows for time-invariant observed and unobserved confounders with time-varying effects on the outcome (detailed in Section \ref{section:framework}). In this setting, the standard identification assumptions, parallel trends or conditional parallel trends, do not hold exactly. Therefore, the causal estimand of interest, the ATT, is not point identified. In other words, all three estimators, whether matching-based or not, are subject to existence of bias and further increased variance. Our work does not aim to develop a ``better'' identification strategy, if there ever exists one, to eliminate bias or achieve efficiency in estimations, but rather to evaluate which estimator among three produces least bias, variance and mean squared error (MSE) (defined formally in Section~\ref{section: univariate results}), so that applied researchers can select the corresponding matching (or not) strategy.

The existing literature (e.g., \cite{chabe2017should}; \cite{ham2024benefits}) considers only those estimators defined at the population level in terms of limiting expected values. Our first contribution is to introduce finite-sample DiD and matching DiD estimators and formally justify their consistency to the corresponding limiting expected values.  Specifically, the three estimators we analyze are: 1) the classic unmatched DiD estimator, 2) the DiD estimator after matching on observed covariate(s), and 3) the DiD estimator after matching additionally on the pre-treatment outcome(s).

Our second contribution is that we demonstrate that there is a tradeoff between bias and variance when selecting the optimal estimator. In particular, we find there exists a sample-size tradeoff for variance when matching on observed covariate(s) ($X$), i.e., it is not always recommended to match on $X$. On the other hand, \cite{ham2024benefits} show that such matching is always recommended when considering bias. Moreover, we show that matching additionally on pre-treatment outcome(s) always improves variance (compared to just matching on $X$) since such sample-size trade-off no longer applies once matching is performed. Consequently, we recommend incorporating MSE as an additional metric to assess performance of these estimators if applied researchers weigh bias and variance equally. 

When studying MSE, we revisit and generalize previous bias results and relax some assumptions imposed while deriving the same results. We also address an important gap in the current literature and offer a more complete insight for applied researchers when selecting matching estimators in the context of Difference-in-Differences analysis. In particular, our third contribution is that beyond theoretical tradeoffs, we provide more comprehensive estimation strategies and a step-by-step guideline in Section \ref{section: Determining What to Match} to assist applied researchers in making matching decisions when considering bias, variance, and MSE.
 
To complement our theoretical results and to demonstrate the practical value of our guideline, we revisit a recent empirical study of \cite{wang2022monetary} based on Zhihu Live data in Section \ref{section: Application}. In \cite{wang2022monetary}, the authors aim to study the impact of introducing monetary incentives to content creators on Zhihu, a popular online Q\&A platform in China, on platform activity and user engagement, focusing especially on spillover effects of the non-rewarded knowledge activity. To answer this question, the authors employ a M-DiD strategy. We revisit their analysis and answer whether the authors' original decision to match was statistically warranted. Based on our estimation results, we justify that the authors' decision on choosing the M-DiD estimator with matching on both covariates and pre-treatment outcome is warranted and optimal.

The rest of the paper is organized as follows: Section \ref{section:framework} introduces the framework with a linear equation structural model DiD set-up; Section \ref{section: univariate results} aims to offer the main intuitions and takeaways by only characterizing the variance properties of the estimators in a simplified model with univariate variables; Section \ref{section: multivariate results} generalizes the model to include multiple time points and examines the full properties of the generalized estimators; Section \ref{section: Determining What to Match} provides estimation strategies and guidelines for applied researchers on the matching decision process; Section \ref{section: Application} examines our guideline and estimation strategies in an empirical application; Section \ref{section: Discussion} concludes with a discussion.

\section{Framework} \label{section:framework}
\noindent Consider a canonical difference-in-differences framework with a binary treatment and two time periods: a pre-treatment period ($t=T-1=0$) and a post-treatment period ($t=T=1$), where $T$ denotes the final period, in which treatment is first received. We extend our results to multiple pre-treatment periods in Section \ref{section: multivariate results}. 

Let $Z_i$ denote a binary treatment indicator, where $Z_i =1$ indicates individual $i$ is assigned to the treatment group, and $Z_i =0$ indicates the membership in the control group.  Let $n_1 = \sum_{i=1}^{n} Z_i = n \pi$, with  $\pi:=$ the empirical proportion of the $n$ units are treated, and $n_0 = \sum_{i=1}^{n} (1-Z_i) = n (1-\pi)$ so that $n_1 + n_0 = n$. Without loss of generality, we assume $n_1$ and $n_0$ are positive integers and as it is usually the case in practice, we also assume the number of control units ($n_0$) to be greater or equal to the number of treated units ($n_1$); we also assume a balanced panel data setting so that  they stay constant over time. 

We adopt the potential outcome notation where $Y_{i,t}(0)$ and $Y_{i,t}(1)$ denote the untreated and treated potential outcome for unit $i$ at time $t$, respectively. For now, we only have two time points, i.e., $t \in \{0,1\}.$ As usual, the observed outcomes $Y_{i,t} \in \mathbf{R}$ are related to the potential outcomes via the following relationship:
\begin{equation}
Y_{i,t} = Z_i Y_{i,t}(1) + (1 - Z_i) Y_{i,t}(0),     
\end{equation}
in which we adopt the stable unit treatment value assumption \cite{rubin1980randomization}, i.e., no interference between units and no multiple forms of treatment.

For the rest of the paper, we present our main results assuming a linear structural equation model following similar studies
related to DiD and matching in the literature (see \cite{daw2018matching}; \cite{zeldow2021confounding}; 
\cite{ham2024benefits}).  Although our model incorporates the specific parametric form, we believe it nevertheless serves
as a useful starting point for understanding the main tradeoffs at play.

\subsection{A Linear Structural Equation Model}
\noindent Formally, we generate potential outcomes with the following linear model:
\begin{align}
\begin{split}
Y_{i,t}(0) &= \beta_{0,t} + \vec{\beta}^{\top}_{\theta,t} \boldsymbol{\theta}_i + \vec{\beta}^{\top}_{x,t} \boldsymbol{X}_i + \epsilon_{i,t}, \\
Y_{i,t}(1) &= Y_{i,t}(0) + \tau_i \mathbf{1}(t = T),      
\end{split}
\label{eq:LSEM}
\end{align}
\noindent where $\boldsymbol{X_i} \in \mathbf{R}^{p}$ represents individual $i$'s $p$-dimensional observed covariates,  $\boldsymbol{\theta_i} \in \mathbf{R}^{q}$ denotes the $q$-dimensional latent variables, and $\epsilon_{i, t}$ denotes the irreducible mean zero error; $\beta_{0,t}, \vec{\beta}_{x,t}, \vec{\beta}_{\theta,t}$ are fixed constants that denote the intercept at time $t$ and the slopes of $\boldsymbol{X}, \boldsymbol{\theta}$  at time $t$ , respectively. 
 
While both $\boldsymbol{\theta}$ and $\boldsymbol{X}$ have time-varying effects on the outcome captured by the varying subscript `$t$' in $ \vec{\beta}_{\cdot,t}$, $\boldsymbol{\theta}$ and $\boldsymbol{X}$ themselves are time invariant, i.e., they do not change over time. This simplification may appear strong, but it helps to clearly isolate the effects of time-varying factors and group-changing covariates, and therefore captures the breakage in parallel trends. We, however, believe our framing does not lose any ``generality'' to the time-varying covariates context (see discussion in Section \ref{section: Discussion} for a more formal argument). 

Our target estimand is defined to be the average treatment effect on the treated (ATT)\footnote{Under the super-population framework, treatment effects could be modeled as random variables, i.e., $\tau_i \overset{\text{i.i.d.}}{\sim} \mathcal{P}$ with some mean $\tau$ and variance $\sigma_{\tau}^2$. For brevity, we do not present our results in this form, as the main insights remain unchanged (the variance results are the same, aside from additional terms introduced by $\sigma_{\tau}^2$).}:
\begin{equation} \label{PATT estimand}
\text{ATT} := \mathbb{E}\left[Y_{i,T}(1) - Y_{i,T}(0) \mid Z_i = 1\right] ,   
\end{equation}
and we assume no anticipation of treatment effect with \( Y_{i,t}(0) = Y_{i,t}(1) \) for \( t \neq T \), which is ensured by the usage of the indicator \( \mathbf{1}(t = T) \).  

Given the linear structural equation model, we assume a random data generating process (DGP) throughout the rest of paper. 
\begin{assumption}[\textbf{Random Data Generating Process}]\label{Assumption: RDGP}
The joint vector of random variables $\{ \{Y_{i,t}\}_{t=0}^{T}, Z_i, \boldsymbol{\theta}_i, \boldsymbol{X}_i \}_{i=1}^{n}$ is independently and identically distributed (i.i.d.).
\end{assumption} 
\noindent We take a modeling perspective and Assumption \ref{Assumption: RDGP} can be interpreted as the standard random sampling condition of pooled data but with access to latent variables. 
We crucially allow the latent variables $\boldsymbol{\theta}_i$ (as well as the observed covariates $\boldsymbol{X}_i$) to be not independent of the treatment assignment $Z_i$. For instance, we have for different units $i,j$: \[
\left\{ \boldsymbol{\theta}_i : Z_i = 1 \right\} \;\perp\; \left\{ \boldsymbol{\theta}_j : Z_j = 0 \right\}, 
\quad \text{with} \quad 
\boldsymbol{\theta}_i \mid Z_i = 1 \overset{\text{i.i.d.}}{\sim} P_1, 
\quad 
\boldsymbol{\theta}_j \mid Z_j = 0 \overset{\text{i.i.d.}}{\sim} P_0,
\]
where $P_1, P_0$ can be different to allow confounding.

Throughout the paper, we also assume i.i.d. homoscedastic noise terms.  
\begin{assumption}[\textbf{Conditions on Irreducible Errors}] \label{Assumption: errors}  
For all units $i$ at all time points $t$, the irreducible errors $\epsilon_{i,t}  \in \mathbf{R} $ satisfy:                    
\begin{enumerate}
     \item \textbf{Homoscedasticity across time with no serial correlation}: $\epsilon_{i,t}  \perp\!\!\!\perp (\boldsymbol{\theta}_i, \boldsymbol{X}_i, Z_i), \ \text{with} \ \mathbb{E}[\epsilon_{i,t} ] = 0, \ \mathbb{E}[\epsilon_{i,t}^2 ] = \sigma_E^2 > 0, \
\text{and} \ \epsilon_{i,t} \perp\!\!\!\perp \epsilon_{i,t'} \ \text{for all } i$ and all $ t \neq t'.$
    \item \textbf{Existence of fourth moments:} $\mathbb{E}[\epsilon_{i,t}^4] < \infty. $
\end{enumerate}
\end{assumption}
\noindent This assumption is mainly for analytical convenience and it does not preclude the presence of serial correlations in outcomes over time. In particular, we allow for arbitrary temporal dependence through the latent variable $\theta_i$, which implicitly accounts for dependent mechanisms in error terms $\epsilon_{i,t}$. 

The next set of structural assumptions introduce confounding, i.e., we allow $\boldsymbol{\theta}, \boldsymbol{X}$ to have different means based on their treatment groups. We also make regularity assumptions on higher-order moments. 
 
\begin{assumption}[\textbf{Conditions on Observed Covariates and Latent Variables}] \label{Assumption: variables}  
The $\boldsymbol{\theta_i} \in \mathbf{R}^{q}$ and $\boldsymbol{X_i} \in \mathbf{R}^{p}$ satisfy:

\begin{enumerate}
     \item \textbf{Finite Group-specific means}: For $z \in \{0,1\}$, $\mathbb{E}[\boldsymbol{\theta}_i \mid Z_i = z] = \vec{\mu}_{\theta,z} < \infty$ and $\mathbb{E}[\boldsymbol{X}_i \mid Z_i = z] = \vec{\mu}_{x,z} < \infty$ 
    \item \textbf{Positive-definite covariances:}  For $z \in \{0,1\}$, assume $\Var (\boldsymbol{\theta}_i \mid Z_i=z) \triangleq  \Sigma_{\theta\theta} \in \mathbf{R}^{q \times q}$ and $\operatorname{Var}(\boldsymbol{X}_i \mid Z_i=z) \triangleq  \Sigma_{XX} \in \mathbf{R}^{p \times p}$ to be positive-definite, i.e.,  $\Sigma_{\theta\theta} \succ 0, \ \Sigma_{XX} \succ 0.$ Let $\Sigma_{\theta X} \in \mathbf{R}^{q \times p}$ denote the cross-covariance matrix between $\boldsymbol{\theta}_i$ and $\boldsymbol{X}_i$, allowing for arbitrary dependency.
    \item \textbf{Lipschitz continuity of conditional moments:} For $z \in \{0,1\}$, the conditional mean functions
    $m_z(\boldsymbol{x}):= \mathbb{E}[\boldsymbol{\theta}_i \mid \boldsymbol{X}_i = \boldsymbol{x}, Z_i = z]$ and  $m_z(\boldsymbol{x},y) := \mathbb{E}[\boldsymbol{\theta}_i \mid \boldsymbol{X}_i = \boldsymbol{x}, Y_{i,T-1} = y, Z_i = z]$
    are Lipschitz continuous in their arguments; and the conditional variance functions
    $v_z(\boldsymbol{x}):= \Var(\boldsymbol{\theta}_i \mid \boldsymbol{X}_i = \boldsymbol{x}, Z_i = z)$ and $ v_z(\boldsymbol{x},y):= \Var(\boldsymbol{\theta}_i \mid \boldsymbol{X}_i = \boldsymbol{x}, Y_{i,T-1} = y, Z_i = z)$ are also Lipschitz continuous in their arguments.
    \item \textbf{Existence of fourth moments:} $\mathbb{E}[\|\boldsymbol{\theta_i}\|^4 ] < \infty, \quad
    \mathbb{E}[\|\boldsymbol{X_i}\|^4 ] < \infty.$
\end{enumerate}
\end{assumption}

In contrast with \cite{ham2024benefits}, we get rid of the normality assumption of $\boldsymbol{\theta}_i, \boldsymbol{X}_i$ through the Lipschitz continuity assumption while retaining distributional structures, i.e., we also assume group-variant first moments (conditional means) along with group-invariant higher-order moments (thus also conditional covariances). Not indexing the higher-order moments with subscript `$z$' is only for the purpose of simplifying the discussion, and this could be relaxed with additional, albeit messy, algebra that obfuscates the main takeaways of our results. One immediate advantage of our framework is that we enrich the setting to accommodate discrete covariates and thus have broader applications. The Lipschitz continuity conditions are commonly imposed as technical regularity assumptions to ensure the validity of matching procedures (see, e.g., \cite{abadie2006large}; \cite{abadie2012martingale}; \cite{abadie2022robust}). We also remark the conditional Gaussian assumption \citep{ham2024benefits} use instead implies the Lipschitz continuity assumption, thus we accommodate a more general class of conditional distributions of the covariates.

\subsection{Parallel Trends}
\noindent We start this section by revisiting the parallel trends (PT) assumption implied by our model. The PT requires that, in expectation, the change in untreated potential outcomes between the pre-treatment period and post-treatment period is the same across treated and control units. Specifically,
\begin{align}
\label{eq:PT}
\underbrace{\mathbb{E}[Y_{i,T}(0) - Y_{i,T-1}(0) \mid Z_i = 1]}_{\text{Change in treatment group}} &= \underbrace{\mathbb{E}[Y_{i,T}(0) - Y_{i,T-1}(0) \mid Z_i = 0]}_{\text{Change in control group}},
\end{align}
and Equation~\eqref{eq:PT} holds in our linear equation structural model when
\begin{align} \label{eq: specific PT under LSEM}
\vec{\Delta}^{\top}_\theta \vec{\delta}_\theta + \vec{\Delta}^{\top}_X \vec{\delta}_x = 0,
\end{align}
holds, where we define $\vec{\Delta}_{\theta}, \vec{\delta}_{\theta}$ to denote the time varying change in the effect and imbalance of latent variables $\boldsymbol{\theta}$ respectively:
\[
\begin{aligned}
\vec{\Delta}_{\theta} &:= \vec{\beta}_{\theta,T}  -\vec{\beta}_{\theta,T-1} \quad \text{time variation of effect} \\ 
\vec{\delta}_\theta &:= \vec{\mu}_{\theta,1} - \vec{\mu}_{\theta,0} \quad \text{imbalance across treatment groups.}
\end{aligned}
\]
We denote $\vec{\delta}_{x}, \vec{\Delta}_{X}$ similarly for observed covariates $\boldsymbol{X}$. 

We critically allow the slopes of \( \boldsymbol{\theta}, \boldsymbol{X} (\vec{\beta}_\theta, \vec{\beta}_X) \) to differ across time, which allows the latent and observed covariates to have a time-varying effect on the outcome, thus leading to potential breakages in the parallel trends assumption. Under this setting, PT holds as long as there is either time invariance ($\Delta_{.} = 0$) or group invariance ($\delta_{.} = 0$) of both latent and observed covariates (assuming that the effects do not cancel with each other). In other words, covariates with different means within a single time point are not confounders in the DiD framework, so long as they do not evolve differently over time.

\subsection{DiD and Matching DiD Estimators}
\label{section:framework:estimators}
\noindent The existing literature (e.g., \cite{chabe2017should}; \cite{ham2024benefits}) considers only estimators defined at the population level in terms of limiting expected values. In this section, we introduce finite-sample analogues of the DiD and matching DiD estimators. In particular, since our setup is asymptotically equivalent to the counterpart in \cite{ham2024benefits}, we propose finite sample analogues of their relevant estimators and justify that they converge in probability to the corresponding limiting expected values.

We first propose the classic unmatched DiD estimator that one uses in a canonical two time period setting:
\begin{align}  \label{Eq DiD}
\hat{\tau}_{\text{DiD}} &=  \frac{1}{n_1} \sum_{i=1}^{n} \left(Y_{i,T} - {Y}_{i,T-1}\right) Z_i -  \frac{1}{n_0} \sum_{i=1}^{n} \left(Y_{i,T} - {Y}_{i,T-1}\right) (1-Z_i).  
\end{align}
Due to the randomness in the number of treated units $n_1$, to conduct valid asymptotic analysis, we assume a standard condition that the treated group remains a substantial portion of the total sample in the limit. Formally, we have 
\begin{assumption}[\textbf{Boundedness in Treatment Group Fraction}]  \label{Assumption:boundedness}
$$\frac{n_1}{n} = \pi \xrightarrow{p} \mathbb{P} (Z_i =1) :=p\in (0,0.5] \quad  \text{as} \ n_1, n \rightarrow \infty. $$  
\end{assumption} 
 
We then introduce two matching DiD estimators: 1) matching on observed covariates and 2) matching on both observed covariates and pre-treatment outcome(s). For the purpose of isolating discrepancy due to matching as an approach versus its practical implementation, we assume (asymptotically) close-to-perfect matching without replacement.  As sample sizes increase, discrete covariates with a finite number points are perfectly matched, so they can be easily dealt with by conditioning on their values. For continuous covariates, regularity conditions under which such approximations are valid can be found in e.g., \cite{abadie2006large}, \cite{abadie2012martingale} and \cite{abadie2022robust}. Similar to \cite{chabe2017should} and \cite{ham2024benefits}, we also intend to use this assumption to convey an ``idealized'' setting for illustrating the main insights, and therefore we do not explicitly state the associated technical conditions. 

Under this perspective, we introduce the following notation. For each treated unit $i$, let $\mathcal{J}_M(i)$ be the indices of $M$ untreated units, whose covariate values are similar to $\boldsymbol{W}_i$, the variable(s) that one wants to match on. For matching without replacement, the elements of $\{\mathcal{J}_M(i) \ \text{s.t.} \ Z_i = 1\}$ are nonoverlapping subsets of $\{j \in \{1, \ldots, n\} \ \text{s.t.} \ Z_j = 0\}$, chosen to minimize the sum of the matching discrepancies:
\[
\sum_{i=1}^n Z_i \sum_{j \in \mathcal{J}_M(i)} \|\boldsymbol{W}_i - \boldsymbol{W}_j\|,
\]
where $\|\cdot \|$ is some norm in $\mathbb{R}^{\operatorname{dim}(\boldsymbol{W})}$.

In our setting, $\boldsymbol{W}$ could be either the observed covariates $\boldsymbol{X}$ alone or the observed covariates $\boldsymbol{X}$ and the pre-treatment outcome $Y_{T-1}$, shorthanded as $Y_0$. Accordingly, we write $\mathcal{J}_M^{\boldsymbol{X}}(i)$ and $\mathcal{J}_M^{\boldsymbol{X},Y_{0}}(i)$ to represent the corresponding sets of matching indices. Now we propose matching DiD estimators on observed covariates ($\hat{\tau}_{\text{DiD}}^{\boldsymbol{X}}$) and on both observed covariates and the pre-treatment outcome ($\hat{\tau}_{\text{DiD}}^{\boldsymbol{X},Y_{0}}$):
\begin{align} \label{Eq: M-DiD}
\hat{\tau}_{\text{DiD}}^{\boldsymbol{X}} &=  \frac{1}{n_1} \sum_{i=1}^{n} \left( Y_{i,T} - \frac{1}{M} \sum_{j \in \mathcal{J}_M^{\boldsymbol{X}}(i)} Y_{j,T} \right)Z_i - \frac{1}{n_1} \sum_{i=1}^{n} \left( Y_{i,T-1} - \frac{1}{M} \sum_{j \in \mathcal{J}_M^{\boldsymbol{X}}(i)} Y_{j,T-1} \right)Z_i   \\
\hat{\tau}_{\text{DiD}}^{\boldsymbol{X},Y_{0}} &=  \frac{1}{n_1}\sum_{i=1}^{n}  \left( Y_{i,T} - \frac{1}{M} \sum_{j \in \mathcal{J}_M^{\boldsymbol{X},Y_{0}}(i)} Y_{j,T} \right)Z_i.
\end{align}

The idea behind the matching DiD estimator is to fit the DiD estimator according to a matched data. When matching on the pre-treatment outcome in addition to observed covariates, the (asymptotically) close-to-perfect matching assumption implies that only the post-treatment difference remains. As a result, $\hat{\tau}_{\text{DiD}}^{\boldsymbol{X},Y_{0}}$ collapses the panel structure and ignores the second ``difference'' in a DiD estimator. This phenomenon was extensively discussed in Section 4 of \cite{ham2024benefits}, where they show a critical trade-off. We will explore this tradeoff further in Section~\ref{section: univariate results: match on both}.

For the purpose of illustrating the primary insights related to sample size tradeoffs, we further impose a one-to-one matching framework, i.e., $M=1$. Although we acknowledge there are many popular variants of one-to-one matching, we believe the main insights from our paper carry through.
Ultimately, the key consideration is the tradeoff in effective sample size: even if one considers alternative matching schemes such as one-to-many or weighting based matching methods, it may increase the nominal samples (at the cost of match quality) but the effective sample size still decreases, recovering our main insights. Therefore, we present results based on the following matching procedure, denoted by $\mathcal{M}$:
\begin{definition}[\textbf{Matching Algorithm $\mathcal{M}$}] \label{Matching: near-perfect one-to-one matching without replacement}
Let $\mathcal{M}$ denote an (asymptotically) close-to-perfect one-to-one nearest neighbor matching procedure without replacement.    
\end{definition}

We now establish the consistency of the DiD and matching DiD estimators, showing that they converge in probability to the corresponding limiting expected values in \cite{ham2024benefits}.
\begin{theorem}[\textbf{Consistency of DiD and Matching DiD Estimators}] \label{Thm: Consistency of DiD and Matching DiD Estimators}
Assume the data-generating process (DGP) for $\{ \{Y_{i,t}\}_{t=0}^{T}, Z_i, \boldsymbol{\theta}_i, \boldsymbol{X}_i \}_{i=1}^{n}$ follows the linear equation structural model in Equation~\eqref{eq:LSEM}
with $T=1$. Furthermore, suppose this DGP together with a matching procedure $\mathcal{M}$ under Definition \ref{Matching: near-perfect one-to-one matching without replacement}, satisfy Assumptions \ref{Assumption: RDGP}-\ref{Assumption:boundedness}. Then we have
\begin{align*}
\hat{\tau}_{\text{DiD}} \xrightarrow{p} \mathbb{E}\left[\hat{\tau}_{\text{DiD}}\right], \quad \hat{\tau}_{\text{DiD}}^{\boldsymbol{X}}  \xrightarrow{p} \mathbb{E}\left[\hat{\tau}^{\boldsymbol{X}}_{\text{DiD}} \right], \quad  \hat{\tau}_{\text{DiD}}^{\boldsymbol{X},Y_{0}}   \xrightarrow{p} \mathbb{E}\left[\hat{\tau}_{\text{DiD}}^{\boldsymbol{X},Y_0}\right],
\end{align*}
where the limiting expected value of the classic unmatched DiD estimator $\hat{\tau}_{\text{DiD}}$ is: 
\begin{align*}\mathbb{E}\left[\hat{\tau}_{\text{DiD}} \right] 
&=  \mathbb{E}[Y_{i,T} \mid Z_i = 1] - \mathbb{E}[Y_{i,T} \mid Z_i = 0] - \left( \mathbb{E}\left[{Y}_{i,T-1} \mid Z_i = 1\right] - \mathbb{E}\left[{Y}_{i,T-1} \mid Z_i = 0\right] \right), 
\end{align*}
and the limiting expected values of two matching DiD estimators are:
\begin{align*}   
\mathbb{E}\left[\hat{\tau}^{\boldsymbol{X}}_{\text{DiD}} \right] 
&= \mathbb{E}[Y_{i,T} \mid Z_i = 1] - \mathbb{E}\left[{Y}_{i,T-1} \mid Z_i = 1\right] \\
&\quad - \left( \mathbb{E}_{\boldsymbol{x} \mid Z_i = 1} \left[ \mathbb{E}[Y_{i,T} \mid Z_i = 0, \boldsymbol{X}_i = \boldsymbol{x}] \right] 
- \mathbb{E}_{\boldsymbol{x} \mid Z_i = 1} \left[ \mathbb{E}\left[{Y}_{i,T-1} \mid Z_i = 0, \boldsymbol{X}_i = \boldsymbol{x} \right] \right] \right),   \\
\mathbb{E}\left[\hat{\tau}_{\text{DiD}}^{\boldsymbol{X},Y_0}\right]
&= \mathbb{E}[Y_{i,T} \mid Z_i = 1] 
- \mathbb{E}_{\boldsymbol{x}, y \mid Z_i = 1} \left[ \mathbb{E}[Y_{i,T} \mid Z_i = 0, \boldsymbol{X}_i = \boldsymbol{x}, Y_{i,T-1}  = y] \right].
\end{align*}
\end{theorem}
\noindent The proof is provided in Appendix \ref{Appendix A}.  We remark that, given the linear structural equation model in Equation~\eqref{eq:LSEM}, these limiting expected values are generally not equal to the ATT in Equation~\eqref{PATT estimand} because we are in settings where parallel trends or conditional parallel trends do not exactly hold. The resulting limiting double expectation expression for our matching estimators is intuitive: first find a control unit for each treatment unit that shares the  same value of covariates that one wants to match on (e.g. shown by the condition $\mid Z_i =0, \mathbf{X}_i = \mathbf{x}$). Then we must account for the fact this was not any randomely chosen $\mathbf{x}$ but specifically chosen to match the treatment distribution $\mathbf{x} \mid Z_i = 1$ (e.g. shown by the outside expectation $E_{\mathbf{x} \mid Z_i = 1}$).

Due to the randomness in the number of treated units $n_1$, for the rest of the paper, we present bias and variance results ``asymptotically''. We now formally define what we mean by asymptotic bias and variance.
 \begin{lemma}[\textbf{Asymptotic Variance of DiD Estimator}] \label{Lemma: Asymptotic Variance of DiD Estimator}
Assume the data-generating process (DGP) for $\{ \{Y_{i,t}\}_{t=0}^{T}, Z_i, \boldsymbol{\theta}_i, \boldsymbol{X}_i \}_{i=1}^{n}$ follows the linear equation structural model in Equation~\eqref{eq:LSEM}
with $T=1$. Furthermore, suppose this DGP satisfies Assumptions \ref{Assumption: RDGP}-\ref{Assumption:boundedness}. Then we have
\[
\sqrt{n} (\hat{\tau}_{\text{DiD}} - \mathbb{E}\left[\hat{\tau}_{\text{DiD}}\right]) \xrightarrow{d} \mathcal{N}(0,\operatorname{Avar}(\hat{\tau}_{\text{DiD}})),
\]
where $\mathbb{E}\left[\hat{\tau}_{\text{DiD}}\right]$ is defined in Theorem \ref{Thm: Consistency of DiD and Matching DiD Estimators} and 
\begin{align*}
\operatorname{Avar}(\hat{\tau}_{\text{DiD}}) &= \frac{\Var(Y_T(1) \mid Z = 1) + \Var(Y_{T-1}(1) \mid Z =1 ) - 2\,\Cov(Y_T(1), Y_{T-1}(1) \mid Z =1 )}{p} \\
& + 
\frac{\Var(Y_T(0) \mid Z =0 ) + \Var(Y_{T-1}(0) \mid Z =0) - 2\,\Cov(Y_T(0), Y_{T-1}(0) \mid Z =0 )}{1 - p}.
\end{align*}
\end{lemma}
\noindent The proof is provided in Appendix \ref{Appendix: Proofs of Preliminary Lemmas}. Instead of reporting the asymptotic variance $\operatorname{Avar}(\hat{\tau}_{\text{DiD}})$ as in Lemma \ref{Lemma: Asymptotic Variance of DiD Estimator}, we present terms like $\Var \left( \frac{1}{n_1} \sum_{i=1}^{n} Y_{i,T} Z_i \right) \triangleq \Var(Y_T(1) \mid Z = 1)/n_1$, corresponding to its variance in the limit, $\Var(Y_T(1) \mid Z = 1)/p$, over $n$ units. That said, we present variance results with first-order approximations in Theorem \ref{Thm 3.1}, omitting the $o_p(n^{-1})$ remainder term; and bias results analogously without the $o(1)$ remainder term (more details in \ref{Appendix B.First-order Sampling Approximations}). Further, if variances are considered to be evaluated conditional on the realized treatment units, the finite-sample difference between having a random $n_1$ and fixed $n_1$ could also be ignored (see e.g., \cite{wager2020stats} page 4 or \cite{abadie2020sampling} page 269 for the conditional arguments). We proceed similarly for the variance results of matching DiD estimators and defer the exact asymptotic variance results to Lemma \ref{Lemma: Asymptotic Variance of Mathing DiD Estimator(s)} in Appendix \ref{Appendix: Proofs of Preliminary Lemmas}. Throughout the rest of the paper, when we write $\Var(\hat\tau_{DiD})$ we formally mean the asymptotic variance expression given in $\operatorname{Avar}(\hat\tau_{DiD})$ (and respectively for the matching estimators), but scaled with sample sizes terms accordingly as mentioned above.

\section{Variance in Univariate Simplified Case}
\label{section: univariate results}
\noindent In this section, we first present variance results in full generality and then compare estimators in sub-sections. Although we present our framework in Section~\ref{section:framework} with general multi-dimensional $(\boldsymbol{\theta}, \boldsymbol{X})$, in this section we focus on the simplified framework with only one latent variable and one observed covariate to clearly present the main tradeoffs. We extend our results to the more general case and formally show that these trade-offs are preserved in Section~\ref{section: multivariate results}.

As we now focus exclusively on the univariate results for this section, we begin with a remark to simplify notation.
\begin{remark} \label{Remark: notations on univariate variance}
For the univariate case ($X_i, \theta_i \in \mathbf{R}$), we adopt the following notational conventions for the elements of the covariance matrix notations in Assumption \ref{Assumption: variables}:
$$\Sigma_{\theta \theta} = \sigma_{\theta}^2, \quad \Sigma_{\theta X} = \rho \sigma_{\theta} \sigma_{x}, \quad \Sigma_{X X} = \sigma_{x}^2.$$
\end{remark}

To parallel our results with \cite{ham2024benefits}, for the rest of the paper, we assume the following specific structures on Lipschitz continuous conditional moments in Assumption \ref{Assumption: variables}. 
\begin{remark} \label{Remark: Structures on Lipschitz continuous functions}
Assume the data-generating process (DGP) for $\{ \{Y_{i,t}\}_{t=0}^{T}, Z_i, \boldsymbol{\theta}_i, \boldsymbol{X}_i \}_{i=1}^{n}$ following the linear equation structural model in Equation~\eqref{eq:LSEM}
with $T=1$ and $X_i, \theta_i \in \mathbf{R}$. We further impose for $z \in \{0,1\}$,
\begin{align*}
\mathbb{E}\left[\theta_i \mid X_i=x, Z_i=z \right]  &=  \mu_{\theta,z} + \rho \frac{\sigma_\theta}{\sigma_x}(x - \mu_{x,z}) \\
\mathbb{E}\left[\theta_i \mid X_i=x, Y_{i,0} = y, Z_i=z \right]  &=  \mu_{\theta,z} + \frac{m_1}{\text{det}}(x - \mu_{x,z}) + \frac{m_2}{\text{det}}(y - \beta_{\theta,0} \mu_{\theta,z} - \beta_{x,0} \mu_{x,z}) 
\end{align*}
and for $z =0$, i.e., for the control group, 
\begin{align*}
\Var\left( \theta_i \mid X_i=x, Z_i=0 \right)  &= (1 - \rho^2)\sigma_\theta^2 \\
\Var\left( \theta_i \mid X_i=x, Y_{i,0} = y, Z_i=0 \right)  &= \sigma_{\theta}^2 -\frac{m_1}{\text{det}} \left(\sigma_{\theta} \sigma_{x} \rho \right) - \frac{m_2}{\text{det}} \left(\beta_{\theta,0} \sigma_{\theta}^2 + \beta_{x,0} \sigma_{\theta} \sigma_{x} \rho  \right),
\end{align*}
where we define the following to avoid long algebraic expressions:
\begin{align*}
m_1 &= -\beta_{x,0}\beta_{\theta,0}\sigma_{x}^2\sigma_{\theta} ^2(1-\rho^2) + \rho\sigma_{\theta} \sigma_{x}\sigma_E^2, \ m_2= \beta_{\theta,0}\sigma_{\theta}^2\sigma_{x}^2(1-\rho^2), \ \text{det} = \sigma_{x}^2(\beta_{\theta,0}^2\sigma_{\theta}^2(1-\rho^2) + \sigma_E^2).
\end{align*}
\end{remark}
\noindent We acknowledge that this is only one instance of the Lipschitz continuity conditions in Assumption \ref{Assumption: variables} with pre-specified Lipschitz constants (e.g., $L = |\rho \frac{\sigma_\theta}{\sigma_x}|$ for $\mathbb{E}\left[\theta_i \mid X_i=x, Z_i=0 \right]  =  \mu_{\theta,0} + \rho \frac{\sigma_\theta}{\sigma_x}(x - \mu_{x,0})$), but the condition remains broadly applicable and is not as restrictive as it may initially appear. 
In fact, our primary assumptions only need to be on specifications of the Lipschitz constants and the remaining constants serve only for normalization. Moreover, the exact normality DGP assumed in \cite{ham2024benefits} implies our conditions (see proofs in Appendix \ref{Appendix: Reparametrization of Lipschitz continuous moments}), but not vice versa, i.e., ours could hold in other scenarios. In this sense, we view these ``pre-specified'' Lipschitz constants reflecting convenient reparameterization, rather than just being restrictive. Moreover, the resulting Lipschitz constants retain clear economic and statistical interpretations (e.g., in terms of correlations and scale parameters).

Given Remark \ref{Remark: notations on univariate variance} and Remark \ref{Remark: Structures on Lipschitz continuous functions}, we are now ready to present our main results.
\begin{theorem}[\textbf{Variance of DiD and Matching DiD Estimators}] 
\label{theom:var_2x2}
Assume the data-generating process (DGP) for $\{ \{Y_{i,t}\}_{t=0}^{T}, Z_i, \boldsymbol{\theta}_i, \boldsymbol{X}_i \}_{i=1}^{n}$ follows the linear equation structural model in Equation~\eqref{eq:LSEM}
with $T=1$ and $X_i, \theta_i \in \mathbf{R}$. Furthermore, suppose this DGP together with a matching procedure $\mathcal{M}$ under Definition \ref{Matching: near-perfect one-to-one matching without replacement}, satisfy Assumptions \ref{Assumption: RDGP}-\ref{Assumption:boundedness}. Then the variances of our estimators are given by: 
\[
\begin{aligned}
\Var \left( \hat{\tau}_{\text{DiD}} \right) &= \left(\frac{1}{n_1} + \frac{1}{n_0} \right) \left\{  2\sigma_{E}^2 + \Delta_{\theta}^2  \sigma_{\theta}^2 +  \Delta_{X}^2 \sigma_{x}^2 + 2 \Delta_{\theta}  \Delta_{X} \rho  \sigma_{x}\sigma_{\theta}  \right\} \\
\Var \left( \hat{\tau}_{\text{DiD}^{X}} \right)  &= \left(\frac{1}{n_1} + \frac{1}{n_1} \right)  \left\{ 2\sigma_{E}^2 + 
\Delta_{\theta}^2 (1-\rho^2) \sigma^2_{\theta}    \right\}  \\
\Var \left( \hat{\tau}_{\text{DiD}^{X,Y_0}} \right) &= \left(\frac{1}{n_1} + \frac{1}{n_1} \right) \left\{ \sigma_{E}^2  + \beta_{\theta,1}^2 \sigma_{\theta}^2 \left(1-\rho^2 \right)   \left( 1 - r_{\theta|x} \right)  \right\},   \label{Thm 3.1} 
\end{aligned}
\] 
where
\[
r_{\theta|x} := 1 - \frac{\Var(Y_{i,0} \mid Z_i = 0, X_i, \theta_i)}{\Var(Y_{i,0} \mid Z_i = 0, X_i)} = \frac{(\beta_{\theta,0})^2 \sigma_\theta^2 (1 - \rho^2)}{(\beta_{\theta,0})^2 \sigma_\theta^2 (1 - \rho^2) + \sigma_E^2}.
\]
\end{theorem}
\noindent The proof is provided in Appendix \ref{Appendix B.Main results}. Before comparing when matching is better or not, we make some notational remarks. We refer to $r_{\theta|x}$ as the reliability term throughout the paper, which measures how much variance $\theta$ accounts for $Y$ after controlling for $X$ (\cite{trochim2006types}; \cite{ham2024benefits}). More formally stated:
\begin{definition}[\textbf{Reliability}] \label{Def: reliabiltiy}
The reliability \( r_\theta \) of a random variable \( Y \) as a measure of a random variable \( \theta \) is
\[
0 \leq r_\theta = 1 - \frac{\Var(Y \mid \theta)}{\Var(Y)} \leq 1.
\]
\end{definition}

In particular, we use $r_{\theta|x}$, the (conditional) reliability of the pre-treatment outcome with respect
to $\theta$ within the control group after controlling for $X$. In our linear framework, the reliability term can be interpreted as the population R-squared statistic if we were able to regress the pre-treatment outcome on the latent variable after accounting for observed covariates within the control group.

\subsection{Matching on \texorpdfstring{$X$}{X} - Uncorrelated Case}
\noindent In this section, we explore when it is better to match just on $X$ compared to no matching when considering variance, i.e., when $\Var \left( \hat{\tau}_{\text{DiD}} \right) \geq \Var \left( \hat{\tau}_{\text{DiD}^X} \right)$.  We first explore this under the setting when parallel trends (PT) exactly holds, particularly when coefficients are time invariant $(\Delta_\theta =\Delta_X = 0)$. We remind readers that the unmatched classic DiD is justified and unbiased when PT holds, thus intuitively the variance should also be the lowest when PT holds. We indeed see that both the variances of $\Var \left( \hat{\tau}_{\text{DiD}} \right)$ and $\Var \left( \hat{\tau}_{\text{DiD}^X} \right)$ are minimized when PT holds since all terms except $\sigma_E^2$ are zero when $\Delta_\theta =\Delta_X = 0$. Since $\sigma_E^2$ is coming from the irreducible errors, this is the minimal attainable variance with finite samples. Moreover, the classic unmatched DiD estimator always exhibits lower variance than $\Var \left( \hat{\tau}_{\text{DiD}^X} \right)$ when PT holds, where $\Var \left( \hat{\tau}_{\text{DiD}} \right) - \Var \left( \hat{\tau}_{\text{DiD}^X} \right) = 2\sigma_E^2 / n_0 - 2\sigma_E^2 / n_1 \leq 0$. This is because it utilizes the full sample size of control units, thereby incorporating more information, while not necessarily ``losing'' anything due to PT exactly holding. This underscores two important facts if PT is known to hold: 1) no additional matching procedure should be considered even when considering variance, and 2) matching strictly ``harms'' the DiD estimator by utilizing fewer samples as matching forcibly removes the undesirable samples of the control group to resemble the treatment group. This is the core of the ``sample-size tradeoff'', which we explore further next. 

The question remains what happens when parallel trends is not exactly satisfied, i.e., $\Delta_\theta =\Delta_X \neq 0$. Since there are multiple sources of variations, we will first focus on the case when $\theta$ and $X$ are uncorrelated ($\rho = 0$) and then explain how the correlation changes the results. When $\rho = 0$, there are four main sources of variations that arises in the variance expression for the classic unmatched DiD estimator in Theorem~\ref{Thm 3.1}. In particular, we label the sources as follows: 
\begin{equation}
\label{eq:source_var_did}
  \Var \left( \hat{\tau}_{\text{DiD}}^{\rho=0} \right) = \underbrace{\left(\frac{1}{n_1} + \frac{1}{n_0} \right)}_{(1): \text{sampling variation}} \left\{  \underbrace{2\sigma_{E}^2}_{(2): \text{noise variation}} + \underbrace{\Delta_{\theta}^2 \sigma_{\theta}^2}_{(3): \text{latent variable variation}} +  \underbrace{\Delta_{X}^2 \sigma_{x}^2}_{(4): \text{observed covariate variation}}  \right\},  
\end{equation}
where (1) the sampling variation captures the trade-off in sample sizes as the matching estimator(s) utilizes fewer units $({n_1}^{-1} + {n_1}^{-1})$ than full sample sizes $({n_1}^{-1} + {n_0}^{-1})$; (2) the noise variation captures the variation coming from irreducible errors; (3) the latent variable variation captures the variation coming from imbalance in $\theta$: our outcome depends on $\theta$ and $\theta$ has its variation, where the variation scales by the magnitude of breakage in PT ($\Delta_\theta$); (4) the observed covariate variation captures the variation coming from imbalance in $X$ similarly as (3).

To facilitate comparison, we also show the variance breakdown for the DiD estimator that matches on $X$ when $\rho=0$:
\begin{equation}
\label{eq:source_var_didx}
  \Var \left( \hat{\tau}_{\text{DiD}^{X}}^{\rho=0} \right) = \underbrace{\left(\frac{1}{n_1} + \frac{1}{n_1} \right)}_{(1): \text{sampling variation}} \left\{  \underbrace{2\sigma_{E}^2}_{(2): \text{noise variation}} + \underbrace{\Delta_{\theta}^2 \sigma_{\theta}^2}_{(3): \text{latent variable variation}} \right\}, 
\end{equation}

Given the labeled sources of variations in Equation~\eqref{eq:source_var_did} and Equation~\eqref{eq:source_var_didx}, we see the 
sample-size tradeoff still carries over as shown by the difference in the sampling variation (1). We assume more control units $n_0 \geq n_1$, thus matching always increases the sampling variation. This is the sample-size ``harm'' of matching. The ``benefit'' of matching on $X$ completely removes the variation contributed by $X$ (variation source (4)) as shown by the lack of the term $\Delta^2_X\sigma_x^2$ in Equation~\eqref{eq:source_var_didx}.

Due to the aforementioned ``harm'' and ``benefit'' of matching on $X$, we state the following Lemma to formally characterize this trade-off:

\begin{lemma}[\textbf{Sufficient and necessary condition to match on $X$ in the uncorrelated case}] \label{Lemma 3.1} Assuming the same conditions in Theorem \ref{Thm 3.1}, but with $\rho =0$, $\Var \left( \hat{\tau}_{\text{DiD}} \right) \geq \Var \left( \hat{\tau}_{\text{DiD}^{X}} \right)$ if and only if  
$$\underbrace{\Delta_{X}^2 \sigma_{x}^2}_{(4)} \underbrace{\left(\frac{1}{n_0} +\frac{1}{n_1} \right)}_{(1)}  \geq \underbrace{\Delta_{\theta}^2  \sigma_{\theta}^2}_{(3)} \underbrace{\left(\frac{1}{n_1} - \frac{1}{n_0} \right)}_{(1)}    + \  \underbrace{2\sigma_{E}^2}_{(2)} \underbrace{\left(\frac{1}{n_1} - \frac{1}{n_0} \right)}_{(1)},$$    
\end{lemma}
\noindent Lemma~\ref{Lemma 3.1} mathematically captures the above intuition. The left-hand side (LHS) of the inequality can be interpreted as the benefits obtained from matching while the right-hand side (RHS) represents the sample-size related costs we hope to outweigh. For example, the LHS becomes larger, thus preferring matching, when the variation source (4) coming from $X$ ($\Delta_x \sigma_x$) is large, showing the ``benefit'' of matching on $X$ by eliminating variation source (4). On the other hand, the RHS also becomes larger, thus not preferring matching, when $n_0 >> n_1$ so that $(n_1^{-1} - n_0^{-1})$ increases, i.e., when the sample sizes between $n_1$ and $n_0$ are largely different it is not preferred to match.

The aforementioned rationale explains how variation source (4) and (1) are relevant. However, Lemma~\ref{Lemma 3.1} also shows that sources (3) and (2) have impacts on the RHS through variation source (1). This is because sources (1)-(4) do not contribute equally to the final variance as shown by Equation~\eqref{eq:source_var_did}-\eqref{eq:source_var_didx}. In particular, source (1) has a multiplicative impact, i.e., increasing its variation inadvertently increases variations in all other sources. Consequently, Lemma~\ref{Lemma 3.1} shows the contribution of variation source (3) and (2) through the sample-size tradeoffs. 

To summarize, on one hand, matching on $X$ pays a variance cost through the sample-size in source (1), but instead improves variance by reducing the variation coming from $X$, i.e., source (4). On the other hand, variations from source (2) and source (3) do not get directly worse, but, nevertheless, still contribute to Lemma~\ref{Lemma 3.1} through the positive scaling from source (1).

\subsection{Matching on \texorpdfstring{$X$}{X} - Correlated Case}
\noindent In this section, we generalize the above findings to the correlated case, where in our setting, $\rho$ captures the correlation between $\theta$ and $X$. Consequently, one should expect the benefit of matching on $X$ to increase when the correlation is higher, since high correlation implies that one is not only matching on $X$, but also indirectly matching and recovering variations from the unobserved $\theta$ through the correlation.

To formalize this, we first note an additional variance source (5) through the correlation $\rho$ ($\Delta_{\theta}\Delta_{X}\rho \sigma_{\theta} \sigma_{x}$) in $\text{Var} \left( \hat{\tau}_{\text{DiD}}\right)$. On one hand, matching on $X$ eliminates the variation from source (5), and reduces the source (3) variation by an amount proportional to the correlation factor $1-\rho^2$.  In particular, in the perfectly correlated case ($|\rho|= 1$), we only have the irreducible errors remaining in the variance expression: 
$$\Var \left( \hat{\tau}_{\text{DiD}^{X}}^{|\rho|= 1} \right)  = \left(\frac{1}{n_1} + \frac{1}{n_1} \right)  \left\{ 2\sigma_{E}^2 + 
\Delta_{\theta}^2 (\underbrace{1-\rho^2}_{=0}) \sigma^2_{\theta}    \right\} = \left(\frac{1}{n_1} + \frac{1}{n_1} \right) \left(2\sigma_{E}^2 \right) $$
This result verifies and proves the aforementioned intuition that matching on $X$, under perfect correlation, matches indirectly on 
the latent variable $\theta$. In this case, matching on $X$ fully recovers $\theta$, thus eliminating all variations from the imbalance in the latent variable as well. Although this might seem that the correlation only helps the matching estimator, i.e., a non-zero $\rho$ reduces $\Var( \hat{\tau}_{\text{DiD}^{X}})$, a non-zero $\rho$ could, in theory, also help reduce the variance of the classical unmatched DiD estimator through the final term in $2\Delta_\theta \Delta_X \rho $. In other words, there is a mathematical possibility  that $\Delta_\theta \Delta_X \rho < 0$, i.e., the signs of the time-varying effects of  $\theta$ and $X$ ($\Delta_{\theta}, \Delta_x$) move in opposing directions relative to their correlation. In this case, the unmatched DiD estimator gets a further reduction in variance than the matching on $X$ DiD estimator. We, however, believe this is more of a mathematical technicality and not likely to occur in practice. For example, if $\theta$ and $X$ are positively correlated, then there is a higher chance that they have similar underlying effects, so it may be reasonable to believe the effects would evolve in the same direction over time as opposed to suddenly changing signs to allow this cancellation to occur.

In summary, although there is a mathematical possibility that the correlation ``cancels'' out some variation, making the naive DiD estimator better, we believe that matching on a correlated $X$ generally improves the matching estimator as it also recovers information about $\theta$. This conclusion was also echoed in Section 4 of \cite{ham2024benefits}. For the sake of completeness, we still provide a full condition in the presence of correlation with respect to when matching on $X$ is better in the following Lemma.
\begin{lemma}[\textbf{Sufficient and necessary condition to match on $X$ in the general correlated case}] \label{Lemma 3.2} Assuming the same conditions in Theorem \ref{Thm 3.1}, 
$\Var \left( \hat{\tau}_{\text{DiD}} \right) \geq \Var \left( \hat{\tau}_{\text{DiD}^{X}} \right)$ if and only if 
$$\left(\underbrace{\Delta_{X}^2 \sigma_{x}^2}_{(4)}  + \underbrace{2\Delta_{\theta}\Delta_{X}\rho \sigma_{\theta} \sigma_{x}}_{(5)} \right)   \underbrace{\left(\frac{1}{n_0} +\frac{1}{n_1} \right)}_{(1)}  \geq \underbrace{\Delta_{\theta}^2  \sigma_{\theta}^2}_{(3)} \underbrace{\left(\frac{1 - 2 \rho^2}{n_1} - \frac{1}{n_0} \right)}_{(1)}    + \  \underbrace{2\sigma_{E}^2}_{(2)} \underbrace{\left(\frac{1}{n_1} - \frac{1}{n_0} \right)}_{(1)}$$  
\end{lemma}
\noindent The proof comes from an algebraic simplification of the main theorem. Again, the left-hand side (LHS) of the inequality can be interpreted as the benefit achieved through matching on observed covariates, while the right-hand side (RHS) represents the associated cost that we hope to outweigh. As mentioned above, in general, the presence of a non-zero correlation leads to greater benefits (additional variation source (5) shown in LHS) and lower costs (reduction of variance from also recovering $\theta$, shown as the first source (1) in RHS diminishing).

In summary, this section highlights that applied researchers should be cautious with whether to match or not, even only on observed covariates, when considering variance. Unlike the bias results presented in \cite{ham2024benefits}, which shows one should always match on $X$ when considering bias, our results show a more nuanced narrative through the sample-size tradeoff not captured by bias consideration alone. Because of this tradeoff, applied researchers may now want to know whether they should indeed match on $X$ when considering variance or MSE. We thus provide guidelines and consistent estimation strategies for variance to answer this question in Section \ref{section: Determining What to Match}.

\subsection{Matching on \texorpdfstring{$X,Y_0$}{X,Y0}} \label{section: univariate results: match on both}
\noindent In this section, we analyze further when it is better to additionally match on $Y_0$ after matching on $X$. We begin by stating a surprising result:
\begin{lemma}[\textbf{Matching on $X,Y_0$ always helps}] \label{Lemma 3.3}
Assuming the same conditions in Theorem \ref{Thm 3.1}, additionally matching on pre-treatment outcomes always reduces variance relative to matching only on observed covariates, i.e., $\Var \left( \hat{\tau}_{\text{DiD}^{X}} \right) \geq \Var\left( \hat{\tau}_{\text{DiD}^{X,Y_0}} \right)$ always holds.  
\end{lemma}
\noindent The proof is provided in Appendix \ref{Appendix B.2}. This result states that regardless of the degree of informativeness of the reliability term, matching additionally on the pre-treatment outcome always reduces the variance relative to matching only on observed covariates. In other words, when considering matching additionally on $Y_0$ there is no tradeoff. As a reminder, the pre-treatment outcome contains information about $\theta$, thus matching on the pre-treatment outcome ``imperfectly'' matches $\theta$. 

We acknowledge that this result may, to some extent, be attributed to our specific matching algorithm $\mathcal{M}$ used in Definition \ref{Matching: near-perfect one-to-one matching without replacement}. However, the general insights behind this result remain clear. Specifically, Theorem~\ref{Thm 3.1} shows that many sources of variation are similar when comparing $\Var \left( \hat{\tau}_{\text{DiD}^{X}} \right)$ and $\Var\left( \hat{\tau}_{\text{DiD}^{X, Y_0}} \right)$. For example, there is no sample-size tradeoff since one has already conducted matching and the effective sample size for both estimators are $n_1$ for the treatment group and also $n_1$ for the new matched control group. Furthermore, both estimators do not have variation source (4) from $X$ as $X$ is matched on. Lastly, both estimators are similarly affected by the correlation through the $(1- \rho^2)$ term. 

The main difference comes from (a) how (perfect) matching on $Y_0$ gets rid of the second-difference in DiD and (b) how matching on $Y_0$ indirectly matches on $\theta$ as captured by the new reliability term $r_{\theta \mid x}$ in the expression of $\Var \left( \hat{\tau}_{\text{DiD}^{X, Y_0}} \right)$. We note that matching on $Y_0$ removes any baseline differences of $Y_0$ between the treatment and control, thus the DiD estimator reduces to a difference-in-means estimator. This realization was also explored thoroughly in \cite{ham2024benefits}. The main consequence of difference (a) is how $\Var \left( \hat{\tau}_{\text{DiD}^{X}} \right)$ scales with both pre and post-period slopes captured by $\Delta_{\theta}$ while $\text{Var} \left( \hat{\tau}_{\text{DiD}^{X, Y_0}} \right)$ scales only with the post slope $\beta_{\theta, 1}$. This difference also explained why there was a trade-off in Section 4 of \cite{ham2024benefits} since even if parallel trends holds exactly, i.e., $\Delta_{\theta} = 0$, the second term for $\Var \left( \hat{\tau}_{\text{DiD}^{X}} \right)$ is zero and minimized but the second term for $\text{Var} \left( \hat{\tau}_{\text{DiD}^{X, Y_0}} \right)$ is non-zero and could be potentially still large. 

This ``harm'' was offset by the introduction of the reliability term $(1- r_{\theta \mid x})$, which captures the aforementioned intuition of how matching on $Y_0$ can reduce the variation coming from $\theta$ as it indirectly matches on $\theta$. More formally, as the reliability, $0 < r_{\theta \mid x} <1$, approaches 1, i.e., $\theta$ perfectly explains $Y_0$, $\text{Var} \left( \hat{\tau}_{\text{DiD}^{X, Y_0}} \right)$ is minimized to the irreducible errors. 

Given these two main differences, one may expect a similar tradeoff where sometimes it is better to match on $Y_0$ depending on whether the reliability is high and sometimes it is  better to not additionally match on $Y_0$ if $\Delta_{\theta}$ is much smaller than $\beta_{\theta, 1}$, such as when parallel trends holds. However, such a tradeoff does not exist in the variance case because of the variation source (2) $\sigma_E^2$. Difference (a), i.e., how matching on $Y_0$ erodes the second difference in DiD, also reduces one less ``difference'' to contribute the irreducible noise. As a result, there is one less $\sigma_E^2$ term in the expression of $\Var \left( \hat{\tau}_{\text{DiD}^{X, Y_0}} \right)$ compared to  $\text{Var} \left( \hat{\tau}_{\text{DiD}^{X}} \right)$. This difference allows us to prove that $\text{Var} \left( \hat{\tau}_{\text{DiD}^{X}} \right)$ is always larger than $\Var \left( \hat{\tau}_{\text{DiD}^{X, Y_0}} \right)$. Although not immediately obvious from the expressions in Theorem~\ref{Thm 3.1}, the reliability, $r_{\theta \mid x} = ((\beta_{\theta,0})^2 \sigma_\theta^2 (1 - \rho^2))/((\beta_{\theta,0})^2 \sigma_\theta^2 (1 - \rho^2) + \sigma_E^2)$, is a function of both $\sigma_E^2$ and $\beta_{\theta, 0}$, and thus one can show the relationship where the benefits from the reliability and one less $\sigma_E^2$ is sufficient to outweigh any cost from eroding the second difference in the DiD estimator. Lemma~\ref{Lemma 3.3} illustrates and proves this idea.

\section{Full Bias-Variance Results in General Settings} \label{section: multivariate results}
\noindent While our analysis so far has focused on the canonical difference-in-differences set-up with a single pre-treatment period and univariate observed covariate $X$ and latent variable $\theta$, the results and insights extend similarly to the general framework described in Section \ref{section:framework} with multivariate ($\boldsymbol{X}, \boldsymbol{\theta}$). Furthermore, in practice, it is common to observe multiple pre-treatment time points, where applied researchers may want to match on multiple pre-treatment outcomes and multiple covariates (\cite{bartanen2019impacts, ham2024benefits}). We therefore generalize the previous setting with multiple time periods, namely $T > 1$, in which we assumed $T = 1$ in all previous sections. 
 
When there are multiple time points, it is popular to use a linear regression framework to estimate ATT, especially the following two-way fixed effects regression (see e.g., chapter 5.2 in \cite{angrist2009mostly}): 
\begin{equation}
Y_{i,t} = \alpha_i + \lambda_t + \beta_{ZW} Z_i W_t + \epsilon_{i,t},
\end{equation}
where $W_t$ is an indicator that is 1 if in the post-treatment period ($t = T$) and zero otherwise, $\alpha_i, \lambda_t$ are fixed effects for unit and time, respectively, and $\epsilon_{i,t}$ is the residual error. The estimate $\hat{\beta}_{ZW}$ for the interaction term would then be taken as the DiD estimate of the ATT. It is well known that in a balanced panel data with multiple pre-treatment outcomes, $\hat{\beta}_{ZW}$ is equivalent to the estimate one would obtain using a classical two-period DiD using the average of all the $T$ pre-treatment outcomes, $\overline{Y}_{i,0:(T-1)} = \sum_{t=0}^{T-1} Y_{i,t} / T$, for the pre-treatment outcome \cite{wooldridge2021two}. Throughout this section, we leverage this balancing property to define and analyze the generalized estimators with multivariate covariates and time points. Formally, we have 
\[
\begin{aligned}
\vec{\Delta}_{\theta} &= \vec{\beta}_{\theta,T} - \frac{\sum_{t=0}^{T-1} \vec{\beta}_{\theta,t}}{T} \quad \text{time variation of effect of $\theta$} \\ 
\vec{\Delta}_{X} &= \vec{\beta}_{x,T} - \frac{\sum_{t=0}^{T-1} \vec{\beta}_{x,t}}{T} \quad \text{time variation of effect of $X$,}
\end{aligned}
\]
where the original parallel trends can still be represented in a vectorized form in Equation \eqref{eq: specific PT under LSEM}.

Since our setup is still asymptotically equivalent to the counterpart in \cite{ham2024benefits} under this framework,
we propose finite sample analogs of their relevant generalized DiD and generalized matching DiD estimators, where they match on all available pre-treatment outcomes and covariates\footnote{We choose to do such matching because our model does not explicitly index covariates by time and we implicitly avoid potential bias issues from matching on post-treatment variables. For applied researchers who are interested in not matching on all available covariates, see, e.g., \cite{rosenbaum1984consequences}, \cite{stuart2010matching} for further discussion.}. Specifically, given the model and the balanced panel setting, we introduce the following:
\begin{align}
\hat{\tau}_{\text{gDiD}} &=  \frac{1}{n_1} \sum_{i=1}^{n} \left(Y_{i,T} - \bar{Y}_{i,0:(T-1)}\right) Z_i -  \frac{1}{n_0} \sum_{i=1}^{n} \left(Y_{i,T} - \bar{Y}_{i,0:(T-1)}\right) (1-Z_i)  \\
\hat{\tau}_{\text{gDiD}}^{\boldsymbol{X}} &=  \frac{1}{n_1} \sum_{i=1}^{n} \left( Y_{i,T} - \frac{1}{M} \sum_{j \in \mathcal{J}_M^{\boldsymbol{X}}(i)} Y_{j,T} \right)Z_i -  \frac{1}{n_1} \sum_{i=1}^{n} \left( \bar{Y}_{i,0:(T-1)} - \frac{1}{M} \sum_{j \in \mathcal{J}_M^{\boldsymbol{X}}(i)} \bar{Y}_{j,0:(T-1)} \right)Z_i  \\
\hat{\tau}_{\text{gDiD}}^{\boldsymbol{X},\boldsymbol{Y}^{\mathbf{T}}} &= \frac{1}{n_1}\sum_{i=1}^{n}  \left( Y_{i,T} - \frac{1}{M} \sum_{j \in \mathcal{J}_M^{\boldsymbol{X},\boldsymbol{Y}^{\mathbf{T}}}(i)} Y_{j,T} \right)Z_i,
\end{align}
where $\bar{Y}_{i,0:(T-1)}$ is the average of all the $T$ pre-treatment outcomes, i.e.,
\[
\bar{Y}_{i,0:(T-1)} =
\frac{\sum_{t=0}^{T-1} \beta_{0,t}}{T}
+ \frac{\sum_{t=0}^{T-1} \vec{\beta}_{\theta,t}^{\top} \boldsymbol{\theta}_i}{T}
+ \frac{\sum_{t=0}^{T-1} \vec{\beta}_{x,t}^{\top} \boldsymbol{X_i}}{T}
+ \frac{\sum_{t=0}^{T-1} \epsilon_{i,t}}{T}.
\]
We still only assume one post-treatment period at the final period $T$, i.e., units receive treatment only at $t = T$ and not before $t < T$. Therefore, we keep the notation $Z_i$ as the binary treatment indicator without a subscript $t$ as $Z_{i, t}$.

\subsection{Generalized Variance Results}
\begin{theorem}[\textbf{Variance of generalized DiD and generalized Matching DiD Estimators}] \label{Thm 4.1} 
Assume the data-generating process (DGP) for $\{ \{Y_{i,t}\}_{t=0}^{T}, Z_i, \boldsymbol{\theta}_i, \boldsymbol{X}_i \}_{i=1}^{n}$ follows the linear equation structural model in Equation~\eqref{eq:LSEM}. Furthermore, this DGP together with a matching procedure $\mathcal{M}$ under Definition \ref{Matching: near-perfect one-to-one matching without replacement}, satisfy Assumptions \ref{Assumption: RDGP}-\ref{Assumption:boundedness}. Then the variances of our estimators are given by: 
\[
\begin{aligned}
\Var \left( \hat{\tau}_{\text{gDiD}} \right) 
&=  \left(\frac{1}{n_1} + \frac{1}{n_0}\right) \left\{  \frac{T+1}{T} \sigma_{E}^2 +  \vec{\Delta}_{\theta}^{\top} \Sigma_{{\theta} {\theta}} \vec{\Delta}_{\theta}  + \vec{\Delta}_{X}^{\top} \Sigma_{X X} \vec{\Delta}_{X}  + 2 \vec{\Delta}_{\theta}^{\top}   \Sigma_{\theta X} \vec{\Delta}_{X} \right\}   \\
\Var  \left( \hat{\tau}_{\text{gDiD}}^{\boldsymbol{X}} \right) &=  \left(\frac{1}{n_1} + \frac{1}{n_1}\right) \left\{  \frac{T+1}{T} \sigma_{E}^2 + \vec{\Delta}_{\theta}^{\top} \Sigma_{\tilde{\theta} \tilde{\theta}} \vec{\Delta}_{\theta} \right\}  \\ 
\Var \left( \hat{\tau}_{\text{gDiD}}^{\boldsymbol{X},\boldsymbol{Y}^{\top}}\right) 
& =  \left( \frac{1}{n_1} +   \frac{1}{n_1} \right)   \left\{ \vec{\beta}_{\theta,T}^{\top}  \left( \mathbf{I}_{q \times q}  - \mathbf{r}_{\theta|x} \right)  \Sigma_{\tilde{\theta} \tilde{\theta}}  \vec{\beta}_{\theta,T}  + \sigma^2_E \right \}
\end{aligned}
\]
where we denote $\Sigma_{\tilde{\theta} \tilde{\theta}} \triangleq \Sigma_{\theta \theta} - \Sigma_{\theta X} \Sigma_{X X}^{-1} \Sigma_{X \theta}$; $\mathbf{r}_{\theta|x}  \triangleq \Sigma_{\tilde{\theta} \tilde{\theta}} B_{\theta}^{\top} (B_\theta \Sigma_{\tilde{\theta} \tilde{\theta}} B_\theta^{\top}+ \Sigma_\epsilon)^{-1} B_\theta    \in \mathbf{R}^{q \times q}$, 
and 
$B_{\theta} = [\vec{\beta}_{\theta,0}^{\top},\cdots, \vec{\beta}_{\theta,T-1}^{\top}]^{\top} \in \mathbf{R}^{T \times q}$, $\Sigma_{\epsilon} = \sigma_E^2 \ \mathbf{I}_T \in \mathbf{R}^{T \times T}.$
\end{theorem}
\noindent The proof is provided in Appendix \ref{Appendix C.1}. The variance results for the generalized estimators are analogous to the counterpart of the estimators presented in Theorem~\ref{theom:var_2x2} but replaced with vector representations. For example, the variance of the classic DiD estimator still scales proportionally to the breakage in parallel trends, where instead of $\Delta_{\theta}^2  \sigma_{\theta}^2 +  \Delta_{X}^2 \sigma_{x}^2$ we have the vectorized version $\vec{\Delta}_{\theta}^{\top} \Sigma_{\theta \theta}\vec{\Delta}_{\theta}  + \vec{\Delta}_{X}^{\top} \Sigma_{X X} \vec{\Delta}_{X}$. As such, the main insights from Lemma \ref{Lemma 3.1} and Lemma \ref{Lemma 3.2} carry over. One primary distinction, however, is that the reliability term ($r_{\theta \mid x}$) can no longer be represented as a scalar and instead has a matrix representation. Since the reliability captures the proportion of the variance of the outcome captured by $\mathbf{\theta}$, after controlling for $\mathbf{X}$, this quantity becomes a matrix due to the multidimensionality of $\mathbf{\theta}$. However, following Definition \ref{Def: reliabiltiy}, this matrix could still be interpreted as a multivariate analog of the population R-squared statistic if we were also able to have access to the latent variables (further discussed in Section~\ref{subsection:generalized_reliability}).

\subsection{Revisiting and Generalizing Existing Bias Results} \label{section: revisiting bias}
\noindent To study MSE, we first revisit the bias results and the corresponding main takeaways from \cite{ham2024benefits}. For any generalized DiD estimator (matched or unmatched) $\hat{\tau}$, we denote its bias relative to the ATT ($\tau$) as
$$
\Bias \left( \hat{\tau} \right) := \mathbb{E}\left[\hat{\tau}\right] - \tau, 
$$ 
where we still consider ``asymptotic bias'' as defined in the limiting expressions in Theorem~\ref{Thm: Consistency of DiD and Matching DiD Estimators}.

\begin{lemma}[\textbf{Bias of generalized DiD and generalized Matching DiD Estimators}] \label{Lemma 4.1}  Under the same conditions in Theorem \ref{Thm 4.1}, the biases of our estimators are given by:
\[
\begin{aligned}
\Bias \left( \hat{\tau}_{\text{gDiD}} \right) 
 &= \vec{\Delta}_{\theta}^{\top} \vec{\delta}_{\tilde{\theta}} + \left( \Sigma_{X X}^{-1} \Sigma_{X \theta} \vec{\Delta}_{\theta}   + \vec{\Delta}_{X} \right)^{\top} \vec{\delta}_{x} \\
\Bias  \left( \hat{\tau}_{\text{gDiD}}^{\boldsymbol{X}} \right) &= \vec{\Delta}_{\theta}^{\top} \vec{\delta}_{\tilde{\theta}}\\
\Bias \left( \hat{\tau}_{g\text{DiD}}^{\boldsymbol{X}, \boldsymbol{Y^T}} \right) & = \vec{\beta}_{\theta,T}^\top  \left[ \left(\mathbf{I}_{q \times q} - \mathbf{r}_{\theta|x} \right) \vec{\delta}_{\tilde{\theta}}  \right],
\end{aligned}
\]
where we denote $\vec{\delta}_{\tilde{\theta}} \triangleq \vec{\delta}_{\theta}  - \Sigma_{\theta X} \Sigma_{X X}^{-1} 
 \vec{\delta}_{x}$; and $\mathbf{r}_{\theta|x}$ is defined in Theorem \ref{Thm 4.1}. 
\end{lemma}
\noindent The proof is provided in Appendix \ref{Appendix C.2}. The first two bias results are the same as those derived in Theorem 5.1 of \cite{ham2024benefits}. However, our bias results still remain novel as we derive $\Bias \left( \hat{\tau}_{g\text{DiD}}^{\boldsymbol{X}, \boldsymbol{Y^T}} \right)$ with fewer assumptions. In particular, we do not assume unconditional parallel trends for pre-treatment outcomes ($\vec\beta_{\theta,t} = \vec\beta_{\theta,t'} \ \text{and} \ \vec{\beta}_{x,t} = \vec{\beta}_{x,t'}$ for all $t,t'$ such that $0\leq t,t' \leq T-1$) and deal with multi-dimensional $\boldsymbol{\theta}$, generalizing reliability through the matrix $\mathbf{r}_{\theta|x}$.

To summarize findings from \cite{ham2024benefits}: the authors find that under some regular sign conditions, the bias of $\hat{\tau}^X_{\text{DiD}}$ is always smaller than that of $\hat{\tau}_{\text{DiD}}$, thus they recommend one should always match on observed covariates (Section 6.1 of \cite{ham2024benefits}). However, whether one should additionally match on pre-treatment outcomes depends on a trade-off between the reliability term and the relative magnitude of the breakage in parallel trends from the time-varying coefficients (Section 6.2 of \cite{ham2024benefits}).

These bias takeaways, however, are opposite to our variance takeaways. Specifically, when considering variance, it is not always recommended to match on observed covariates $X$ as shown by Lemma \ref{Lemma 3.2}. Lemma \ref{Lemma 3.2} highlights the importance of carefully analyzing the trade-off between sample size and the extent of parallel trends violation when comparing matching on $X$ or not. When considering only bias, the sample size $n_1, n_0$ does not play an important role, hiding a crucial trade-off one incurs when matching. Furthermore, when considering variance, it is always recommended to match additionally on the pre-treatment outcome(s) after matching on $X$ as shown by Lemma \ref{Lemma 3.3}. This is in stark contrast to major results in \cite{ham2024benefits}, where the authors demonstrate a critical tradeoff when additionally matching on the pre-treatment outcome(s), i.e., it is sometimes but not always better to match on pre-treatment outcomes. Together, these results illustrate a bias-variance trade-off in selecting a uniformly better estimator in practice, underscoring the benefits of evaluating the performance of estimators through the mean squared errors (MSE), which balances both bias and variance, in the decision-making process.

We therefore state the following corollary to present MSE results. Since MSE is equal to Bias$^2$ + Variance, we show the final MSE expressions for completeness: 

\begin{corollary}[\textbf{MSE of generalized DiD and generalized Matching DiD Estimators}]  \label{Corollary 4.1}  Under the same conditions in Theorem \ref{Thm 4.1}, the mean squared errors of our estimators are given by:
\[
\begin{aligned}
 \MSE  \left( \hat{\tau}_{g\text{DiD}} \right)  &=  \Bias^2  \left( \hat{\tau}_{g\text{DiD}} \right) +  \Var  \left( \hat{\tau}_{g\text{DiD}} \right)    \\
 \MSE  \left( \hat{\tau}_{g\text{DiD}}^{\boldsymbol{X}} \right)  &=  \Bias^2  \left( \hat{\tau}_{g\text{DiD}}^{\boldsymbol{X}} \right) +  \Var  \left( \hat{\tau}_{g\text{DiD}}^{\boldsymbol{X}} \right)    \\
 \MSE  \left( \hat{\tau}_{g\text{DiD}}^{\boldsymbol{X}, \boldsymbol{Y^T}} \right)  &=  \Bias^2  \left( \hat{\tau}_{g\text{DiD}}^{\boldsymbol{X}, \boldsymbol{Y^T}} \right) +  \Var \left( \hat{\tau}_{g\text{DiD}}^{\boldsymbol{X}, \boldsymbol{Y^T}} \right)    
\end{aligned}
\]
where $\Var(\cdot)$ is defined in Theorem \ref{Thm 4.1} and $\Bias(\cdot)$ is defined in Lemma \ref{Lemma 4.1}.
\end{corollary}

\subsection{Generalized Reliability}
\label{subsection:generalized_reliability}
\noindent Though we show in the canonical DiD setting in Lemma~\ref{Lemma 3.3} that additionally matching on pre-treatment outcomes always reduces variance relative to matching only on observed covariates, it is unclear how much reduction one could obtain from matching on the pre-treatment outcomes. Theorem~\ref{theom:var_2x2} shows that $\text{Var} \left( \hat{\tau}_{\text{DiD}^{X,Y_0}} \right)$ is reduced critically by the reliability term $r_{\theta|x}$. We explore this further in this section.


To build intuition, we first consider the case where the latent variable is univariate but under multiple pre-period time points $T > 1$. In this setting, the reliability matrix simplifies to a scalar, which admits a more transparent interpretation. The scalar form aligns with the structure developed in the canonical DiD setting, shown in Theorem \ref{Thm 3.1}. We formalize this finding in the following Lemma. 
\begin{lemma}[\textbf{Equivalence in multi-multi case with univariate latent variables}] \label{Lemma 4.2}
Assuming $\theta \in \mathbf{R}$, the matrix reliability term in Theorem \ref{Thm 4.1} could be reduced to a scalar:
$$r_{\theta|x} = \frac{T \bar{\beta}^2_{\theta, \text{pre}} {\sigma}^2_{\tilde \theta}}{T \bar{\beta}^2_{\theta, \text{pre}} {\sigma}^2_{\tilde \theta} + \sigma^2_E},
$$ 
where 
$$
{\sigma}^2_{\tilde \theta} = {\sigma}^2_{\theta} - \Sigma_{\theta X} \Sigma_{X X}^{-1} \Sigma_{X \theta}; \quad \bar{\beta}^2_{\theta, \text{pre}} := \frac{1}{T} \sum_{t=0}^{T-1} \beta_{\theta,t}^2; 
$$
\end{lemma}
\noindent The proof is provided in Appendix \ref{Appendix C.3}. This is equivalent to the reliability term derived in Theorem 5.4 of \cite{ham2024benefits}, though their result requires an additional assumption of unconditional parallel trends for pre-treatment outcomes, i.e., $\beta_{\theta,t} = \beta_{\theta,t'} \ \text{and} \ \vec{\beta}_{x,t} = \vec{\beta}_{x,t'} \ \text{for all } t, t' \in \{0, \ldots, T-1\}.$

Lemma \ref{Lemma 4.2} highlights the benefits of the number of pre-treatment periods available. Specifically, if the average latent coefficients ($\bar{\beta}^2_{\theta, \text{pre}}$) does not shrink as $T \rightarrow \infty$, the reliability increases with more pre-treatment periods because it approaches 1 ($r_{\theta \mid x} < 1$ by Definition \ref{Def: reliabiltiy})  as $T \rightarrow \infty$.

Matching on multiple pre-treatment time points is closely related to synthetic control methods \cite{abadie2010synthetic}, in which one constructs a synthetic comparison unit as a weighted average of control units to closely replicate the observed characteristics, particularly pre-treatment outcomes, of a treated unit. Since the benefit of the reliability term also appears in the corresponding bias expression (See Lemma \ref{Lemma 4.1}), our findings are consistent with the findings in the synthetic control literature \cite{abadie2021using}. Specifically, the literature shows that a large number of pre-treatment periods can help improve pre-treatment fit and control the bias; and the synthetic control estimator is consistent with infinite pre-treatment periods ($T \rightarrow \infty$) under perfect pre-treatment fit, which we exactly recover in our reliability term. Specifically, our results extend to the variance results as Theorem~\ref{theom:var_2x2} shows $r_{\theta \mid x} \rightarrow 1$ as $T \rightarrow \infty$, making $\text{Var} \left( \hat{\tau}_{\text{DiD}^{X,Y_0}} \right) \rightarrow 0$ (ignoring the remaining irreducible error term). Furthermore, with our new results developed in Lemma \ref{Lemma 4.1} and Lemma \ref{Lemma 4.2}, it also follows that $\text{Bias} \left( \hat{\tau}_{\text{DiD}^{X,Y_0}} \right) \rightarrow 0$, making the mean squared errors asymptotically approach zero.

We showed the above holds exactly in the univariate $\theta$ case, where the reliability $r_{\theta \mid x}$ is written as a scalar. However, we can also generalize the above result to multivariate latent variables. To achieve this, we introduce matrix norms to formally state the following proposition: 
\begin{proposition}[\textbf{Convergence of The Reliability Matrix}] \label{Prop 4.1}
Let $\Sigma_{\theta \theta} \succ 0$ and suppose that $\frac{1}{T} \sum_{t=0}^{T-1} \vec{\beta}_{\theta,t} \vec{\beta}_{\theta,t}^\top \xrightarrow{\|\cdot\|_2} Q \succ 0$ as $T \rightarrow \infty$. Then as $T \rightarrow \infty$, we have that the reliability matrix converges to the identity matrix in spectral norm\footnote{Our convergence result also holds under alternative matrix norm such as Frobenius norm $(\|\cdot\|_{F})$.}, i.e. $\mathbf{r}_{\theta|x} \xrightarrow{\|\cdot\|_2}  \mathbf{I}_{q}$, where $\mathbf{r}_{\theta|x}$ is defined in Theorem \ref{Thm 4.1}. 
\end{proposition}
\noindent Proposition~\ref{Prop 4.1} assumes some regularity conditions such that the average matrix of latent coefficients does not degenerate in the limit and the covariance matrix of $\boldsymbol{\theta}$ is positive definite. The core of Proposition~\ref{Prop 4.1} is identical to that in Lemma~\ref{Lemma 4.2}: as $T \rightarrow \infty$, the reliability is ``maximized'' at 1. The matrix version of ``1'' is the identity that provides the highest reliability. The proof is provided in Appendix \ref{Appendix C.4}.

\section{Determining What to Match}
\label{section: Determining What to Match}
\noindent So far, we characterize theoretical results that show the bias and variance in terms of population-level parameters, e.g., $\beta_{\theta, t}, r_{\theta | x}$, etc. While we theoretically accommodate a multi-dimensional latent variable for general mathematical insights, in this section, we recommend using a single latent variable $\theta$ that can arbitrarily capture any degree of confounding in practice unless applied researchers have specific priors regarding the latent information. Doing so makes the reliability term more interpretable and if applied researchers are interested, it also allows them to assess the quality of pre-treatment outcomes by directly estimating the reliability (see Proposition \ref{Prop: Relative Absolute Bias Reduction Pair 2} for the estimator of the single individual reliability term).

In the following subsections, we first present estimation strategies (Section \ref{section: Estimation Strategies}) to quantify pairwise reductions in bias, variance, and MSE based on the theoretical results from Sections \ref{section: univariate results} and \ref{section: multivariate results}. We then discuss practical considerations for determining when and what variables to match on in Section \ref{section: Practical guidance}.

\subsection{Estimation Strategies} \label{section: Estimation Strategies}
\noindent We build upon the estimation strategies proposed by \cite{ham2024benefits}. The key idea underlying these estimation strategies is that, since the latent variables are unobserved, we need to residualize the effect of $X$ from the outcomes for each period to isolate the effect of $\theta$, i.e., we focus on
\begin{equation} \label{Eq: resdiualized outcome}
\tilde{Y}_{i,t} := Y_{i,t} - \hat{\vec{\beta}}_{x,t}^{\top} \boldsymbol{X}_i,   
\end{equation}
where $\hat{\vec{\beta}}_{x,t}$ is obtained by a linear regression of the outcomes at period $t$ onto the covariates within control group.  

Building on this idea, we propose the following estimation strategies. We start by revisiting the conditions and estimation strategies for bias.

\begin{proposition}[\textbf{Relative Absolute Bias Reduction Pair 1 ($\hat{\tau}_{\text{gDiD}}$ and $ \hat{\tau}_{\text{gDiD}}^{\boldsymbol{X}}$}) \textbf{- Conditions and Estimation Strategy for Matching on $\boldsymbol{X}$}]  \label{Prop: Relative Absolute Bias Reduction Pair 1}
Assume $\theta \in \mathbf{R}$ and $T>1$. 
Suppose the conditions of Corollary \ref{Corollary 4.1} hold and further suppose the following three sign conditions hold:
\begin{enumerate}
    \item $\operatorname{sign}\!\left(\vec{\Delta}_x^{\,\top}\,\vec{\delta}_x\right) 
= \operatorname{sign}\!\left(\Delta_\theta \,\delta_\theta\right)$
\item $\operatorname{sign}\!\left(\Delta_\theta \,\Sigma_{\theta X}\,\Sigma_{XX}^{-1}\,\vec{\delta}_x\right) 
= \operatorname{sign}\!\left(\Delta_\theta \,\delta_\theta\right)$
\item $\operatorname{sign}(\delta_\theta) 
= \operatorname{sign}\!\left(\delta_\theta - \Sigma_{\theta X}\,\Sigma_{XX}^{-1}\,\vec{\delta}_x\right).$
\end{enumerate}
Then matching on $\boldsymbol{X}$ will have less absolute bias than the classic unmatched DiD. Additionally,  
\begin{align*}
\widehat{\Delta}^{\text{absolute}}_{\tau_x}  &\xrightarrow{p} \left|\Bias  \left( \hat{\tau}_{\text{gDiD}} \right) \right| -  \left| \Bias    \left( \hat{\tau}_{\text{gDiD}}^{\boldsymbol{X}} \right)  \right| \ge 0, 
\end{align*}
where
\begin{align*}
\widehat{\Delta}^{\text{absolute}}_{\tau_x} &:=  \left| \left(\hat{\vec{\beta}}_{x,T} - \frac{\sum_{t=0}^{T-1} \hat{\vec{\beta}}_{x,t}}{T}\right)^{\top} \hat{\vec{\delta}}_x    \right|,
\end{align*}
with $\hat{\vec{\beta}}_{x,t}$ defined in Equation \eqref{Eq: resdiualized outcome} and $\hat{\vec{\delta}}_x$ as the plug-in estimator of ${\vec{\delta}}_x$ defined in Equation \eqref{eq: specific PT under LSEM}.
\end{proposition}
\noindent This is the same as Guideline 1 (Theorem 6.1) in \cite{ham2024benefits}, so we omit its proof. These sign conditions rule out the edge cases and formalize the claim that we discussed in Section \ref{section: revisiting bias} regarding when matching on $X$ generally removes bias from the observed covariates. The first sign condition states that the pre-existing biases in $\theta$ and $\boldsymbol{X}$ are not in opposite directions. The second condition guarantees that the imbalance in $\delta_\theta$ is decreased, rather than amplified, by the additional information about $\theta$ obtained through matching on a correlated $\boldsymbol{X}$. The third sign condition ensures that the additional bias reduction achieved by matching on a correlated $\boldsymbol{X}$ does not overcorrect the bias. We leave evaluating how plausible these conditions are to future empirical work. Importantly, these conditions are not necessary: there are many situations in which they fail, yet matching on observed covariates remains beneficial. We now turn to the bias estimation strategy when considering matching additionally on the pre-treatment outcome.

\begin{proposition}[\textbf{Relative Absolute Bias Reduction Pair 2 ($ \hat{\tau}_{\text{gDiD}}^{\boldsymbol{X}}$ and $\hat{\tau}_{g\text{DiD}}^{\boldsymbol{X}, \boldsymbol{Y^T}} $}) \textbf{- Estimation Strategy for Matching on $\boldsymbol{X} \& \boldsymbol{Y^T}$}]  \label{Prop: Relative Absolute Bias Reduction Pair 2}
Assume $\theta \in \mathbf{R}$ and $T>1$. Suppose the conditions of Corollary \ref{Corollary 4.1} hold and further suppose the stable effects of $\theta$ across two arbitrary  pre-treatment periods, i.e., ${\beta}_{\theta,s} =  {\beta}_{\theta,q}, s \neq q <T$. Then we have 
\begin{align*}
\widehat{\Delta}^{\text{absolute}}_{\tau_{x,xy}}  &\xrightarrow{p}  \left| \Bias    \left( \hat{\tau}_{\text{gDiD}}^{\boldsymbol{X}} \right)\right| - \left|\Bias    \left( \hat{\tau}_{g\text{DiD}}^{\boldsymbol{X}, \boldsymbol{Y^T}} \right) \right|,  
\end{align*}
where
\begin{align*}
\widehat{\Delta}^{\text{absolute}}_{\tau_{x,xy}}  &:= 
 \widehat{\Delta}_{\theta} \widehat{\delta}_{\tilde{\theta}}   -  \widehat{\beta}_{\theta,T} \widehat{\delta}_{\tilde{\theta}} (1-\hat{r}_{\theta|x} ) .
\end{align*}
Individual terms are defined as: 
\[
\hat{\Delta}_\theta = \sqrt{ \widehat{\Var}\left(\tilde{Y}_{i,T} - \bar{\tilde{Y}}_{i,0:(T-1)} \mid Z_i = 0\right) - \frac{T+1}{T} \hat{\sigma}_E^2 }, \quad \hat\delta_{\tilde{\theta}} =
\frac{
\left|\widehat{\mathbb{E}}\!\left[\bar{\tilde{Y}}_{i,0:(T-1)} \mid Z_i = 1\right]
-
\widehat{\mathbb{E}}\!\left[\bar{\tilde{Y}}_{i,0:(T-1)} \mid Z_i = 0\right]  
\right|}{
\hat{\bar \beta}_{\theta, \text{pre}}
},
\]
where $\tilde{Y}_{i,t}$ is the residualized outcome at period $t$ in Equation~\eqref{Eq: resdiualized outcome}; $\bar{\tilde{Y}}_{i,\cdot}$ denotes an average of these residualized outcomes; $\widehat{\mathbb{E}}$ and $\widehat{\Var}$ denote the plug-in mean and variance estimator; and
\begin{align*}
& \hat{\sigma}_E^2 = \frac{1}{2} \widehat{\Var} \left( \tilde{Y}_{i, t} - \tilde{Y}_{i, t'} \right), t \neq t' <T; \\
&\hat{\bar \beta}_{\theta, \text{pre}} = \frac{1}{T} \sum_{t=0}^{T-1} \hat{\beta}_{\theta, t}, \quad \hat{\beta}_{\theta, t} = \sqrt{ \widehat{\Var}(\tilde{Y}_{i,t} \mid Z_i = 0) - \hat{\sigma}_E^2 }, \quad t = 0, 1, \dots, T; 
\end{align*}
and 
\[
\hat{r}_{\theta|x} = \frac{T \hat{\bar \beta}_{\theta, \text{pre}}^2}{T \hat{\bar \beta}_{\theta, \text{pre}}^2 + \hat{\sigma}_E^2}
\quad \text{with} \quad
\hat{\bar \beta}_{\theta, \text{pre}}^2 = \frac{1}{T} \sum_{t=0}^{T-1} \hat{\beta}_{\theta, t}^2.
\]
\end{proposition}
\noindent The proof is provided in Appendix \ref{Appendix: Proof of Prop: Relative Absolute Bias Reduction Pair 2}. As shown in Lemma \ref{Lemma 4.1} and Lemma \ref{Lemma 4.2}, the key difference between ours and Guideline 2 (Theorem 6.2) in \cite{ham2024benefits} is that we don't need an additional assumption of unconditional parallel trends for pre-treatment outcomes to derive a closed-form ${r}_{\theta|x}$, though we still need  the stable effects of $\theta$ across two arbitrary pre-treatment periods to consistently estimate the residual error term ($\sigma^2_{E}$) so that the consistent estimation of the products can be achieved.  Different choices of time periods are feasible here, so we recommend that applied researchers select two periods that seem the most stable.

Next, we propose estimation strategies for results of relative variance reductions. Let ${\text{Var}_\text{core}}  \left( \cdot \right)$ denote the expression inside the brackets of variance components (the ones after upfront sample coefficients) in Theorem \ref{Thm 4.1}, for instance, 
${\text{Var}_\text{core}}  \left( \hat{\tau}_{\text{gDiD}}^{\boldsymbol{X}} \right) =   \frac{T+1}{T} \sigma_{E}^2 + \vec{\Delta}_{\theta}^{\top} \Sigma_{\tilde{\theta} \tilde{\theta}} \vec{\Delta}_{\theta}.$ Then we denote $\widehat{\text{Var}_\text{core}}  \left( \cdot \right)$ as the corresponding estimator. Due to technical nuisances, we do not formally state a consistency theorem for the actual variance (and thus actual MSE). For example, suppose we want a consistency result such as: $\widehat{\text{Var}}  \left( \hat{\tau}_{\text{gDiD}}^{\boldsymbol{X}} \right)= \left(n_1^{-1} + n_1^{-1} \right) \widehat{\text{Var}_\text{core}}  \left( \hat{\tau}_{\text{gDiD}}^{\boldsymbol{X}} \right) \xrightarrow{p} \left(n_1^{-1} + n_1^{-1} \right)  {\text{Var}_\text{core}} \left( \hat{\tau}_{\text{gDiD}}^{\boldsymbol{X}} \right)$. The asymptotics becomes cumbersome because the limit is taken with respect to $n_1$, which are directly in the expressions of the variance. Therefore, to bypass this issue, we present consistency results only on ${\text{Var}_\text{core}}(\cdot)$ terms. We formalize this in the following proposition:

\begin{proposition}[\textbf{Estimation Strategies for Relative Variance Reduction}] \label{Prop: Estimation Strategies for Relative Variance Reduction}
Suppose the conditions of Corollary \ref{Corollary 4.1} hold. Then we have 
$$
\widehat{\Var_{\text{core}}}  \left( \hat{\tau}_{\text{gDiD}} \right) \xrightarrow{p} \Var_{\text{core}}  \left( \hat{\tau}_{\text{gDiD}} \right), \ \widehat{\Var_{\text{core}}}  \left( \hat{\tau}_{\text{gDiD}}^{\boldsymbol{X}} \right) \xrightarrow{p} \Var_{\text{core}}    \left( \hat{\tau}_{\text{gDiD}}^{\boldsymbol{X}} \right), \ 
\widehat{\Var_{\text{core}}}  \left( \hat{\tau}_{g\text{DiD}}^{\boldsymbol{X}, \boldsymbol{Y^T}}\right) \xrightarrow{p} \Var_{\text{core}}    \left( \hat{\tau}_{g\text{DiD}}^{\boldsymbol{X}, \boldsymbol{Y^T}}\right),
$$
where
\begin{align*}
\widehat{\Var_{\text{core}}}  \left( \hat{\tau}_{\text{gDiD}} \right) &:=  \widehat{\Var} \left( \tilde{Y}_{i,T} - \frac{\sum_{t=0}^{T-1}\tilde{Y}_{i,t}}{T} \mid Z_i = 0 \right)  + 
\left(\hat{\vec{\beta}}_{x,T} - \frac{\sum_{t=0}^{T-1} \hat{\vec{\beta}}_{x,t}}{T}\right)^{\top} \hat{\Sigma}_{XX} \left(\hat{\vec{\beta}}_{x,T} - \frac{\sum_{t=0}^{T-1} \hat{\vec{\beta}}_{x,t}}{T}\right)  
 \\    
\widehat{\Var_{\text{core}}}  \left( \hat{\tau}_{\text{gDiD}}^{\boldsymbol{X}} \right) &:=
\widehat{\Var} \left( \tilde{Y}_{i,T} - \frac{\sum_{t=0}^{T-1}\tilde{Y}_{i,t}}{T} \mid Z_i = 0 \right)  \\
\widehat{\Var_{\text{core}}}   \left( \hat{\tau}_{g\text{DiD}}^{\boldsymbol{X}, \boldsymbol{Y^T}} \right)  &:= 
\widehat{\Var}\left( \tilde{Y}_{i, T} \mid Z_i = 0 \right) \\
&-\widehat{\Cov}\left( \tilde{Y}_{i, T}, \tilde{Y}_{i, 0:T-1} \mid Z_i = 0\right) \left(\widehat{\Var}\left( \tilde{Y}_{i, 0:T-1}\right)\right)^{-1} 
 \widehat{\Cov} \left( \tilde{Y}_{i, 0:T-1}, \tilde{Y}_{i, T} \mid Z_i = 0 \right),
\end{align*}
and $\tilde{Y}_{i,t}$ is the residualized outcome at period $t$  with $\hat{\vec{\beta}}_{x,t}$ to be the regression coefficients defined in Equation~\eqref{Eq: resdiualized outcome}; $\widehat{\Sigma}_{XX} := \frac{1}{n-1} \sum_{i=1}^{n} (\boldsymbol{X}_i -\bar{\boldsymbol{X}})(\boldsymbol{X}_i - \bar{\boldsymbol{X}})^{\top};$ $\widehat\Var$ and $\widehat\Cov$ denote the plug-in estimators for variance and covariance, respectively.

For practical usage of obtaining relative variance reductions, we recommend to simply multiply $\widehat{\Var}_{\text{core}}(\cdot) $ by their corresponding upfront sample coefficients ($1/n_1, 1/n_0$) in Theorem \ref{Thm 4.1} to retrieve back complete finite-sample variance estimates. In particular, with a slight abuse of notation, 
\begin{align*}
\widehat{\Gamma}_{x} \approx \Var  \left( \hat{\tau}_{\text{gDiD}} \right) - \Var    \left( \hat{\tau}_{\text{gDiD}}^{\boldsymbol{X}} \right); \quad \widehat{\Gamma}_{x,xy}  \approx \Var   \left( \hat{\tau}_{\text{gDiD}}^{\boldsymbol{X}} \right)  - \Var   \left( \hat{\tau}_{g\text{DiD}}^{\boldsymbol{X}, \boldsymbol{Y^T}}\right),
\end{align*}
where $\approx$ denotes that the expression is an approximate estimate, and 
\begin{align*}
\widehat{\Gamma}_{x} &:= \left(\frac{1}{n_1} + \frac{1}{n_0} \right) \widehat{\Var_{\text{core}}}  \left( \hat{\tau}_{\text{gDiD}} \right)  -  \left(\frac{1}{n_1} + \frac{1}{n_1} \right)  \widehat{\Var_{\text{core}}}  \left( \hat{\tau}_{\text{gDiD}}^{\boldsymbol{X}} \right) \\
\widehat{\Gamma}_{x,xy} &:= \left(\frac{1}{n_1} + \frac{1}{n_1} \right)  \widehat{\Var_{\text{core}}}  \left( \hat{\tau}_{\text{gDiD}}^{\boldsymbol{X}} \right) - \left(\frac{1}{n_1} + \frac{1}{n_1} \right)  \widehat{\Var_{\text{core}}}  \left( \hat{\tau}_{\text{gDiD}}^{\boldsymbol{X}, \boldsymbol{Y^T}} \right) 
\end{align*}
\end{proposition}
\noindent The proof is provided in Appendix \ref{Appendix: Proof of Estimation Strategies for Relative Variance Reduction}.

For notational simplicity, denote $\mathcal{S} := \left\{ \hat{\tau}_{\mathrm{gDiD}}, \ \hat{\tau}_{\mathrm{gDiD}}^{\boldsymbol{X}}, \ \hat{\tau}_{\mathrm{gDiD}}^{\boldsymbol{X}, \boldsymbol{Y}^{\mathbf{T}}} \right\}.$ To obtain the results for relative reductions in MSE, we need to be able to consistently estimate relative reductions in bias square, i.e., $\text{Bias}(i)^2 - \text{Bias}(j)^2$ for estimators $i, j \in \mathcal{S}$. Unfortunately being able to consistently estimate $\text{Bias}(i) - \text{Bias}(j)$ as shown in Proposition~\ref{Prop: Relative Absolute Bias Reduction Pair 1} -~\ref{Prop: Relative Absolute Bias Reduction Pair 2} does not imply we can estimate the squares of those due to the non-linear transformation. Formally, we have the following impossibility results: 
\begin{lemma}[\textbf{Impossibility Results}] \label{Lemma: Impossibility Results} 
Let $\widehat{\Bias}(\cdot)$ denote a bias estimator. Without imposing additional assumptions beyond Corollary \ref{Corollary 4.1}, for any estimator $i \in \mathcal{S}$, there exists no consistent estimator of its bias, i.e., for any choice of $\widehat{\Bias}(i)$,
\[
\widehat{\Bias}(i) \xrightarrow{p} \Bias(i)
\]
does not hold, where $\Bias(i)$ is the asymptotic bias of estimator $i$ defined in Lemma \ref{Lemma 4.1}. 

Moreover, there does not exist a consistent estimator for the relative reduction in bias square and thus mean squared error (MSE) between any pair of estimators in $\mathcal{S}$.
\end{lemma}

\noindent The proof is provided in Appendix \ref{Appendix: Impossibility Results}. The main idea behind the failure of existence of such consistent bias (and thus bias square and MSE) estimator is that we can neither identify the ATT nor consistently estimate it. Therefore, we need to impose additional sign conditions to provide meaningful ``bounds'' to acquire relative reductions of bias square and MSE. Formally, we have 
\begin{lemma}[\textbf{Bounds on Bias Square under Sign Conditions}] \label{Lemma: Bounds on Bias Square under Sign Conditions}
Under the same notations of Lemma \ref{Lemma: Impossibility Results}, denote $\widehat{\text{Bias}}^2 (\cdot) := \left(\widehat{\Bias}  \left( \cdot \right) \right)^2.$ For any $i,j \in \mathcal{S}$ with $i \neq j$, suppose that 
\[
\widehat{\Bias}^2(i) - \widehat{\Bias}^2(j)
\;\xrightarrow{p}\;
\Delta^{\Bias^2}_{ij}.
\]
Then,
\begin{align*}
\text{if} \quad 
&\operatorname{sign}\left[\left( \Bias(i) - \Bias(j) \right)\cdot \tau \right] \ge 0,
\\
&\quad \Delta^{\Bias^2}_{ij} \;\ge\; \Bias^2(i) - \Bias^2(j),
\\[6pt]
\text{if} \quad 
&\operatorname{sign}\left[\left( \Bias(i) - \Bias(j) \right)\cdot \tau \right] < 0,
\\
&\quad \Delta^{\Bias^2}_{ij} < \Bias^2(i) - \Bias^2(j).
\end{align*}
\end{lemma}
\noindent The proof is provided in Appendix \ref{Appendix: Proof of Lemma: Bounds on Bias Square under Sign Conditions}. We acknowledge that the sign conditions can't hold simultaneously, i.e., we could only get one side of the bounds (either the upper bound or the lower bound) for the relative bias square reductions. However, this still remains informative in some cases. For any $i,j \in \mathcal{S}$ with $i \neq j$, denote $\widehat{\Bias}^2(i) - \widehat{\Bias}^2(j) := \widehat \Delta^{\Bias^2}_{ij}$ and $\Bias^2(i) - \Bias^2(j) := \Delta^{\text{Actual} \Bias^2}_{ij}$.  Then we have the following two informative scenarios laid out in Figure \ref{Figure: informative scenario 1} and \ref{Figure: informative scenario 2}. The interpretation goes as follows: for Figure \ref{Figure: informative scenario 1},  under the sign assumption that $ \operatorname{sign}\left[\left( \Bias(i) - \Bias(j) \right)\cdot \tau \right] \ge 0$, from Lemma \ref{Lemma: Bounds on Bias Square under Sign Conditions}, we know we could only consistently estimate an upper bound of the actual bias square; therefore if $\widehat \Delta^{\Bias^2}_{ij} < 0$, then we could be certain that $\Delta^{\text{Actual} \Bias^2}_{ij} < 0$, which then implies that $i$ has smaller bias square than the counterpart of $j$. Similar reasoning applies to Figure \ref{Figure: informative scenario 2}.

\begin{figure}[ht]
\centering  
\begin{tikzpicture}[scale=1.0]
 
  \draw (-1,0) -- (8,0);
  
  \draw (0,0.1) -- (0,-0.1);
  \node[above] at (0,0.1) {\textcolor{red}{$\Delta^{\text{Actual} \Bias^2}_{ij}$}};

  \draw (2,0.1) -- (2,-0.1);
  \node[above] at (2,0.1) {\textcolor{magenta}{$\widehat \Delta^{\Bias^2}_{ij}$}};
  
  \draw (3.5,0.1) -- (3.5,-0.1);
  \node[below] at (3.5,-0.1) {0};
  
\node at (1,-0.6) {}; 
\end{tikzpicture} 
\caption{Informative Scenario 1: $\operatorname{sign}\left[\left( \Bias(i) - \Bias(j) \right)\cdot \tau \right] \ge 0$, where $i,j \in \mathcal{S}$ with $i \neq j$}
\label{Figure: informative scenario 1}
\end{figure}

\begin{figure}[ht]
\centering  
\begin{tikzpicture}[scale=1.0]
  \node at (0,1) {} ;

  \draw (-1,0) -- (8,0);
  
  \draw (3.5,0.1) -- (3.5,-0.1);
  \node[below] at (3.5,-0.1) {0};

  \draw (4.5,0.1) -- (4.5,-0.1);
  \node[above] at (4.5,0.1) {\small \textcolor{cyan}{$\widehat \Delta^{\Bias^2}_{ij}$}};

  \draw (6,0.1) -- (6,-0.1);
  \node[above] at (6,0.1) {\textcolor{red}{$\Delta^{\text{Actual} \Bias^2}_{ij}$}};

  \node at (5.25,-0.6) {} ; 
\end{tikzpicture}
\caption{Informative Scenario 2: $\operatorname{sign}\left[\left( \Bias(i) - \Bias(j) \right)\cdot \tau \right] < 0$, where $i,j \in \mathcal{S}$ with $i \neq j$}
\label{Figure: informative scenario 2}
\end{figure}

We acknowledge the sign conditions are not testable in practice. However, if the researcher has a prior guess on the sign of the treatment effect $\tau$, then there exist data-driven strategies to check this condition. 
First, the researcher can still estimate the relative bias reduction by taking the pairwise difference of bias estimator in Proposition \ref{Prop: Estimation Strategies for Relative MSE Reduction under Sign Conditions}, or directly using 
Proposition \ref{Prop: Estimation Strategies for Relative Bias Reduction} in Appendix \ref{Appendix: Proofs for Relative (Non-Absolute) Bias Reduction}. Consequently, researchers can estimate $\operatorname{sign}\left[\left( \Bias(i) - \Bias(j) \right)  \right]$. Unfortunately, what we want is $\operatorname{sign}\left[\left( \Bias(i) - \Bias(j) \right)\cdot \tau \right] $. Therefore, knowing the sign $\tau$ is a sufficient condition for the final estimate. Although this still requires some prior knowledge, researchers often have strong domain knowledge on at least the sign of the treatment effect for many applications.

Now given these sign conditions, we are ready to present estimation strategies for results of relative MSE reductions.

\begin{proposition}[\textbf{Estimation Strategies for Relative MSE Reduction under Sign Conditions}] \label{Prop: Estimation Strategies for Relative MSE Reduction under Sign Conditions}
Suppose the conditions of Corollary \ref{Corollary 4.1} hold, along with symbols defined same as in Proposition \ref{Prop: Estimation Strategies for Relative Variance Reduction}, Lemma \ref{Lemma: Impossibility Results} and Lemma \ref{Lemma: Bounds on Bias Square under Sign Conditions}. For any $i,j \in \mathcal{S}$ with $i \neq j$, without loss of generality, further assume one of the sign conditions hold: $\operatorname{sign}\left[\left( \Bias(i) - \Bias(j) \right)\cdot \tau \right] \ge 0$. Then with a slight abuse of notation, we have 
\begin{align*}
\widehat{\kappa}_{x} & \rightarrow \kappa^{\MSE}_{x} \ \geq\  \MSE\left( \hat{\tau}_{\text{gDiD}} \right) - \MSE   \left( \hat{\tau}_{\text{gDiD}}^{\boldsymbol{X}} \right) \\
\widehat{\kappa}_{x,xy} & \rightarrow  \kappa^{\MSE}_{x,xy} \ \ge \ \MSE   \left( \hat{\tau}_{\text{gDiD}}^{\boldsymbol{X}} \right)  - \MSE  \left( \hat{\tau}_{g\text{DiD}}^{\boldsymbol{X}, \boldsymbol{Y^T}}\right),
\end{align*}
where the estimates are given by:
\begin{align*}
\widehat{\kappa}_{x}  &:= \left[ \left(\widehat{\Bias}  \left( \hat{\tau}_{\text{gDiD}} \right) \right)^2  - \left(\widehat{\Bias}   \left( \hat{\tau}_{\text{gDiD}}^{\boldsymbol{X}} \right) \right)^2   \right] + \widehat{\Gamma}_{x}  \\
\widehat{\kappa}_{x,xy} &:= \left[   \left(\widehat{\Bias}   \left( \hat{\tau}_{\text{gDiD}}^{\boldsymbol{X}} \right) \right)^2 - \left(\widehat{\Bias}   \left( \hat{\tau}_{g\text{DiD}}^{\boldsymbol{X}, \boldsymbol{Y^T}} \right) \right)^2 \right] +  \widehat{\Gamma}_{x,xy},
\end{align*}
with $\hat\Gamma_x, \hat\Gamma_{x, xy}$ as estimators of relative variance reductions defined in  Proposition~\ref{Prop: Estimation Strategies for Relative Variance Reduction}; and one set of potential bias estimators\footnote{We remark again that these bias estimators are not consistent to the corresponding bias in the limit. We therefore need to put sign conditions and pairwise MSE reductions are inequality based.} is:
\begin{align*}
\widehat{\Bias}  \left( \hat{\tau}_{\text{gDiD}} \right) &:=
\widehat{\mathbb{E}}\left[ \tilde{Y}_{i, T} -\bar{\tilde{Y}}_{i, 0:(T-1)} \mid Z_i = 1 \right]
- \widehat{\mathbb{E}}\left[ \tilde{Y}_{i, T} -\bar{\tilde{Y}}_{i, 0:(T-1)} \mid Z_i = 0 \right]  +  \left(\hat{\vec{\beta}}_{x,T} - \frac{\sum_{t=0}^{T-1} \hat{\vec{\beta}}_{x,t}}{T}\right)^{\top} \hat{\vec{\delta}}_x \\
\widehat{\Bias}   \left( \hat{\tau}_{\text{gDiD}}^{\boldsymbol{X}} \right) &:=
\widehat{\mathbb{E}}\left[ \tilde{Y}_{i, T} -\bar{\tilde{Y}}_{i, 0:(T-1)} \mid Z_i = 1 \right]
- \widehat{\mathbb{E}}\left[ \tilde{Y}_{i, T} -\bar{\tilde{Y}}_{i, 0:(T-1)} \mid Z_i = 0 \right] \\
\widehat{\Bias}   \left( \hat{\tau}_{g\text{DiD}}^{\boldsymbol{X}, \boldsymbol{Y^T}} \right)  &:=
\widehat{\mathbb{E}}\left[ \tilde{Y}_{i, T}  \mid Z_i = 1 \right]
- \widehat{\mathbb{E}}\left[ \tilde{Y}_{i, T} \mid Z_i = 0 \right] \\
&-   \widehat{\Cov}\left( \tilde{Y}_{i, T}, \tilde{Y}_{i, 0:T-1} \mid Z_i = 0 \right) \left(\widehat{\Var}\left( \tilde{Y}_{i, 0:T-1}\right)\right)^{-1} \left(\widehat{\mathbb{E}}\left[ \tilde{Y}_{i, 0:T-1}  \mid Z_i = 1 \right]
- \widehat{\mathbb{E}}\left[ \tilde{Y}_{i, 0:T-1} \mid Z_i = 0 \right] \right).
\end{align*}
\end{proposition}
\noindent The proof is omitted as it immediately follows from other Lemmas and Propositions. For practical usage of $\widehat{\kappa}_{x}$ and $\widehat{\kappa}_{x,xy}$, we need to determine if $\kappa^{\MSE}_{x} $, ${\kappa}^{\MSE}_{x,xy}$ serve as an upper bound or a lower bound of the corresponding relative MSE reductions. For clarity, we present the theorem when $\operatorname{sign}\left[\left( \Bias(i) - \Bias(j) \right)\cdot \tau \right] \ge 0$ holds. However, if $\operatorname{sign}\left[\left( \Bias(i) - \Bias(j) \right)\cdot \tau \right] < 0$ alternatively held, then the first two lines would be replaced with: 
$$\kappa^{\MSE}_{x} <  \MSE\left( \hat{\tau}_{\text{gDiD}} \right) - \MSE   \left( \hat{\tau}_{\text{gDiD}}^{\boldsymbol{X}} \right), \kappa^{\MSE}_{x,xy} \ < \ \MSE   \left( \hat{\tau}_{\text{gDiD}}^{\boldsymbol{X}} \right)  - \MSE  \left( \hat{\tau}_{g\text{DiD}}^{\boldsymbol{X}, \boldsymbol{Y^T}}\right)$$
to flip the inequality. Therefore, we recommend applied researchers to assume a certain sign condition as described in Lemma \ref{Lemma: Bounds on Bias Square under Sign Conditions} and proceed with reasoning corresponding to Figure \ref{Figure: informative scenario 1} or Figure \ref{Figure: informative scenario 2} to make ultimate decisions.

\subsection{Practical Guidance} \label{section: Practical guidance}
\noindent Now using our model as a working approximation, and balancing the bias-variance trade-off across estimators, we propose the following step-by-step guideline with respect to when and what to match for applied researchers. \\
\textbf{Step 1: } First, if applied researchers have strong domain knowledge and reasonable belief that the parallel trends assumption holds, then they should use the classic unmatched DiD estimator because it identifies ATT while retaining efficiency among the three estimators. Otherwise, proceed to the next steps, considering the potential usage of a matching-based DiD estimator. \\
\textbf{Step 2: } Second,  if applied researchers weigh bias and variance equally, in light of the bias-variance trade-off, then we recommend that applied researchers to invoke the usage of Proposition \ref{Prop: Estimation Strategies for Relative MSE Reduction under Sign Conditions} to have estimates of relative MSE reductions. If it is hard to assume the sign of the treatment effect $\tau$ (thus inferring the sign condition from Lemma \ref{Lemma: Bounds on Bias Square under Sign Conditions}) or if the obtained results with the bounds are not informative, or the estimated MSEs of any pairwise comparisons are of a similar magnitude, then proceed with the next steps to have a further case-by-case diagnosis. \\
\textbf{Step 3: } Third, depending on the specific context and the available sample sizes, we suggest placing emphasis on different considerations. We showcase two examples as implementation criteria below. 
\begin{enumerate}
    \item For applied researchers working with moderate to large sample sizes, we recommend placing greater emphasis on the bias component when choosing an estimator. In such settings, variance is less volatile and tends to play a smaller role, and minimizing bias is more crucial to avoid potential sign reversals in the estimated treatment effects.
\begin{example}[\textbf{Matching or not based on bias criteria}] \label{Example 1}
     If an applied researcher aims to select an estimator based on (absolute) bias considerations, then use Proposition \ref{Prop: Relative Absolute Bias Reduction Pair 1} and Proposition \ref{Prop: Relative Absolute Bias Reduction Pair 2}  to quantify the pairwise (absolute) bias reduction. In particular, it may suffice to only compare if $\widehat{\Delta}^{\text{absolute}}_{\tau_{x,xy}} > 0$ because one is guaranteed to have a bias reduction for matching on observed covariates over classic unmatched DiD under sign conditions in Proposition \ref{Prop: Relative Absolute Bias Reduction Pair 1}.
\end{example}
    \item For applied researchers working with small sample sizes, we recommend placing greater emphasis on the variance component because variance decreases with sample sizes. This is particularly important for cases where limited data may lead to poor matches and variance will be more volatile, while our bias results may be fairly conservative.  
        \begin{example}[\textbf{Matching or not based on variance criteria}] \label{Example 2}
If an applied researcher aims to select an estimator based on variance considerations, then use Proposition \ref{Prop: Estimation Strategies for Relative Variance Reduction} to quantify the pairwise variance reduction. In particular, it may suffice to only compare if $\widehat{\Gamma}_{x} > 0 $ because one is guaranteed to have a variance reduction of matching on both pre-treatment outcomes and observed covariates over matching only observed covariates (See Lemma \ref{Lemma 3.3}).     
\end{example}
\end{enumerate}
Though we present a step-by-step guideline, it is meant to serve as a general template rather than a strict protocol, i.e., we are not claiming that MSE is the most important metric (nor are we claiming MSE is not the most important) for different scenarios. Our estimation strategies in Section \ref{section: Estimation Strategies} are designed to be flexible, allowing applied researchers to separately estimate relative bias, relative variance and relative MSE (under sign conditions) reductions, and then choose, based on their own goals, which objective to prioritize.

\section{Application}  \label{section: Application}
\noindent In this section, we apply the estimation strategies and step-by-step guideline outlined in Section \ref{section: Determining What to Match} to an empirical setting in operations to provide novel managerially robust causal conclusions. Specifically, we revisit the dataset from \cite{wang2022monetary}, where the authors study how the introduction of monetary incentives (paying content creators directly by making users pay for consuming their content) on the knowledge-sharing platform, Zhihu\footnote{Zhihu is a popular online Q\&A community in China; as of December 2020, it had over 370 million registered users, 44 million questions, and 260 million answers, where the numbers are obtained from \href{https://www.36kr.com/p/1025344651233288}{36Kr} (accessed September 2025).},  influences non-rewarded knowledge activities. Understanding this problem is of both empirical and managerial importance: Consumers benefit from both paid and unpaid knowledge, while content producers are naturally drawn to paid activities, and at the same time, sustaining contributions to unpaid activities is essential for the long-term sustainability of digital platforms. However, since the introduction of monetary incentives can generate both positive and negative spillover effects on users’ contribution to non-paid activities, the net impact on unpaid contributions remains unclear. 

In May 2016, Zhihu launched Zhihu Live, a service that allowed users to monetize their expertise through interactive live talks. \cite{wang2022monetary} examine the spillover effects of this service, which was not randomized, on hosts’ unpaid activity, in terms of the quantity of contributed answers, and find evidence of positive spillovers without negative impacts on quality of answers. Their results suggest that monetary incentives can strengthen platform sustainability by encouraging continued user engagement, with reputation building as a major plausible mechanism. In particular, they find that the spillover effect is stronger among live
hosts with lower prior reputation as higher reputation content creators could translate reputation into more
revenue from holding the live talks via both higher entrance fees and more listeners. Consequently, lower reputation hosts tend to invest more effort in non-paid activities to build their reputation, enabling them to later cash more from holding live talks. These findings remain robust against alternative explanations such as explicit promotion or reciprocity.

\subsection{Empirical Setting}
\noindent To reach the aforementioned conclusions, the authors use a M-DiD strategy, given that the treatment (whether an individual is a live host) was self-selected (non-randomized). To address potential self-selection bias, the authors construct matched sets of live hosts and non-hosts with comparable user characteristics. They then estimate the treatment effect using DiD (TWFE regression) on the matched sample. Their primary outcome of interest is the log of the number of answers contributed by a user in a given month, and  they consider two covariates\footnote{In some analyses they consider the third covariate: user tenure, namely number of months since
the user's first answer contribution. }: (1) the log of the average number of characters per answer by a user in a given month, (2) the log of the average number of vote-ups per answer by a user in a given month. We follow the estimation strategies and guidelines in Section \ref{section: Determining What to Match} using the same dataset the authors used with a few adjustments. 

First, although the application involves a staged design, with multiple cutoff dates to accommodate heterogeneous exposure of live hosts to the Zhihu Live feature, the authors, in Appendix E, pool all samples to estimate an overall effect.  For presentational clarity, we adopt this pooled specification as our primary replication strategy. Following their set-up, we also take the main cutoff date May 2016 as the intervention date $T=0$, and define the pre-treatment window as the 12 months preceding the intervention\footnote{Following their approach, we also restrict the analysis to users who provided at least one answer before this point.}. We further aggregate all post-treatment periods into a single post-treatment period\footnote{This application has multiple (monthly-based) post-treatment periods until May 2018, but the authors averaged all those periods to one period to get an estimate of ATT.}. Second, for time varying covariates, we take the average of these covariates over the relevant periods and treat them as time invariant covariates.

\subsection{Results}
\noindent As context, the authors reported that the estimated treatment effect is $\hat{\tau} = 0.102$ with standard error of 0.029 only after matching on  both covariates $\mathbf{X}$ and the pre-treatment outcome (see more details in Appendix E in \cite{wang2022monetary}), i.e., giving monetary incentives to live hosts statistically significantly increases user engagement and platform activity. The authors didn't give formal justifications on why such matching is performed. Given this, we apply our estimation strategies and proposed guidelines to examine whether the authors’ decision to match on both covariates and the pre-treatment outcome is justified, or would it be preferable to match solely on covariates or to not match at all?  

\begin{table}[h]
\centering
\renewcommand{\arraystretch}{1}
\resizebox{1\linewidth}{!}{  
\begin{tabular}{lcc}
\toprule
 & No Match and Match on $\mathbf{X}$ &  \quad  Match on $\mathbf{X}$ and Match on $\mathbf{X}$ and $\mathbf{Y^T}$ \\
\midrule
Estimated |Bias| Reduction & $\widehat{\Delta}^{\text{absolute}}_{\tau_x} = 0.02509 $ & $\widehat{\Delta}^{\text{absolute}}_{\tau_{x,xy}}  = 0.02593$ \\
\addlinespace
Estimated Variance Reduction & $\widehat{\Gamma}_{x} =  -0.00027 $ & $\widehat{\Gamma}_{x,xy} = 0.00041 $ \\
\addlinespace
Sign of Relative Bias Reduction & Negative & Negative \\
\addlinespace
Estimated MSE Reduction & Uninformative (Figure \ref{Figure: informative scenario 2} scenario): $\widehat \Delta^{\Bias^2}_{x} = -0.0018< 0$ & Uninformative (Figure \ref{Figure: informative scenario 2} scenario): $\widehat \Delta^{\Bias^2}_{x,xy} = -0.0028< 0$ \\
\addlinespace
Suggested Match Decision & Match on $\mathbf{X}$ (|Bias| Criteria; Variance Cost Negligible) &  Match on $\mathbf{X}$ and $\mathbf{Y^T}$ (|Bias| and Variance Criteria) \\
\bottomrule
\end{tabular}
}
\caption{Comparison of Matching Strategies (Relative Reductions): The total sample size is 3,841, of which 1,296 units are treated.}  
\label{tab:application}
\end{table}

Table \ref{tab:application} reports the following findings. First, the first row of Table~\ref{tab:application} shows that matching on observed covariates $\mathbf{X}$ yields a sizable reduction in (absolute) bias by about $\widehat{\Delta}^{\text{absolute}}_{\tau_{x}}  = 0.02509$, corresponding to a reduction of about 24.6\%; and matching on both observed covariates and the pre-treatment outcome(s) reduces bias further by $\widehat{\Delta}^{\text{absolute}}_{\tau_{x,xy}}  = 0.02593$, or about a 25.4\% decrease: to apply Proposition \ref{Prop: Relative Absolute Bias Reduction Pair 2} for estimating $\widehat{\Delta}^{\text{absolute}}_{\tau_{x,xy}}$, we assume $\beta_{\theta,-1} = \beta_{\theta,-2}$ with $T=0$ denoting the intervention period in this application, but the results are robust (of similar magnitudes of reduction) across different specifications of the two pre-intervention periods. These magnitudes are non-negligible and underscore the importance and benefits of matching prior to a DiD analysis. Moreover, the pattern aligns with our theoretical results in Section~\ref{section: revisiting bias} and is consistent with the finding in \cite{ham2024benefits}. Second, the second row further illustrates these benefits. In particular, it shows that the variance cost from using a smaller matched sample is essentially negligible ($\widehat{\Gamma}_{x} = -0.00027$). 
Even though the treated sample only comprises about one third of the total sample, the efficiency gains from incorporating additional control observations are not sufficient to outweigh the benefits of matching. Taken together with the first-row results, this pattern alludes substantial violations of the parallel trends assumption, under which even the sample size tradeoff becomes somewhat irrelevant. We also observe a slightly smaller variance reduction when matching on both covariates and pre-treatment outcomes relative to matching on covariates alone, consistent with our theoretical result in Lemma~\ref{Lemma 3.3}. 

Unfortunately, in this empirical study, we are unable to obtain informative results on the estimated MSE reductions by Proposition \ref{Prop: Estimation Strategies for Relative MSE Reduction under Sign Conditions}. Assuming that the treatment effect is positive (as the original authors had this prior guess and also reported $\hat{\tau} =0.102$), we estimate the relative bias for each pair, which turns out to be negative in both cases. Consequently, we fall into the setting illustrated in Figure \ref{Figure: informative scenario 2}, where only a lower bound can be estimated. However, the resulting lower bounds for both pairs are below zero, preventing us from drawing definitive conclusions for either comparison because the sign of actual reduction may then be both positive or negative. Nevertheless, given the magnitude of reductions shown in first two rows of Table \ref{tab:application}, we believe that the findings on the estimated relative absolute bias reduction and estimated variance reduction already provide sufficiently strong evidence supporting the authors’ decision to match on both the covariates $\mathbf{X}$ and the pre-treatment outcome.

\section{Discussion}
\label{section: Discussion}
\noindent We investigate the variance and mean squared error of matching prior to difference-in-differences (DiD) analysis. Our findings suggest that matching on observed covariates may reduce variance relative to classic unmatched DiD, and further matching on pre-treatment outcomes always reduces variance compared to matching on observed covariates alone DiD estimator. These results contrast with previous findings in the bias results, underscoring the importance of considering MSE as an additional metric when selecting an optimal estimator. We propose new estimation strategies and more comprehensive practical guidelines for determining when matching is recommended. As an illustration, we apply them to the Zhihu Live dataset from \cite{wang2022monetary} to demonstrate the trade-offs and further show that their decision of matching on both covariates and pre-treatment outcomes is warranted.

However, our study has several limitations. Firstly, our model assumes a linear structural equation model with group-invariant covariates and time-varying coefficients, limiting our insights into scenarios with time-varying covariates. Also, our guideline also relies on these assumptions such as the linear structural equation model. We leave the extension of our results to non-linear models to future work. Despite the aforementioned limitations, we believe our framing loses little generality compared to the time-varying contexts. For example, 
in the simple pre-post setting, let \((\theta_{i,t}, X_{i,t})\) evolve according to an AR(1) process, i.e., $X_{i,t+1} = \rho_X X_{i,t} + \kappa^X_{i,t},$ with $|\rho_X| \leq 1$, where \(\kappa^X_{i,t}\) is an independent Gaussian noise error, and \(\theta_{i,t+1}\) is defined similarly. Then our model in Equation \eqref{eq:LSEM} can be reparameterized to incorporate these time-varying covariates as follows:
\[
Y_{i,0}(0) = \beta_{0,0} + \beta_{\theta,0} \theta_{i,0} + \beta_{x,0} X_{i,0} + \varepsilon_{i,0},
\]
\[
Y_{i,1}(0) = \beta_{0,1} + \beta_{\theta,1} \theta_{i,1} + \beta_{x,1} X_{i,1} + \varepsilon_{i,1} = \beta_{0,1} + \tilde{\beta}_{\theta,1} \theta_{i,0} + \tilde{\beta}_{x,1} X_{i,0} + \tilde{\varepsilon}_{i,1},
\]
where $\tilde{\beta}_{x,1} = \rho_X \beta_{x,1}$, and $\tilde{\beta}_{\theta,1}$ is defined similarly, and $\tilde{\varepsilon}_{i,1} = \varepsilon_{i,1} + \beta_{x,1} \kappa^X_{i,1} + \beta_{\theta,1} \kappa^\theta_{i,1}$. This suggests that Equation \eqref{eq:LSEM} is mathematically equivalent to the more complicated time-varying confounders.

Secondly, in our setting, we do not consider the case where researchers may not want to match on all available covariates after balance checks and it thus remains an open and practically relevant question on whether one should consider discarding certain matched covariates if the quality of the match is poor. Lastly, our framework focuses on settings with multiple pre-treatment time points. We believe it is interesting to explore scenarios where certain time periods are particularly influential so that one may perform matching on certain pre-treatment periods rather than all, as well as to extend the analysis to include multiple post-treatment periods, which would allow for a dynamic regime of treatment effects.

\bibliography{ref}

\appendix

\clearpage

\section*{Appendix}
\addcontentsline{toc}{section}{Appendix}

\section{Proof of Theorem \ref{Thm: Consistency of DiD and Matching DiD Estimators}}\label{Appendix A}
\subsection{Consistency of the Classic Unmatched DiD Estimator} \label{Appendix A.1}
\begin{proof} 
We show that $\hat{\tau}_{\text{DiD}} \xrightarrow{p} \mathbb{E}\left[\hat{\tau}_{\text{DiD}}\right]$. This is a standard diff-in-means type of consistency proof with the only difference that our potential outcomes $Y_{i,t}(z)$ are now dependent on $Z_i$ for values $z = 0,1$. Therefore, it suffices to show a sub-case among four and the rest follows. In particular, we have 
\begin{align*}
\frac{1}{n_1} \sum_{i=1}^{n} Y_{i,T} Z_i &=  \frac{n}{n_1} \cdot \frac{1}{n} \sum_{i=1}^{n} Y_{i,T} Z_i   \\
& \xrightarrow{p} \frac{\mathbb{E}(Y_{i,T} Z_i)}{\mathbb{P} (Z_i=1)}  \\
&=  \frac{\mathbb{E}(Y_{i,T} Z_i \mid Z_i =1) \mathbb{P} (Z_i=1) + \mathbb{E}(Y_{i,T} Z_i \mid Z_i =0) \mathbb{P} (Z_i=0) }{\mathbb{P} (Z_i=1)} \\
&= \mathbb{E} (Y_{i,T} \mid Z_i = 1),
\end{align*}  
where the second line follows by the weak law of large numbers (WLLN) and Slutsky's theorem with Assumption \ref{Assumption:boundedness}.
\end{proof}

\subsection{Consistency of Matching DiD Estimators} \label{Appendix A.2}
\begin{proof}
It suffices to show $\hat{\tau}_{\text{DiD}}^{\boldsymbol{X},Y_{0}}   \xrightarrow{p} \mathbb{E}\left[\hat{\tau}_{\text{DiD}}^{\boldsymbol{X},Y_0}\right]$ as the estimator depends only on post-treatment differences, reducing the problem to a cross-sectional formulation with time index $T$. The result for  $\hat{\tau}_{\text{DiD}}^{\boldsymbol{X}}  \xrightarrow{p} \mathbb{E}\left[\hat{\tau}^{\boldsymbol{X}}_{\text{DiD}} \right]$  follows similarly by establishing consistency on the pre-treatment period $T-1$.

For $z \in \{0,1\}$, define $\mu_z(\boldsymbol{w}) := \E[Y_{T} \mid \boldsymbol{W} = \boldsymbol{w}, Z = z]$, where in the case of $\hat{\tau}_{\text{DiD}}^{\boldsymbol{X},Y_{0}}$, $\boldsymbol{W} = (\boldsymbol{X}, Y_{T-1})$. Noticing that $\frac{1}{n_1} \sum_{i=1}^{n} Z_i \mathbb{E}\left[\hat{\tau}_{\text{DiD}}^{\boldsymbol{X},Y_0}\right] = \mathbb{E}\left[\hat{\tau}_{\text{DiD}}^{\boldsymbol{X},Y_0}\right]$, we could rewrite $\hat{\tau}_{\text{DiD}}^{\boldsymbol{X},Y_{0}}  - \mathbb{E}\left[\hat{\tau}_{\text{DiD}}^{\boldsymbol{X},Y_0}\right] = D_n + R_n$, \text{ where}
\[
D_n = \frac{1}{n_1} \sum_{i=1}^{n} \left(\mu_1(\boldsymbol{W}_i) - \mu_0(\boldsymbol{W}_i) - \mathbb{E}\left[\hat{\tau}_{\text{DiD}}^{\boldsymbol{X},Y_0}\right] \right) Z_i
+ \frac{1}{n_1} \sum_{i=1}^{n} \left( \left(Y_{i,T} - \mu_1(\boldsymbol{W}_i) \right) - \frac{1}{M} \sum_{j \in \mathcal{J}_M^{\boldsymbol{X},Y_{0}}(i)} \left(Y_{j,T} - \mu_0(\boldsymbol{W}_j) \right) \right)  Z_i
\]
\[
\text{and}
\]
\[
R_n = \frac{1}{n_1} \sum_{i=1}^{n} \left( \mu_0(\boldsymbol{W}_i) - \frac{1}{M} \sum_{j \in \mathcal{J}_M^{\boldsymbol{X},Y_{0}}(i)} \mu_0(\boldsymbol{W}_j) \right) Z_i.
\] 
The term $R_n$ is the conditional bias of matching estimators, vanishing in the limit under the regularity conditions such as ones introduced in Section 4 in \cite{abadie2012martingale} (compatible with our matching procedure $\mathcal{M}$ under Definition \ref{Matching: near-perfect one-to-one matching without replacement}). The rest of the proof is analogous to \cite{abadie2012martingale}'s Section 3 for showing $D_n$ is a martingale array with respect to a certain filtration and applying the associated WLLN to conclude $D_n \xrightarrow{p} 0$.
\end{proof}

\section{Proofs of Theoretical Results in Section \ref{section: univariate results}}\label{Appendix B}

\subsection{Proof of Theorem \ref{Thm 3.1}} \label{Appendix B.1}
\noindent Before presenting main variance results in \ref{Appendix B.Main results}, we present the first-order approximation used throughout in which we need to account for the randomness issue in the denominator of $n_1$ when evaluating estimators at finite samples.

\subsubsection{First-order Sampling Approximations} \label{Appendix B.First-order Sampling Approximations}
\noindent We first present first-order variance approximations and explain how we omit the $o_p(n^{-1})$ remainder term. For example, we state that 
$$\text{Var}\left(\frac{1}{n_1} \sum_{i=1}^{n} Y_{i,T} Z_i\right)  := \frac{\text{Var}[Y_{i,T}(1) \mid Z_i =1]}{n_1},$$ 
rather than $$\text{Var}\left(\frac{1}{n_1}  \sum_{i=1}^{n} Y_{i,T} Z_i\right)  = \frac{\text{Var}[Y_{i,T}(1) \mid Z_i =1]}{n\pi} + o_p\left(n^{-1}\right),$$
where the plug-in error is derived by denoting 
$\text{Var}[Y_{i,T}(1) \mid Z_i =1] \triangleq \sigma_1^2 $ and noticing 
\[
\begin{aligned}
\frac{\sigma_1^2}{n \pi} - \frac{\sigma_1^2}{n p}
= \frac{\sigma_1^2}{n} \left( \frac{1}{\pi} - \frac{1}{p} \right)
= \frac{\sigma_1^2}{n} \cdot \frac{p - \pi}{p \pi}  = \frac{\sigma_1^2}{n} \cdot O_p(n^{-1/2}) \cdot O_p(1) = O_p(n^{-3/2}),
\end{aligned}
\]
since $\pi -p =  O_p(n^{-1/2})$ by central limit theorem (CLT) for sample proportions, and $\frac{1}{p\pi} = O_p(1)$ because we assume that $\pi,p$ stays bounded away from zero; and we have $n^{-3/2} = o(n^{-1})$. 

We then present the first-order approximation of bias. For any classic unmatched or matching DiD estimators $\hat{\tau}_n$ introduced in Section \ref{section:framework:estimators}, by Appendix \ref{Appendix A}, we have $\hat{\tau}_n \xrightarrow{p} \mathbb{E}[\hat{\tau}]$. By i.i.d. property in Assumption \ref{Assumption: RDGP}, and boundedness in Assumption \ref{Assumption: errors} and \ref{Assumption: variables}, the uniform integrability conditions hold so that we also have $\hat{\tau}_n \xrightarrow{L^1} \mathbb{E}[\hat{\tau}]
\;\; \Longleftrightarrow \;\;
\mathbb{E}[|\hat{\tau}_n - \mathbb{E}[\hat{\tau}]|] \to 0
\;\; \Longrightarrow \;\;
\mathbb{E}[\hat{\tau}_n] \to \mathbb{E}[\hat{\tau}].$ Then our finite sample bias could be rewritten as:
$$
\text{Bias}_n\left( \hat{\tau} \right) := \underbrace{\mathbb{E}\left[\hat{\tau}_n \right] - \tau}_{\text{finite-sample bias}} = \underbrace{\mathbb{E}\left[\hat{\tau} \right] - \tau}_{\text{limiting bias}} +  \underbrace{\mathbb{E}\left[\hat{\tau}_n \right] - \mathbb{E}\left[\hat{\tau} \right]}_{r_n=o(1)}.
$$ 
We thus only present the limiting bias as our bias result and omit the $o(1)$ remainder term. 

\clearpage

\subsubsection{Proofs of Preliminary Lemmas} \label{Appendix: Proofs of Preliminary Lemmas}
\begin{proof}[Proof of Lemma \ref{Lemma: Asymptotic Variance of DiD Estimator}]
We first remark that our linear equation structural model in Equation~\eqref{eq:LSEM} implies the no anticipation assumption. The proof then follows by standard derivations of CLT results for DiM (e.g., see Theorem 1.2 in \cite{wager2020stats}) but with two periods, i.e., a canonical 2-period DiD CLT. We thus omit the details and the only difference between the canonical results and ours is the additional conditioning required to account for group dependence arising from latent confounders. In particular, let 
\[
\Delta Y(d) := Y_{T}(d) - Y_{T-1}(d), \quad d \in \{0,1\}.
\]
Then 
\[
\operatorname{Avar}(\hat{\tau}_{\text{DiD}}) = \frac{\Var(\Delta Y(1) \mid Z =1)}{p} + \frac{\Var(\Delta Y(0) \mid Z =0)}{1-p},
\]
and expanding the terms leads to the desired result.    
\end{proof}

\begin{lemma}[Asymptotic Variance of Matching DiD Estimator(s)] \label{Lemma: Asymptotic Variance of Mathing DiD Estimator(s)}
Suppose the matching procedure $\mathcal{M}$ satisfies Definition \ref{Matching: near-perfect one-to-one matching without replacement}, and some regularity conditions (such as in Theorem 1 of \cite{abadie2012martingale}). Then we have
$\sqrt{n_1} (\hat{\tau}_{\text{DiD}}^{\boldsymbol{X},Y_0} - \mathbb{E}\left[\hat{\tau}_{\text{DiD}}^{\boldsymbol{X},Y_0}\right]) \xrightarrow{d} \mathcal{N}(0,\operatorname{Avar}(\hat{\tau}_{\text{DiD}}^{\boldsymbol{X},Y_0})),$
where $\mathbb{E}\left[\hat{\tau}_{\text{DiD}}^{\boldsymbol{X},Y_0}\right]$ is defined in Theorem \ref{Thm: Consistency of DiD and Matching DiD Estimators}, and with slight abuse of notation,
\begin{align*}
\operatorname{Avar}(\hat{\tau}_{\text{DiD}}^{\boldsymbol{X},Y_0}) &= \operatorname{Avar}\left(\frac{1}{n_1}\sum_{i=1}^{n}  Y_{i,T} Z_i \right) + \operatorname{Avar}\left(\frac{1}{n_1}\sum_{i=1}^{n}  \left( \frac{1}{M} \sum_{j \in \mathcal{J}_M^{\boldsymbol{X},Y_{0}}(i)} Y_{j,T} \right)Z_i\right)  \\
&- 2\operatorname{Acov}\left(\frac{1}{n_1}\sum_{i=1}^{n}  Y_{i,T} Z_i, \quad \frac{1}{n_1}\sum_{i=1}^{n}  \left( \frac{1}{M} \sum_{j \in \mathcal{J}_M^{\boldsymbol{X},Y_{0}}(i)} Y_{j,T} \right)Z_i \right).
\end{align*}    
\end{lemma}

\begin{proof}
For $z \in \{0,1\}$, define $\mu_z(\boldsymbol{w}) := \E[Y_{T} \mid \boldsymbol{W} = \boldsymbol{w}, Z = z]$, 
and $\sigma_z^2(\boldsymbol{w}) = \mathrm{Var}(Y_{T} \mid \boldsymbol{W} = \boldsymbol{w}, Z = z)$, where in the case of $\hat{\tau}_{\text{DiD}}^{\boldsymbol{X},Y_{0}}$, $\boldsymbol{W} = (\boldsymbol{X}, Y_{T-1})$. Then following Theorem 1 of \cite{abadie2012martingale}, we have 
$\operatorname{Avar}(\hat{\tau}_{\text{DiD}}^{\boldsymbol{X},Y_0}) = \E\left[\left(\mu_1(\boldsymbol{W}) - \mu_0(\boldsymbol{W}) - \mathbb{E}\left[\hat{\tau}_{\text{DiD}}^{\boldsymbol{X},Y_0}\right]\right)^2 \mid Z = 1\right] 
+ \E\left[\sigma_1^2(\boldsymbol{W}) + \sigma_0^2(\boldsymbol{W}) \,\middle|\, Z = 1\right]$ because with one-to-one match, we have $M=1$. We then could rewrite 
\begin{align} \label{Eq: Avar of Matching Estimator}
\operatorname{Avar}(\hat{\tau}_{\text{DiD}}^{\boldsymbol{X},Y_0}) 
&= \Var[\mu_1(\boldsymbol{W}) - \mu_0(\boldsymbol{W}) \mid Z =1] +\E\left[\sigma_1^2(\boldsymbol{W}) + \sigma_0^2(\boldsymbol{W}) \,\middle|\, Z = 1\right].
\end{align}
This follows because
\begin{align*}
&\Var\left[\mu_1(\boldsymbol{W}) - \mu_0(\boldsymbol{W}) \mid Z =1\right]  = \Var\left[\mu_1(\boldsymbol{W}) - \mu_0(\boldsymbol{W}) - \mathbb{E}\left[\hat{\tau}_{\text{DiD}}^{\boldsymbol{X},Y_0}\right] \mid Z =1 \right]    \\
&= \E\left[\left(\mu_1(\boldsymbol{W}) - \mu_0(\boldsymbol{W}) - \mathbb{E}\left[\hat{\tau}_{\text{DiD}}^{\boldsymbol{X},Y_0}\right]\right)^2 \mid Z = 1\right] - \left(\E\left[\mu_1(\boldsymbol{W}) - \mu_0(\boldsymbol{W}) - \mathbb{E}\left[\hat{\tau}_{\text{DiD}}^{\boldsymbol{X},Y_0}\right] \mid Z = 1\right]\right)^2,
\end{align*}
in which the second term is zero because of law of total expectation.

Inspecting Equation~\eqref{Eq: Avar of Matching Estimator}, this is exactly an application of law of total variance and law of total covariance to three terms in the limit. In particular, rearranging it leads to 
\begin{align*}
\operatorname{Avar}(\hat{\tau}_{\text{DiD}}^{\boldsymbol{X},Y_0})  &= \underbrace{\left(\Var[\mu_1(\boldsymbol{W}) \mid Z =1] + \E[\sigma_1^2(\boldsymbol{W}) \mid Z = 1]\right)}_{\operatorname{Avar} \left( \cdot \right)} + \underbrace{\left(\Var[\mu_0(\boldsymbol{W}) \mid Z =1] + \E[\sigma_0^2(\boldsymbol{W}) \mid Z = 1]\right)}_{\operatorname{Avar} \left( \cdot \right)} \\
& - 2\underbrace{\mathrm{Cov}[\mu_1(\boldsymbol{W}), \mu_0(\boldsymbol{W}) \mid Z =1]}_{\operatorname{Acov} \left( \cdot, \cdot \right)}.
\end{align*}
\end{proof}

\clearpage

\subsubsection{Reparametrization of Lipschitz Continuous Conditional Moments in Remark \ref{Remark: Structures on Lipschitz continuous functions}} \label{Appendix: Reparametrization of Lipschitz continuous moments}
\noindent We start the section by first stating a useful Lemma with respect to the properties of conditional normality. The proof is omitted as it is a well-known result. 
\begin{lemma} \label{Lemma: Conditional Normality Formula}
Suppose random variables $(X,Y)$ jointly follow a Gaussian distribution: 
\[
\begin{pmatrix}
X \\
Y
\end{pmatrix}
\sim \mathcal{N}\!\left(
\begin{pmatrix}
\mu_1 \\
\mu_2
\end{pmatrix},
\begin{pmatrix}
V_{11} & V_{12} \\
V_{21} & V_{22}
\end{pmatrix}
\right),
\]
where \(X \in \mathbb{R}^p\), \(Y \in \mathbb{R}^q\), and assume that \(V_{11}\) is invertible.

Then the conditional distribution of \(Y\) given \(X = x\) is
\[
Y \mid X = x \sim \mathcal{N}\!\left(
\mu_2 + V_{21} V_{11}^{-1} (x - \mu_1), \;
V_{22} - V_{21} V_{11}^{-1} V_{12}
\right).
\]
In particular,
\[
\mathbb{E}[Y \mid X]
=
\mu_2 + V_{21} V_{11}^{-1} (X - \mu_1),
\]
and
\[
\operatorname{Var}(Y \mid X)
=
V_{22} - V_{21} V_{11}^{-1} V_{12}.
\]
\end{lemma}
\noindent For the univariate case ($X_i, \theta_i \in \mathbf{R}$), \cite{ham2024benefits} assumes normality DGP such that 
\begin{align}  \label{Eq: Joint asym normality}
\begin{pmatrix}
\theta_i   \\
X_i  
\end{pmatrix} \mid Z_i = z \sim \mathcal{N}\left(
\begin{pmatrix}
\mu_{\theta, z} \\
\mu_{x, z}
\end{pmatrix},
\begin{pmatrix}
\sigma_{\theta}^2 & \rho \sigma_{\theta} \sigma_{x} \\
\rho \sigma_{\theta} \sigma_{x} & \sigma_{x}^2
\end{pmatrix}
\right), \quad z = 0, 1.    
\end{align}
\noindent Now we are ready to prove the Reparametrization in Remark \ref{Remark: Structures on Lipschitz continuous functions}.

\begin{proof}[Proof of the Remark \ref{Remark: Structures on Lipschitz continuous functions}]
It suffices to prove for $z =0$ case. Applying Lemma \ref{Lemma: Conditional Normality Formula} to Equation~\eqref{Eq: Joint asym normality} with univariate notions directly yields forms of $\mathbb{E}\left[\theta_i \mid X_i=x, Z_i=0 \right]$ and $ \Var\left( \theta_i \mid X_i=x, Z_i=0 \right)$. For conditional moments further conditioning on the pre-treatment outcome, we have 
\[
\begin{pmatrix}
\theta_i \\
X_i \\
Y_{i,0}  
\end{pmatrix}
\Big| Z_i = 0 \sim \mathcal{N}
\left(
\begin{pmatrix}
\mu_{\theta,0} \\
\mu_{x,0} \\
\beta_{\theta,0} \mu_{\theta,0} + \beta_{x,0} \mu_{x,0}
\end{pmatrix}, 
\Sigma_0
\right),
\]
where 
\[
\Sigma_0 = 
\begin{pmatrix}
\sigma_{\theta}^2 & \sigma_{\theta} \sigma_{x} \rho & \beta_{\theta,0} \sigma_{\theta}^2 + \beta_{x,0} \sigma_{\theta} \sigma_{x} \rho \\
\sigma_{\theta} \sigma_{x} \rho & \sigma_{x}^2 & \beta_{x,0} \sigma_{x}^2 + \beta_{\theta,0} \sigma_{\theta} \sigma_{x} \rho \\
\beta_{\theta,0} \sigma_{\theta}^2 + \beta_{x,0} \sigma_{\theta} \sigma_{x} \rho & \beta_{x,0} \sigma_{x}^2 + \beta_{\theta,0} \sigma_{\theta} \sigma_{x} \rho & (\beta_{\theta,0})^2 \sigma_{\theta}^2 + (\beta_{x,0})^2 \sigma_{x}^2 + \sigma_{E}^2 + 2 \beta_{\theta,0} \beta_{x,0} \sigma_{\theta} \sigma_{x} \rho
\end{pmatrix}
\]
Then applying Lemma \ref{Lemma: Conditional Normality Formula} carefully with the block matrix $(X_i,Y_{i,0})$,
\[
\begin{aligned}
&\mathbb{E}\left[\theta_i \mid X_i=x, Y_{i,0} = y, Z_i=0 \right] \\
&= \mu_{\theta,0}
+ \left(
\sigma_{\theta} \sigma_{x} \rho
\quad
\beta_{\theta,0}\sigma_\theta^2
+ \beta_{x,0}\sigma_\theta\sigma_x\rho
\right) \\
& 
\begin{pmatrix}
\sigma_{x}^2 & \beta_{x,0} \sigma_{x}^2 + \beta_{\theta,0} \sigma_{\theta} \sigma_{x} \rho \\
\beta_{x,0} \sigma_{x}^2 + \beta_{\theta,0} \sigma_{\theta} \sigma_{x} \rho & (\beta_{\theta,0})^2 \sigma_{\theta}^2 + (\beta_{x,0})^2 \sigma_{x}^2 + \sigma_{E}^2 + 2 \beta_{\theta,0} \beta_{x,0} \sigma_{\theta} \sigma_{x} \rho
\end{pmatrix}^{-1} 
\begin{pmatrix}
x - \mu_{x,0} \\
y - \beta_{\theta,0}\mu_{\theta,0} - \beta_{x,0}\mu_{x,0}
\end{pmatrix} \\ 
&=  \mu_{\theta,0}
+ \frac{1}{\det}
\begin{pmatrix}
m_1 & m_2
\end{pmatrix}
\begin{pmatrix}
x - \mu_{x,0} \\
y - \beta_{\theta,0}\mu_{\theta,0} - \beta_{x,0}\mu_{x,0} 
\end{pmatrix} \\
&= \mu_{\theta,0}
+ \frac{m_1}{\det}(x - \mu_{x,0})
+ \frac{m_2}{\det}\left(y - \beta_{\theta,0}\mu_{\theta,0} - \beta_{x,0}\mu_{x,0}\right);
\end{aligned}
\]
and 
\[
\begin{aligned}
&\text{Var}\left( \theta_i \mid X_i=x, Z_i=0, Y_{i,0} = y \right) \\
&= \sigma_{\theta}^2 - \left(\sigma_{\theta} \sigma_{x} \rho \quad \beta_{\theta,0} \sigma_{\theta}^2 + \beta_{x,0} \sigma_{\theta} \sigma_{x} \rho  \right) \\
& 
\begin{pmatrix}
\sigma_{x}^2 & \beta_{x,0} \sigma_{x}^2 + \beta_{\theta,0} \sigma_{\theta} \sigma_{x} \rho \\
\beta_{x,0} \sigma_{x}^2 + \beta_{\theta,0} \sigma_{\theta} \sigma_{x} \rho & (\beta_{\theta,0})^2 \sigma_{\theta}^2 + (\beta_{x,0})^2 \sigma_{x}^2 + \sigma_{E}^2 + 2 \beta_{\theta,0} \beta_{x,0} \sigma_{\theta} \sigma_{x} \rho
\end{pmatrix}^{-1} 
\begin{pmatrix}
\sigma_{\theta} \sigma_{x} \rho \\
\beta_{\theta,0} \sigma_{\theta}^2 + \beta_{x,0} \sigma_{\theta} \sigma_{x} \rho 
\end{pmatrix} \\ 
&= \sigma_{\theta}^2 -\frac{1}{\text{det}} \left(m_1 \quad m_2 \right)  \begin{pmatrix}
\sigma_{\theta} \sigma_{x} \rho \\
\beta_{\theta,0} \sigma_{\theta}^2 + \beta_{x,0} \sigma_{\theta} \sigma_{x} \rho 
\end{pmatrix} \\
&= \sigma_{\theta}^2 -\frac{m_1}{\text{det}} \left(\sigma_{\theta} \sigma_{x} \rho \right) - \frac{m_2}{\text{det}} \left(\beta_{\theta,0} \sigma_{\theta}^2 + \beta_{x,0} \sigma_{\theta} \sigma_{x} \rho  \right),
\end{aligned}
\]
where we define the following to avoid long algebraic expressions:
\begin{align*}
m_1 &= -\beta_{x,0}\beta_{\theta,0}\sigma_{x}^2\sigma_{\theta} ^2(1-\rho^2) + \rho\sigma_{\theta} \sigma_{x}\sigma_E^2, \ m_2= \beta_{\theta,0}\sigma_{\theta}^2\sigma_{x}^2(1-\rho^2), \ \text{det} = \sigma_{x}^2(\beta_{\theta,0}^2\sigma_{\theta}^2(1-\rho^2) + \sigma_E^2).
\end{align*}
\end{proof}

\noindent For the rest of Appendix \ref{Appendix B}, we will assume the same corresponding Lipschitz constants for these conditional moments described in Remark \ref{Remark: Structures on Lipschitz continuous functions} (including the generalized multivariate analysis; see Appendix~\ref{Appendix C.1} for details).

\subsubsection{Main results} \label{Appendix B.Main results}
\noindent  We only present the asymptotic variance for $\hat{\tau}_{\text{DiD}}^{\boldsymbol{X},Y_0}$, but the proof also applies to $\hat{\tau}_{\text{DiD}}^{\boldsymbol{X}}$ for similar arguments as discussed in Appendix \ref{Appendix A.2}. As we focus exclusively on the univariate results in Theorem \ref{Thm 3.1}, we introduce simplified notation for clarity. First, we index the post-treatment period with subscript $T=1$ and the pre-treatment period with subscript $T=0$ throughout the rest of this section. Second,  according to Equation \eqref{Eq DiD} and Equation \eqref{Eq: M-DiD}, we could denote estimators (shifting the order of difference taken for the DiD estimator for better comparison here) as:
\begin{align*}
\hat{\tau}_{\text{DiD}} &:= \left( \bar{Y}_{1}^T - \bar{Y}_{1}^C  \right) - \left( \bar{Y}_{0}^T- \bar{Y}_{0}^C \right) \\
\hat{\tau}_{\text{DiD}^{x}} &:= \left( \bar{Y}_{1}^T - \bar{Y}_{1}^{MT} \right) - \left( \bar{Y}_{0}^{T} - \bar{Y}_{0}^{MT} \right) \\
\hat{\tau}_{\text{DiD}^{X,Y_0}} &:=  \bar{Y}_{1}^T -  \bar{Y}_{1}^{MT,{B}}
\end{align*}
where we we use the subscripts $1,0$ to denote the period; the superscript $T$ to denote the treated group, the superscript $C$ for the control group, the superscript $MT$ for the matched control group (the group of matched control units corresponding to the treated group) if matching is only on observed covariates, the superscript $MT,B$ for the matched control group with matching on both observed covariates and the pre-treatment outcome.

Third, since we are calculating variance, we also drop constants whenever necessary in the proof. Now we are ready to present proofs of the main results.

\clearpage
 
\begin{proof}[Proof of $\Var\left( \hat{\tau}_{\text{DiD}} \right)$]

We have 
\begin{align*} 
\text{Var} \left( \hat{\tau}_{\text{DiD}} \right) &= 
\left[ \text{Var} \left( \bar{Y}_{1}^T \right) + \text{Var} \left( \bar{Y}_{1}^C \right) - 2 \text{Cov} \left( \bar{Y}_{1}^T, \bar{Y}_{1}^C \right) \right] \\
&+ \left[ \text{Var} \left( \bar{Y}_{0}^T \right) + \text{Var} \left( \bar{Y}_{0}^C \right) - 2 \text{Cov} \left( \bar{Y}_{0}^T, \bar{Y}_{0}^C \right) \right] \\
&- 2 \left[ \text{Cov} \left( \bar{Y}_{1}^T, \bar{Y}_{0}^T \right) - \text{Cov} \left( \bar{Y}_{1}^T, \bar{Y}_{0}^C \right) - \text{Cov} \left( \bar{Y}_{1}^C, \bar{Y}_{0}^T \right) + \text{Cov} \left( \bar{Y}_{1}^C, \bar{Y}_{0}^C \right) \right]
\end{align*}
As discussed in Section \ref{section: univariate results}, instead of reporting $\operatorname{Avar}(\bar{Y_1}^T) = \Var(Y(1) \mid Z = 1)/p$ as in Lemma \ref{Lemma: Asymptotic Variance of DiD Estimator}, we report its finite sample first-order approximation via $\Var(\bar{Y_1}^T) = \frac{1}{n_1} \Var(Y(1) \mid Z = 1)$ as illustrated in Appendix \ref{Appendix B.First-order Sampling Approximations}. According to Equation \eqref{eq:LSEM} and notations in Remark \ref{Remark: notations on univariate variance}, with some algebra, under Assumption \ref{Assumption: errors}, we have the following variance results: 
\[
\begin{aligned}
\text{Var}(\bar{Y_1}^T) &= \frac{1}{n_1} \left\{ \beta_{\theta,1}^2 \sigma_{\theta}^2 +  \beta_{x,1}^2 \sigma_{x}^2  + \sigma_{E}^2  + 2 \beta_{\theta,1} \beta_{x,1} \rho  \sigma_{x}\sigma_{\theta} \right\} \\
\text{Var}(\bar{Y_0}^T)  &= \frac{1}{n_1} \left\{ \beta_{\theta,0}^2 \sigma_{\theta}^2 +  \beta_{x,0}^2 \sigma_{x}^2  + \sigma_{E}^2  + 2\beta_{\theta,0} \beta_{x,0} \rho  \sigma_{x}\sigma_{\theta}  \right\} \\
\text{Var}(\bar{Y_1}^C) &= \frac{1}{n_0} \left\{ \beta_{\theta,1}^2 \sigma_{\theta}^2 +  \beta_{x,1}^2 \sigma_{x}^2  + \sigma_{E}^2  + 2 \beta_{\theta,1} \beta_{x,1} \rho  \sigma_{x}\sigma_{\theta}  \right\} \\
\text{Var}(\bar{Y_0}^C)  &= \frac{1}{n_0} \left\{ \beta_{\theta,0}^2 \sigma_{\theta}^2 +  \beta_{x,0}^2 \sigma_{x}^2  + \sigma_{E}^2  + 2 \beta_{\theta,0} \beta_{x,0} \rho  \sigma_{x}\sigma_{\theta}  \right\}  
\end{aligned}
\]
Then as shown in Lemma \ref{Lemma: Asymptotic Variance of DiD Estimator}, some (asymptotic) covariance terms are zero.
Specifically, for any $ t,t' \in \{0,1\}, \  \text{Cov} \left( \bar{Y}_{t}^{T}, \bar{Y}_{t'}^{C} \right) = 0$: This is because by Assumption \ref{Assumption: RDGP} and Assumption \ref{Assumption: errors}, after some algebra of expanding sums, it suffices to only consider the inner covariance term with respect to residual errors when $t = t'$, i.e., $\Cov(\epsilon_{i,t} Z_i,  \epsilon_{i,t} (1-Z_i)) = \mathbb{E}[\epsilon_{i,t} Z_i \cdot \epsilon_{i,t} (1-Z_i)]
- \mathbb{E}\left[\epsilon_{i,t} Z_i\right] \mathbb{E}[\epsilon_{i,t} (1-Z_i)] = 0-0 \cdot 0 = 0.$ 

Thus we only work with the rest covariance terms. Under Assumption \ref{Assumption: errors}, we have 
\[
\begin{aligned}
\text{Cov} \left( \bar{Y}_{1}^T, \bar{Y}_{0}^T \right)  &= \frac{1}{n_1} \underbrace{\text{Cov}\left( \beta_{0,1} + \beta_{\theta,1} \theta_i + \beta_{x,1} X_i + \epsilon_{i,1}, \beta_{0,0} + \beta_{\theta,0} \theta_i + \beta_{x,0} X_i + \epsilon_{i,0}  \mid Z_i = 1 \right)}_{= \  \beta_{\theta,1}\beta_{\theta,0} \text{Var}\left(  \theta_i \mid Z_i = 1 \right) +  \beta_{x,1}\beta_{x,0}\text{Var}\left( X_i \mid Z_i = 1 \right)  + \beta_{\theta,1} \beta_{x,0}  \text{Cov}\left( (\theta_i,X_i) \mid Z_i = 1 \right) +  \beta_{\theta,0} \beta_{x,1}  \text{Cov}\left( (\theta_i,X_i) \mid Z_i = 1 \right)  }  \\ 
&= \frac{1}{n_1} \left\{ \beta_{\theta,1}\beta_{\theta,0} \sigma_{\theta}^2 +  \beta_{x,1}\beta_{x,0} \sigma_{x}^2  + \left(\beta_{\theta,1} \beta_{x,0} + \beta_{\theta,0} \beta_{x,1}  \right) \rho  \sigma_{x}\sigma_{\theta}  \right\} \\
\text{Cov} \left( \bar{Y}_{1}^C, \bar{Y}_{0}^C \right)  &= \frac{1}{n_0} \underbrace{\text{Cov}\left( \beta_{0,1} + \beta_{\theta,1} \theta_i + \beta_{x,1} X_i + \epsilon_{i,1}, \beta_{0,0} + \beta_{\theta,0} \theta_i + \beta_{x,0} X_i + \epsilon_{i,0} \mid Z_i = 0 \right)}_{= \  \beta_{\theta,1}\beta_{\theta,0} \text{Var}\left(  \theta_i \mid Z_i = 0 \right) +  \beta_{x,1}\beta_{x,0} \text{Var}\left( X_i \mid Z_i = 0 \right)  + \beta_{\theta,1} \beta_{x,0}  \text{Cov}\left( (\theta_i,X_i) \mid Z_i = 0 \right) +  \beta_{\theta,0} \beta_{x,1}  \text{Cov}\left( (\theta_i,X_i) \mid Z_i = 0 \right)  }  \\ 
&= \frac{1}{n_0} \left\{ \beta_{\theta,1}\beta_{\theta,0} \sigma_{\theta}^2 +  \beta_{x,1}\beta_{x,0} \sigma_{x}^2  + \left(\beta_{\theta,1} \beta_{x,0} + \beta_{\theta,0} \beta_{x,1}  \right) \rho  \sigma_{x}\sigma_{\theta}  \right\}
\end{aligned}
\]
where since we only define first moments to be group-specific, there aren't any subscripts $z$ for variance and covariance terms. Combining all terms and after some algebra, we complete the proof. 
\end{proof}

\clearpage

\begin{proof}[Proof of $\Var \left( \hat{\tau}_{\text{DiD}^{X,Y_0}} \right)$]
We have $$\text{Var} \left( \hat{\tau}_{\text{DiD}^{X,Y_0}} \right) =  \text{Var} \left( \bar{Y}_{1}^T \right) + \text{Var} \left(  \bar{Y}_{1}^{MT,{B}} \right) - 2 \text{Cov} \left( \bar{Y}_{1}^T,  \bar{Y}_{1}^{MT,{B}} \right),$$ and we have derived $\text{Var}\left( \bar{Y}_{1}^T \right)$, so it remains to derive the other two terms. 

Applying law of total variance as in Lemma \ref{Lemma: Asymptotic Variance of Mathing DiD Estimator(s)} and with proper finite sample first-order approximation scaling as in Appendix \ref{Appendix B.First-order Sampling Approximations}, we have $\text{Var}(\bar{Y_1}^{MT,B})  = \frac{1}{n_1} \left[ \text{\ding{172}} + \text{\ding{173}} \right]$. In particular, further by Remark \ref{Remark: Structures on Lipschitz continuous functions}, we have  
\[
\begin{aligned}
\text{\ding{172}} &= 
\mathbb{E}_{X, Y_0|Z=1} \left[\text{Var}\left( Y_{i,1} \mid Z_i=0, X_i =x, Y_{i,0} = y \right) \right] \\ 
&= \mathbb{E}_{X, Y_0|Z=1} \left[\text{Var}\left( Y_{i,1}(0) \mid  Z_i=0, X_i =x, Y_{i,0} = y \right) \right] \\ 
&= \mathbb{E}_{X,Y_0|Z=1} \left[\text{Var}\left(\beta_{\theta,1} \theta_i + \beta_{x,1} X_i + \epsilon_{i,1} \mid Z_i=0, X_i=x, Y_{i,0} = y \right) \right] \\
&=  \mathbb{E}_{X,Y_0|Z=1} \left[\beta_{\theta,1}^2 \text{Var}\left( \theta_i \mid  Z_i=0, X_i=x, Y_{i,0} = y \right) + \sigma^2_{E}\right] \\
& = \beta_{\theta,1}^2\left(\sigma_{\theta}^2 -\frac{m_1}{\text{det}} \left(\sigma_{\theta} \sigma_{x} \rho \right) - \frac{m_2}{\text{det}} \left(\beta_{\theta,0} \sigma_{\theta}^2 + \beta_{x,0} \sigma_{\theta} \sigma_{x} \rho  \right)\right) + \sigma^2_E\\
& \\ 
\text{\ding{173}} 
&=\text{Var}_{X,Y_0|Z=1} \left(\mathbb{E}\left[ Y_{i,1} \mid   Z_i=0, X_i=x, Y_{i,0} =y \right] \right) \\ 
&=\text{Var}_{X,Y_0|Z=1} \left(\mathbb{E}\left[ Y_{i,1}(0) \mid   Z_i=0, X_i=x, Y_{i,0} =y \right] \right) \\ 
&= \text{Var}_{X,Y_0|Z=1} \left( \beta_{\theta,1} \cdot \mathbb{E}\left[ \theta_i \mid  Z_i=0, X_i=x, Y_{i,0} = y \right] + \beta_{x,1}x \right) \\
&= \text{Var}_{X,Y_0|Z=1} \left(  \beta_{\theta,1} \cdot \left(\frac{m_1}{\text{det}}x +\frac{m_2}{\text{det}}y \right) + \beta_{x,1}x \right) \\ 
&= \left( \beta_{\theta,1}
\frac{m_1}{\text{det}}  + \beta_{x,1} \right)^2 \text{Var}_{X,Y_0|Z=1}(x) + \left( \beta_{\theta,1}  
\frac{m_2}{\text{det}}\right)^2 \text{Var}_{X,Y_0|Z=1}(y) \\
&+ 2 \left( \beta_{\theta,1} 
 \frac{m_1}{\text{det}}  + \beta_{x,1} \right) \left( \beta_{\theta,1}  \frac{m_2}{\text{det}}\right) \text{Cov}_{X,Y_0|Z=1}(x,y)
\\
&= \left( \beta_{\theta,1}\frac{m_1}{\text{det}}  + \beta_{x,1} \right)^2  \sigma^2_{x} + \left( \beta_{\theta,1} \frac{m_2}{\text{det}}\right)^2  \left((\beta_{\theta,0})^2 \sigma_{\theta}^2 + (\beta_{x,0})^2 \sigma_{x}^2 + \sigma_{E}^2 + 2 \beta_{\theta,0} \beta_{x,0} \sigma_{\theta} \sigma_{x} \rho \right) \\
&+ 2 \left(\beta_{\theta,1}  \frac{m_1}{\text{det}} + \beta_{x,1} \right) \left( \beta_{\theta,1} \frac{m_2}{\text{det}}\right) \left(\beta_{x,0} \sigma_{x}^2 + \beta_{\theta,0} \sigma_{\theta} \sigma_{x} \rho \right),
\end{aligned}
\]    
where we apply reparameterized forms of conditional mean and variance for the control group from Remark \ref{Remark: Structures on Lipschitz continuous functions} and repeatedly invoke usages of Assumption \ref{Assumption: errors} for irreducible errors.

If we collect terms such as $\frac{m_2}{\text{det}} \cdot \beta_{\theta,0} = \frac{(\beta_{\theta,0})^2 \sigma_\theta^2 (1 - \rho^2)}{(\beta_{\theta,0})^2 \sigma_\theta^2 (1 - \rho^2) + \sigma_E^2} $, then we obtain a quantity denoted as the reliability term, $r_{\theta|x}$, explained in the main text.

Similarly, applying law of total covariance as in Lemma \ref{Lemma: Asymptotic Variance of Mathing DiD Estimator(s)}, we have 
\begin{align*}
\text{Cov} \left( \bar{Y}_{1}^T,  \bar{Y}_{1}^{MT,{B}} \right) = \frac{1}{n_1} \left[ \operatorname{Cov}_{X,Y_0|Z=1} \left(\mathbb{E}\left[ Y_{i,1} \mid  Z_i = 1, X_{i} = x, Y_{i,0} = y \right], \mathbb{E}\left[ Y_{i,1} \mid Z_i = 0, X_{i} = x, Y_{i,0} = y  \right] \right)   \right].
\end{align*}
Focusing on the inner terms and dropping irrelevant constants, by Remark \ref{Remark: Structures on Lipschitz continuous functions}, we have
\begin{align*}
& \operatorname{Cov}_{X,Y_0|Z=1}\left(\beta_{\theta,1} \cdot \left(\frac{m_1}{\text{det}}x +\frac{m_2}{\text{det}}y \right) + \beta_{x,1}x , \  \beta_{\theta,1} \cdot \left(\frac{m_1}{\text{det}}x +\frac{m_2}{\text{det}}y \right) + \beta_{x,1}x  \right)   \\
&= \left( \beta_{\theta,1}\frac{m_1}{\text{det}}  + \beta_{x,1} \right)^2  \sigma^2_{x} + \left( \beta_{\theta,1} \frac{m_2}{\text{det}}\right)^2  \left((\beta_{\theta,0})^2 \sigma_{\theta}^2 + (\beta_{x,0})^2 \sigma_{x}^2 + \sigma_{E}^2 + 2 \beta_{\theta,0} \beta_{x,0} \sigma_{\theta} \sigma_{x} \rho \right) \\
&+ 2 \left(\beta_{\theta,1}  \frac{m_1}{\text{det}} + \beta_{x,1} \right) \left( \beta_{\theta,1} \frac{m_2}{\text{det}}\right) \left(\beta_{x,0} \sigma_{x}^2 + \beta_{\theta,0} \sigma_{\theta} \sigma_{x} \rho \right)  
\end{align*}
Finally, combining all terms and after messy algebra, we reach to the desired result.
\end{proof}

\clearpage

\begin{proof}[Proof of $\Var \left( \hat{\tau}_{\text{DiD}^{X}} \right)$]
The proof idea is similar to ones used in $\Var \left( \hat{\tau}_{\text{DiD}^{X,Y_0}} \right)$, where we apply reparameterized forms of conditional moments from Remark \ref{Remark: Structures on Lipschitz continuous functions} and repeatedly invoke usages of Assumption \ref{Assumption: errors} for irreducible errors, along with the main Lemma \ref{Lemma: Asymptotic Variance of Mathing DiD Estimator(s)}. Specifically, we have 
\begin{align*}
\text{Var} \left( \hat{\tau}_{\text{DiD}^{x}} \right) &= 
\left[ \text{Var} \left( \bar{Y}_{1}^T \right) + \text{Var} \left( \bar{Y}_{0}^T \right) - 2 \text{Cov} \left( \bar{Y}_{1}^T, \bar{Y}_{0}^T \right) \right] \\
&+ \left[ \text{Var} \left( \bar{Y}_{1}^{MT} \right) + \text{Var} \left( \bar{Y}_{0}^{MT} \right) - 2 \text{Cov} \left( \bar{Y}_{1}^{MT}, \bar{Y}_{0}^{MT} \right) \right] \\
&- 2 \left[ \text{Cov} \left( \bar{Y}_{1}^T, \bar{Y}_{1}^{MT} \right) - \text{Cov} \left( \bar{Y}_{1}^T, \bar{Y}_{0}^{MT} \right) - \text{Cov} \left( \bar{Y}_{0}^T, \bar{Y}_{1}^{MT} \right) + \text{Cov} \left( \bar{Y}_{0}^T, \bar{Y}_{0}^{MT} \right) \right] 
\end{align*}
We have derived the first line of variance results, and the second line of variance results are:
\[
\begin{aligned}
\text{Var}(\bar{Y_1}^{MT})  & = \frac{1}{n_1} \left[ \beta_{\theta,1}^2 (1-\rho^2)\sigma^2_{\theta}  + \sigma^2_{E} + \left( \beta_{\theta,1}  \rho \frac{\sigma_{\theta}}{\sigma_{x}} + \beta_{x,1} \right)^2  \sigma^2_{x} \right] \\
\text{Var}(\bar{Y_0}^{MT}) 
&= \frac{1}{n_1} \left[ \beta_{\theta,0}^2 (1-\rho^2)\sigma^2_{\theta}  + \sigma^2_{E} + \left( \beta_{\theta,0}  \rho \frac{\sigma_{\theta}}{\sigma_{x}} + \beta_{x,0} \right)^2  \sigma^2_{x} \right] \\
\text{Cov} \left( \bar{Y}_{1}^{MT}, \bar{Y}_{0}^{MT} \right) 
&= \frac{1}{n_1} \left[ \beta_{\theta,1}\beta_{\theta,0} \left((1-\rho^2)\sigma^2_{\theta} + \rho^2 \frac{\sigma_{\theta}^2}{\sigma_{x}^2} \sigma^2_{x} \right)+ \left(\beta_{\theta,1}\beta_{x,0} + \beta_{\theta,0}\beta_{x,1}  \right) \rho \frac{\sigma_{\theta}}{\sigma_{x}} \sigma_{x}^2 + \beta_{x,1}\beta_{x,0} \sigma_{x}^2
\right] 
\end{aligned}
\]
Among $\text{Var}(\bar{Y_1}^{MT})$ and $\text{Var}(\bar{Y_0}^{MT})$, it suffices to show $\text{Var}(\bar{Y_1}^{MT}) = \frac{1}{n_1} \left[ \text{\ding{172}} + \text{\ding{173}} \right]$. Following the proof of $\Var \left( \hat{\tau}_{\text{DiD}^{X,Y_0}} \right)$ along with Lemma \ref{Lemma: Asymptotic Variance of Mathing DiD Estimator(s)} but with matching only on observed covariates $X$, by law of total variance, we have  
\[
\begin{aligned}
\text{\ding{172}} &= \mathbb{E}_{X|Z=1} \left[\text{Var}\left( Y_{i,1} \mid Z_i=0, X_i =x \right) \right] \\ 
&= \mathbb{E}_{X|Z=1} \left[\text{Var}\left( Y_{i,1}(0) \mid  Z_i=0, X_i =x \right) \right] \\
&= \mathbb{E}_{X|Z=1} \left[\text{Var}\left(\beta_{\theta,1} \theta_i + \beta_{x,1} X_i + \epsilon_{i,1} \mid Z_i=0, X_i =x  \right) \right] \\
&=  \mathbb{E}_{X|Z=1} \left[\beta_{\theta,1}^2 \text{Var}\left( \theta_i \mid Z_i=0, X_i =x \right) + \sigma^2_{E}\right]  \\
& = \beta_{\theta,1}^2 (1-\rho^2)\sigma^2_{\theta}  + \sigma^2_{E} 
\end{aligned}
\]    
and 
\[
\begin{aligned}
\text{\ding{173}} 
&= \text{Var}_{X|Z=1} \left(\mathbb{E}\left[ Y_{i,1} \mid Z_i = 0, X_i = x \right] \right)  \\
&=\text{Var}_{X|Z=1} \left(\mathbb{E}\left[ Y_{i,1}(0) \mid  Z_i = 0, X_i = x\right] \right) \\ 
&= \text{Var}_{X|Z=1} \left(\beta_{\theta,1} \left( \mu_{\theta,0} + \rho \frac{\sigma_{\theta}}{\sigma_{x}} (x - \mu_{x,0}) \right) + \beta_{x,1}x \right)  \\
&= \left( \beta_{\theta,1}  \rho \frac{\sigma_{\theta}}{\sigma_{x}} + \beta_{x,1} \right)^2 \text{Var}_{X|Z=1}(x) \\
&= \left( \beta_{\theta,1}  \rho \frac{\sigma_{\theta}}{\sigma_{x}} + \beta_{x,1} \right)^2  \sigma^2_{x}.
\end{aligned}
\]    
For $\text{Cov} \left( \bar{Y}_{1}^{MT}, \bar{Y}_{0}^{MT} \right)$, applying law of total covariance, in the limit we have
\[
\begin{aligned}
\text{\ding{172}} 
&= \mathbb{E}_{X|Z=1} \left[\operatorname{Cov}\left( Y_{i,1}, Y_{i,0}\mid X_i = x, Z_i =0 \right) \right] \\
&= \mathbb{E}_{X|Z=1} \left[\operatorname{Cov}\left( Y_{i,1}(0), Y_{i,0}(0)\mid X_i = x, Z_i =0 \right) \right] \\
&= \mathbb{E}_{X|Z=1} \left[\operatorname{Cov}\left( \beta_{\theta,1} \theta_i + \beta_{x,1} X_i + \epsilon_{i,1}, \beta_{\theta,0} \theta_i + \beta_{x,0} X_i + \epsilon_{i,0} \mid  X_i = x, Z_i =0 \right) \right] \\
&=  \mathbb{E}_{X|Z=1} \left[\beta_{\theta,1} \beta_{\theta,0} \operatorname{Var}(\theta_i \mid  X_i = x, Z_i =0) \right]  \\
& = \beta_{\theta,1} \beta_{\theta,0} (1-\rho^2)\sigma^2_{\theta}  \\
& \\ 
\text{\ding{173}} 
&=\operatorname{Cov}_{X|Z=1} \left(\mathbb{E}\left[ Y_{i,1} \mid  X_i = x, Z_i =0 \right], \mathbb{E}\left[ Y_{i,0} \mid  X_i = x, Z_i =0  \right] \right)  \\ 
&=\operatorname{Cov}_{X|Z=1} \left(\mathbb{E}\left[ Y_{i,1}(0) \mid  X_i = x, Z_i =0 \right], \mathbb{E}\left[ Y_{i,0}(0) \mid  X_i = x, Z_i =0  \right] \right) \\ 
&= \operatorname{Cov}_{X|Z=1} \left[ \beta_{\theta,1} \left( \mu_{\theta,0} + \rho \frac{\sigma_{\theta}}{\sigma_{x}} (x - \mu_{x,0}) \right) + \beta_{x,1}x ,\ \beta_{\theta,0} \left( \mu_{\theta,0} + \rho \frac{\sigma_{\theta}}{\sigma_{x}} (x - \mu_{x,0}) \right) + \beta_{x,0}x  \right] \\
&= \beta_{\theta,1} \beta_{\theta,0} \rho^2 \frac{\sigma_{\theta}^2}{\sigma_{x}^2} \operatorname{Var}_{X|Z_i=1}(x)
+ (\beta_{\theta,1} \beta_{x,0} + \beta_{\theta,0} \beta_{x,1}) \rho \frac{\sigma_{\theta}}{\sigma_{x}} \operatorname{Var}_{X|Z_i=1}(x) 
+ \beta_{x,1} \beta_{x,0} \operatorname{Var}_{X|Z_i=1}(x) \\ 
&= \beta_{\theta,1}\beta_{\theta,0}  \rho^2 \frac{\sigma_{\theta}^2}{\sigma_{x}^2} \sigma^2_{x} + \left(\beta_{\theta,1}\beta_{x,0} + \beta_{\theta,0}\beta_{x,1} \right) \rho \frac{\sigma_{\theta}}{\sigma_{x}} \sigma_{x}^2 + \beta_{x,1}\beta_{x,0} \sigma_{x}^2. 
\end{aligned}
\]    
For the rest covariance terms, we have: 
\[
\begin{aligned}
&\text{Cov} \left( \bar{Y}_{1}^{T}, \bar{Y}_{1}^{MT} \right) = \frac{1}{n_1} \left[ \left( \beta_{\theta,1}  \rho \frac{\sigma_{\theta}}{\sigma_{x}} + \beta_{x,1} \right) \left( \beta_{\theta,1}  \rho \frac{\sigma_{\theta}}{\sigma_{x}} + \beta_{x,1} \right)  \sigma^2_{x}
\right] \\
&\text{Cov} \left( \bar{Y}_{0}^{T}, \bar{Y}_{0}^{MT} \right) = \frac{1}{n_1} \left[ \left( \beta_{\theta,0}  \rho \frac{\sigma_{\theta}}{\sigma_{x}} + \beta_{x,0} \right) \left( \beta_{\theta,0}  \rho \frac{\sigma_{\theta}}{\sigma_{x}} + \beta_{x,0} \right)  \sigma^2_{x}
\right] \\
&\text{Cov} \left( \bar{Y}_{1}^{T}, \bar{Y}_{0}^{MT} \right) = \frac{1}{n_1} \left[ \left( \beta_{\theta,1}  \rho \frac{\sigma_{\theta}}{\sigma_{x}} + \beta_{x,1} \right) \left( \beta_{\theta,0}  \rho \frac{\sigma_{\theta}}{\sigma_{x}} + \beta_{x,0} \right)  \sigma^2_{x}
\right]  \\
&\text{Cov} \left( \bar{Y}_{0}^{T}, \bar{Y}_{1}^{MT} \right) = \frac{1}{n_1} \left[ \left( \beta_{\theta,0}  \rho \frac{\sigma_{\theta}}{\sigma_{x}} + \beta_{x,0} \right) \left( \beta_{\theta,1}  \rho \frac{\sigma_{\theta}}{\sigma_{x}} + \beta_{x,1} \right)  \sigma^2_{x}
\right] 
\end{aligned}
\]
We only show the first case and the rest three are similar. Applying the law of total covariance in the limit, for $\text{Cov} \left( \bar{Y}_{1}^{T}, \bar{Y}_{1}^{MT} \right)$, we have 
\[
\begin{aligned}
&=\text{Cov}_{X|Z=1} \left(\mathbb{E}\left[ Y_{i,1} \mid  Z_i=1, X_i=x  \right], \mathbb{E}\left[ Y_{i,1} \mid  Z_i=0, X_i=x \right] \right) \\ 
&= \text{Cov}_{X|Z=1} \left[ \beta_{\theta,1} \left( \mu_{\theta,1} + \rho \frac{\sigma_{\theta}}{\sigma_{x}} (x - \mu_{x,1}) \right) + \beta_{x,1}x ,  \beta_{\theta,1} \left( \mu_{\theta,0} + \rho \frac{\sigma_{\theta}}{\sigma_{x}}(x - \mu_{x,0}) \right) + \beta_{x,1}x  \right] \\
&= \left( \beta_{\theta,1}  \rho \frac{\sigma_{\theta}}{\sigma_{x}} + \beta_{x,1} \right) \left( \beta_{\theta,1}  \rho \frac{\sigma_{\theta}}{\sigma_{x}} + \beta_{x,1} \right)  \sigma^2_{x}
\end{aligned}
\]    
Finally, combining all terms and after messy algebra, we reach to the desired result. 
\end{proof}

\clearpage

\subsection{Proof of Lemma \ref{Lemma 3.3}} \label{Appendix B.2}

\begin{proof}

We want to show $\text{Var} \left( \hat{\tau}_{\text{DiD}^{X}} \right) \geq \text{Var} \left( \hat{\tau}_{\text{DiD}^{X,Y_0}} \right)$, i.e., $$2\sigma_{E}^2 + 
\Delta_{\theta}^2 \sigma^2_{\theta}    -\Delta_{\theta}^2 \rho^2 \sigma_{\theta}^2  \geq \sigma_{E}^2  + \beta_{\theta,1}^2 \sigma_{\theta}^2 \left(1-\rho^2 \right)   \left( 1 - r_{\theta|x} \right)$$ always holds.

If $|\rho| =1$, then $\text{Var} \left( \hat{\tau}_{\text{DiD}^{X}} \right)  \geq \text{Var} \left( \hat{\tau}_{\text{DiD}^{X,Y_0}} \right)$ holds trivially. We thus only prove the more interesting case when $|\rho| \neq 1$. 

Let $|\rho| \neq 1$ and if $\beta_{\theta,1} = 0$, then the inequality trivial holds by noting that $\Delta_{\theta} = \beta_{\theta,1} - \beta_{\theta,0}.$ Then it suffices to consider $\beta_{\theta,1} \neq 0$. Denote $s = \frac{\beta_{\theta,0}}{ \beta_{\theta,1}}$, and rewriting the inequality gives 
\[
\begin{aligned}
& \sigma_E^2 + \Delta_{\theta}^2 \left(1- \rho^2 \right)  \sigma_{\theta}^2 \geq \beta_{\theta,1}^2 \sigma_{\theta}^2 \left(1-\rho^2 \right)   \left( 1 - r_{\theta|x} \right)   \\ 
& \Leftrightarrow \frac{\sigma_E^2}{\left(1-\rho^2 \right) \beta_{\theta,1}^2 \sigma_{\theta}^2} + \frac{\Delta_{\theta}^2}{\beta_{\theta,1}^2} + r_{\theta|x} -1 \geq 0 \\
 & \Leftrightarrow \frac{\beta_{\theta,0}^2}{\beta_{\theta,1}^2} \cdot  \frac{\sigma_E^2}{\left(1-\rho^2 \right) \beta_{\theta,0}^2 \sigma_{\theta}^2} + \frac{\beta_{\theta,1}^2 - 2 \beta_{\theta,1} \beta_{\theta,0} + \beta_{\theta,0}^2 }{\beta_{\theta,1}^2} + r_{\theta|x} -1 \geq 0 \\ 
  & \Leftrightarrow s^2 \left(\frac{1}{r_{\theta|x}} -1  \right) -2s + s^2 + r_{\theta|x} \geq 0  \\ 
  & \Leftrightarrow r_{\theta|x}  + \frac{s^2}{r_{\theta|x}} - 2s \geq 0 \\ 
  & \Leftrightarrow  r_{\theta|x}^2 + s^2 - 2s r_{\theta|x} \geq 0 \quad \text{by} \ r_{\theta|x} \in (0,1) \\ 
  & \Rightarrow ( r_{\theta|x} - s)^2 \geq 0
 \end{aligned}
\]    
The second line  goes through because $|\rho| \neq 1$ and $\beta_{\theta,1} \neq 0$.  The second last line uses the definition of the reliability term, otherwise the sign will 
be flipped. The last line shows that the inequality always holds. 
\end{proof}

\clearpage

\section{Proofs of Theoretical Results in Section \ref{section: multivariate results}}
\noindent We state some lemmas about matrix inversions that will be used throughout this section. These proofs are omitted because they are well established properties in the literature. 

\begin{lemma}[\textbf{Block Matrix Inversion}] \label{Lemma: Block Matrix Inversion}
Suppose we have a block matrix 
\[
\mathbf{P} = \begin{pmatrix}
\boldsymbol{A} & \boldsymbol{B} \\
\boldsymbol{C} & \boldsymbol{D}
\end{pmatrix}
\]
then
\[
\mathbf{P}^{-1} = 
\begin{pmatrix}
\boldsymbol{A}^{-1} + \boldsymbol{A}^{-1}\boldsymbol{B}(\boldsymbol{D} - \boldsymbol{C}\boldsymbol{A}^{-1}\boldsymbol{B})^{-1}\boldsymbol{C}\boldsymbol{A}^{-1} & 
-\boldsymbol{A}^{-1}\boldsymbol{B}(\boldsymbol{D} - \boldsymbol{C}\boldsymbol{A}^{-1}\boldsymbol{B})^{-1} \\
-(\boldsymbol{D} - \boldsymbol{C}\boldsymbol{A}^{-1}\boldsymbol{B})^{-1}\boldsymbol{C}\boldsymbol{A}^{-1} & 
(\boldsymbol{D} - \boldsymbol{C}\boldsymbol{A}^{-1}\boldsymbol{B})^{-1}
\end{pmatrix}
\]
\end{lemma} 

\begin{lemma}[\textbf{Sherman–Morrison formula}] \label{Lemma: Sherman–Morrison formula}
Suppose $A \in \mathbf{R}^{n \times n}$ is an invertible square matrix and $u, v \in \mathbf{R}^n$ are column vectors. Then $A + u v^\top$ is invertible if and only if
$$
1 + v^\top A^{-1} u \neq 0.
$$
In this case,
$$
(A + u v^\top)^{-1} = A^{-1} - \frac{A^{-1} u v^\top A^{-1}}{1 + v^\top A^{-1} u}.
$$
\end{lemma}

\begin{lemma}[\textbf{Woodbury Matrix Identity}] \label{Lemma: Woodbury Matrix Identity}
Let \( A \in \mathbf{R}^{n \times n} \), \( C \in \mathbf{R}^{k \times k} \), \( U \in \mathbf{R}^{n \times k} \), and \( V \in \mathbf{R}^{k \times n} \), with \( A \) and \( C \) invertible. Then the following identity holds:
\[
\left( A + UCV \right)^{-1} = A^{-1} - A^{-1} U \left( C^{-1} + V A^{-1} U \right)^{-1} V A^{-1}.
\]
\end{lemma}

\subsection{Proof of Theorem \ref{Thm 4.1}} \label{Appendix C.1}

\begin{proof}
Assuming vector-valued counterparts of Remark \ref{Remark: Structures on Lipschitz continuous functions}, the results of $\text{Var} \left( \hat{\tau}_{\text{gDiD}}\right)$ and $\text{Var} \left( \hat{\tau}_{\text{gDiD}}^{\boldsymbol{X}}\right)$ could be readily extended using their vector analogues, so the proofs are omitted. For notational simplicity, we rewrite $\text{Var} \left( \hat{\tau}_{\text{gDiD}}^{\boldsymbol{X}}\right)$ as 
\[
\begin{aligned}
\text{Var} \left( \hat{\tau}_{\text{gDiD}}^{\boldsymbol{X}} \right) &= \left(\frac{1}{n_1} + \frac{1}{n_1}\right) \left\{  \frac{T+1}{T} \sigma_{E}^2 + \vec{\Delta}_{\theta}^{\top} \Sigma_{\theta \theta} \vec{\Delta}_{\theta} -  \vec{\Delta}_{\theta}^{\top} \Sigma_{\theta X} \Sigma_{X X}^{-1} \Sigma_{X \theta} \vec{\Delta}_{\theta}   \right\}  \\
& \triangleq \left(\frac{1}{n_1} + \frac{1}{n_1}\right) \left\{  \frac{T+1}{T} \sigma_{E}^2 + \vec{\Delta}_{\theta}^{\top} \Sigma_{\tilde{\theta} \tilde{\theta}} \vec{\Delta}_{\theta}  \right\},
\end{aligned}
\]
where we denote $\Sigma_{\tilde{\theta} \tilde{\theta}} \triangleq \Sigma_{\theta \theta} - \Sigma_{\theta X} \Sigma_{X X}^{-1} \Sigma_{X \theta}$.  

We then only focus on $\text{Var} \left( \hat{\tau}_{\text{gDiD}}^{\boldsymbol{X},\boldsymbol{Y}^{\top}}\right)$. Using similar tricks as in Appendix \ref{Appendix B.Main results}, we can derive the following result after some algebraic manipulations: \[
\begin{aligned}
\text{Var} \left( \hat{\tau}_{\text{gDiD}}^{\boldsymbol{X},\boldsymbol{Y}^{\top}}\right) 
&= \frac{1}{n_1} \left\{ \vec{\beta}_{\theta,T}^{\top}  \Sigma_{\theta \theta} \vec{\beta}_{\theta,T} + \vec{\beta}_{x,T}^{\top}  \Sigma_{XX} \vec{\beta}_{x,T} + \sigma_{E}^2  + 2 \vec{\beta}_{\theta,T}^{\top}  \Sigma_{\theta X} \vec{\beta}_{x,T}    \right\}\\ 
& +  \frac{1}{n_1} \left\{ \vec{\beta}_{\theta,T}^{\top}  \left(\Sigma_{\theta \theta} - \begin{pmatrix}
\Sigma_{\theta X} \quad \Sigma_{\theta Y_T} 
\end{pmatrix}
\begin{pmatrix}
\Sigma_{XX} & \Sigma_{XY_T} \\
\Sigma_{Y_T X} & \Sigma_{Y_T Y_T}
\end{pmatrix}^{-1}
\begin{pmatrix}
\Sigma_{X \theta} \\
\Sigma_{Y_T \theta}
\end{pmatrix} \right) \vec{\beta}_{\theta,T}  + \sigma^2_E \right \}\\
&- \frac{1}{n_1} \bigg \{ \vec{\beta}_{\theta,T}^{\top} \left(\begin{pmatrix}
\Sigma_{\theta X} \quad \Sigma_{\theta Y_T} 
\end{pmatrix} \begin{pmatrix}
\Sigma_{XX} & \Sigma_{XY_T} \\
\Sigma_{Y_T X} & \Sigma_{Y_T Y_T}
\end{pmatrix}^{-1} \begin{pmatrix}
\Sigma_{X \theta} \\
\Sigma_{Y_T \theta}
\end{pmatrix}  \right) \vec{\beta}_{\theta,T} + \vec{\beta}_{x,T}^{\top} \Sigma_{X X} \vec{\beta}_{x,T} \\
& + 2 \vec{\beta}_{\theta,T}^{\top}  \begin{pmatrix}
\Sigma_{\theta X} \quad \Sigma_{\theta Y_T} 
\end{pmatrix} \begin{pmatrix}
\Sigma_{XX} & \Sigma_{XY_T} \\
\Sigma_{Y_T X} & \Sigma_{Y_T Y_T}
\end{pmatrix}^{-1}    \begin{pmatrix}
\Sigma_{XX}  \\
\Sigma_{Y_T X} 
\end{pmatrix} 
\vec{\beta}_{x,T}  \bigg \}, 
\end{aligned}
\]
where the upfront coefficient $(-\frac{1}{n_1})$ of the last term is not a mistake because $\operatorname{Cov}_{X,Y_0|Z=1}$ in $\text{Cov} \left( \bar{Y}_{1}^T,  \bar{Y}_{1}^{MT,{B}} \right)$ is the same as $\operatorname{Var}_{X,Y_0|Z=1}$ in $\Var \left( \bar{Y}_{1}^{MT,{B}} \right)$ so that we could cancel one term in which that upfront coefficient originally needs to be $-\frac{2}{n_1}$.

The key is then to decompose the terms multiplied with block matrix inversions. The two main terms are: 
$$
\begin{pmatrix}
\Sigma_{\theta X} \quad \Sigma_{\theta Y_T} 
\end{pmatrix}
\begin{pmatrix}
\Sigma_{XX} & \Sigma_{XY_T} \\
\Sigma_{Y_T X} & \Sigma_{Y_T Y_T}
\end{pmatrix}^{-1}
\begin{pmatrix}
\Sigma_{X \theta} \\
\Sigma_{Y_T \theta}
\end{pmatrix},
$$
and $$
\begin{pmatrix}
\Sigma_{\theta X} \quad \Sigma_{\theta Y_T} 
\end{pmatrix}
\begin{pmatrix}
\Sigma_{XX} & \Sigma_{XY_T} \\
\Sigma_{Y_T X} & \Sigma_{Y_T Y_T}
\end{pmatrix}^{-1}
\begin{pmatrix}
\Sigma_{X X} \\
\Sigma_{Y_T X}
\end{pmatrix}.
$$
We first employ the matrix-stacked notation on all $T$-period pre-treatment outcomes so that $Y_T = B_0 + B_\theta \theta + B_X X + \epsilon  $, where $Y_T \in \mathbf{R}^{T}$. Then we have
\[
\begin{aligned}
\Sigma_{\theta Y_T} = \Sigma_{\theta \theta} B_\theta^{\top} + \Sigma_{\theta X} B_X^{\top}, \quad \Sigma_{X Y_T} = \Sigma_{X\theta} B_{\theta}^\top + \Sigma_{XX} B_X^\top.     
\end{aligned}
\]
By Lemma \ref{Lemma: Block Matrix Inversion}, we have
$$
\left(
\begin{array}{cc}
\Sigma_{XX} & \Sigma_{XY_T} \\
\Sigma_{Y_TX} & \Sigma_{Y_TY_T}
\end{array}
\right)^{-1} = 
\left(
\begin{array}{cc}
\Sigma_{XX}^{-1} + \Sigma_{XX}^{-1} \Sigma_{XY_T} S^{-1} \Sigma_{Y_T X} \Sigma_{XX}^{-1} & - \Sigma_{XX}^{-1} \Sigma_{XY_T} S^{-1} \\
-S^{-1} \Sigma_{Y_T X} \Sigma_{X X}^{-1}  & S^{-1}
\end{array}
\right),
$$
where we denote $S$ as the Schur complement, and expanding the term leads to: \[
\begin{aligned}
S^{-1} &= (\Sigma_{Y_T Y_T} - \Sigma_{Y_T X} \Sigma_{X X}^{-1} \Sigma_{X Y_T})^{-1}    =  (B_\theta \left(\Sigma_{\theta \theta} - \Sigma_{\theta X} \Sigma_{X X}^{-1} \Sigma_{X \theta} \right) B_\theta^{\top}+ \Sigma_\epsilon)^{-1}. 
\end{aligned}
\]
This is because we have 
\[
\Sigma_{Y_T Y_T} = B_{\theta} \Sigma_{\theta \theta} B_{\theta}^\top + B_x \Sigma_{XX} B_x^\top + B_{\theta} \Sigma_{\theta X} B_x^\top + B_x \Sigma_{X \theta} B_{\theta}^\top + \Sigma_{\epsilon}
\]
and 
\[
\begin{aligned}
\Sigma_{Y_T X} \Sigma_{XX}^{-1} \Sigma_{XY_T}
& = (B_{\theta} \Sigma_{\theta X}  + B_X \Sigma_{XX}) \Sigma_{XX}^{-1} (\Sigma_{X\theta} B_{\theta}^\top + \Sigma_{XX} B_X^\top) \\ 
&= B_{\theta}  \Sigma_{\theta X}  \Sigma_{XX}^{-1} \Sigma_{X\theta} B_{\theta}^\top + B_{\theta}  \Sigma_{\theta X}  B_X^\top  + B_X \Sigma_{X\theta} B_{\theta}^\top  + B_X \Sigma_{XX} B_X^\top.
\end{aligned}
\]
Now we have 
\[
\begin{aligned}
& (\Sigma_{\theta X} \ \Sigma_{\theta Y_T})
\left(
\begin{array}{cc}
\Sigma_{XX} & \Sigma_{XY_T} \\
\Sigma_{Y_TX} & \Sigma_{Y_TY_T}
\end{array}
\right)^{-1} \\
= & (\Sigma_{\theta X} \quad \Sigma_{\theta \theta} B_\theta^{\top} + \Sigma_{\theta X} B_X^{\top} )  \left(
\begin{array}{cc}
\Sigma_{XX}^{-1} + \Sigma_{XX}^{-1} \Sigma_{XY_T} S^{-1} \Sigma_{Y_T X} \Sigma_{XX}^{-1} & - \Sigma_{XX}^{-1} \Sigma_{XY_T} S^{-1} \\
-S^{-1} \Sigma_{Y_T X} \Sigma_{X X}^{-1}  & S^{-1}
\end{array}
\right) \\
=& (\text{\ding{172}} \quad \text{\ding{173}} )
\end{aligned}
\]
where we obtain 
\[
\begin{aligned}
\text{\ding{172}} &= \Sigma_{\theta X} \left( \Sigma_{XX}^{-1} + \Sigma_{XX}^{-1} \Sigma_{XY_T} S^{-1} \Sigma_{Y_T X} \Sigma_{XX}^{-1} \right)
+ \left( \Sigma_{\theta \theta} B_{\theta}^\top + \Sigma_{\theta X} B_{x}^\top \right) \left( -S^{-1} \Sigma_{Y_T X} \Sigma_{XX}^{-1} \right) \\
&= \Sigma_{\theta X} \Sigma_{XX}^{-1}
+ \Sigma_{\theta X} \Sigma_{XX}^{-1} \Sigma_{XY_T} S^{-1} \Sigma_{Y_T X} \Sigma_{XX}^{-1}
- \Sigma_{\theta \theta} B_{\theta}^\top S^{-1} \Sigma_{Y_T X} \Sigma_{XX}^{-1}
- \Sigma_{\theta X} B_x^\top S^{-1} \Sigma_{Y_T X} \Sigma_{XX}^{-1} \\
&=  \Sigma_{\theta X} \Sigma_{XX}^{-1}
+ \left( \Sigma_{\theta X} \Sigma_{XX}^{-1} \Sigma_{XY_T} - \Sigma_{\theta \theta} B_{\theta}^\top - \Sigma_{\theta X} B_x^\top \right) S^{-1} \Sigma_{Y_T X} \Sigma_{XX}^{-1} \\
&= \Sigma_{\theta X} \Sigma_{XX}^{-1}
+ \left( \Sigma_{\theta X} \Sigma_{XX}^{-1} \Sigma_{X\theta} - \Sigma_{\theta \theta} \right) B_{\theta}^\top S^{-1} \Sigma_{Y_T X} \Sigma_{XX}^{-1}  \\
&= \Sigma_{\theta X} \Sigma_{XX}^{-1}
+ \left( \Sigma_{\theta X} \Sigma_{XX}^{-1} \Sigma_{X\theta} - \Sigma_{\theta \theta} \right) B_{\theta}^\top S^{-1} (B_{\theta} \Sigma_{\theta X}  + B_X \Sigma_{XX})  \Sigma_{XX}^{-1} 
\end{aligned}
\]
and 
\[
\begin{aligned}
\text{\ding{173}} &=   - \Sigma_{\theta X}  \Sigma_{XX}^{-1} \Sigma_{XY_T} S^{-1} + (\Sigma_{\theta \theta} B_\theta^{\top} + \Sigma_{\theta X} B_X^{\top} )  S^{-1} \\
&= - \Sigma_{\theta X}  \Sigma_{XX}^{-1} (\Sigma_{X\theta} B_{\theta}^\top + \Sigma_{XX} B_X^\top  )  S^{-1} + (\Sigma_{\theta \theta} B_\theta^{\top} + \Sigma_{\theta X} B_X^{\top}  )  S^{-1} \\
&=  \left(\Sigma_{\theta \theta} - \Sigma_{\theta X} \Sigma_{X X}^{-1} \Sigma_{X \theta} \right) B_\theta^{\top} S^{-1}
\end{aligned}
\]
Finally, we decompose the main terms mentioned previously using $(\text{\ding{172}} \quad \text{\ding{173}} )$ and rearrange terms accordingly: 
\[
\begin{aligned}
& \begin{pmatrix}
\Sigma_{\theta X} \quad \Sigma_{\theta Y_T} 
\end{pmatrix}
\begin{pmatrix}
\Sigma_{XX} & \Sigma_{XY_T} \\
\Sigma_{Y_T X} & \Sigma_{Y_T Y_T}
\end{pmatrix}^{-1}
\begin{pmatrix}
\Sigma_{X \theta} \\
\Sigma_{Y_T \theta}
\end{pmatrix} \\
= & (\text{\ding{172}} \quad \text{\ding{173}} )    
\begin{pmatrix}
\Sigma_{X \theta} \\
\Sigma_{Y_T \theta}
\end{pmatrix} \\
= &  \Sigma_{\theta X} \Sigma_{XX}^{-1} \Sigma_{X \theta}
+ \left( \Sigma_{\theta X} \Sigma_{XX}^{-1} \Sigma_{X\theta} - \Sigma_{\theta \theta} \right) B_{\theta}^\top S^{-1} (B_{\theta} \Sigma_{\theta X}  + B_X \Sigma_{XX})  \Sigma_{XX}^{-1} \Sigma_{X \theta} \\
+ & \left(\Sigma_{\theta \theta} - \Sigma_{\theta X} \Sigma_{X X}^{-1} \Sigma_{X \theta} \right) B_\theta^{\top} S^{-1} \Sigma_{Y_T \theta}\\
= & \Sigma_{\theta X} \Sigma_{XX}^{-1} \Sigma_{X \theta} + \left( \Sigma_{\theta X} \Sigma_{XX}^{-1} \Sigma_{X\theta} - \Sigma_{\theta \theta} \right) B_{\theta}^\top S^{-1} B_{\theta} \Sigma_{\theta X}  \Sigma_{XX}^{-1} \Sigma_{X \theta} +   \left( \Sigma_{\theta X} \Sigma_{XX}^{-1} \Sigma_{X\theta} - \Sigma_{\theta \theta} \right) B_{\theta}^\top S^{-1}  B_X  \Sigma_{X \theta}\\ 
+ & \left(\Sigma_{\theta \theta} - \Sigma_{\theta X} \Sigma_{X X}^{-1} \Sigma_{X \theta} \right) B_\theta^{\top} S^{-1} \left(B_\theta \Sigma_{\theta \theta} + B_X \Sigma_{X \theta} \right) \\
= & \Sigma_{\theta X} \Sigma_{XX}^{-1} \Sigma_{X \theta} + \left( \Sigma_{\theta X} \Sigma_{XX}^{-1} \Sigma_{X\theta} - \Sigma_{\theta \theta} \right) B_{\theta}^\top S^{-1} B_{\theta} \left( \Sigma_{\theta X} \Sigma_{XX}^{-1} \Sigma_{X\theta} - \Sigma_{\theta \theta} \right) \\
= & \Sigma_{\theta X} \Sigma_{XX}^{-1} \Sigma_{X \theta} + \mathbf{r}_{\theta|x} \Sigma_{\tilde{\theta} \tilde{\theta}}
\end{aligned}
\]
where we denote $\Sigma_{\tilde{\theta} \tilde{\theta}} \triangleq \Sigma_{\theta \theta}- \Sigma_{\theta X} \Sigma_{XX}^{-1} \Sigma_{X\theta}  $ and 
\[
\begin{aligned}
\mathbf{r}_{\theta|x} 
&=  \left(\Sigma_{\theta \theta} - \Sigma_{\theta X} \Sigma_{X X}^{-1} \Sigma_{X \theta} \right)  B_\theta^{\top} S^{-1} B_\theta \\
& = \left(\Sigma_{\theta \theta} - \Sigma_{\theta X} \Sigma_{X X}^{-1} \Sigma_{X \theta} \right) B_{\theta}^{\top} (B_\theta \left(\Sigma_{\theta \theta} - \Sigma_{\theta X} \Sigma_{X X}^{-1} \Sigma_{X \theta} \right) B_\theta^{\top}+ \Sigma_\epsilon)^{-1} B_\theta \\
& \triangleq \Sigma_{\tilde{\theta} \tilde{\theta}} B_{\theta}^{\top} (B_\theta \Sigma_{\tilde{\theta} \tilde{\theta}} B_\theta^{\top}+ \Sigma_\epsilon)^{-1} B_\theta    \in \mathbf{R}^{q \times q}.   
\end{aligned}
\]
Similarly, we can decompose the other main term and it easily follows that:
\[
\begin{aligned}
& \begin{pmatrix}
\Sigma_{\theta X} \quad \Sigma_{\theta Y_T} 
\end{pmatrix}
\begin{pmatrix}
\Sigma_{XX} & \Sigma_{XY_T} \\
\Sigma_{Y_T X} & \Sigma_{Y_T Y_T}
\end{pmatrix}^{-1}
\begin{pmatrix}
\Sigma_{X X} \\
\Sigma_{Y_T X}
\end{pmatrix} 
=  \Sigma_{\theta X}. 
\end{aligned}
\]
Combining all terms and simplifying the algebra leads to the desired result. 
\end{proof}

\subsection{Proof of Lemma \ref{Lemma 4.1}} \label{Appendix C.2}
\noindent Before presenting the proof, we first revisit the (limiting) expected values of generalized DiD estimators and generalized matching DiD estimators, either matched on observed covariates or both observed covariates and pre-treatment outcomes, proposed by \cite{ham2024benefits}. This could also be shown by just changing $Y_{i,T-1}$ to the stacked multi pre-treatment outcomes $\bar{Y}_{i,0:(T-1)}$ in Theorem \ref{Thm: Consistency of DiD and Matching DiD Estimators}, and the proof strategy in Appendix \ref{Appendix A} follows. 
\[ 
\begin{aligned}
\mathbb{E}\left[\hat{\tau}_{\text{gDiD}} \right] 
&= \mathbb{E}[Y_{i,T} \mid Z_i = 1] - \mathbb{E}[Y_{i,T} \mid Z_i = 0]  \\
&\quad - \left(  \mathbb{E}\left[\bar{Y}_{i,0:(T-1)} \mid Z_i = 1\right] - \mathbb{E}\left[\bar{Y}_{i,0:(T-1)} \mid Z_i = 0\right] \right) \\
\mathbb{E}\left[\hat{\tau}^{\boldsymbol{X}}_{\text{gDiD}} \right] 
&= \mathbb{E}[Y_{i,T} \mid Z_i = 1] - \mathbb{E}\left[\bar{Y}_{i,0:(T-1)} \mid Z_i = 1\right] \\
&\quad - \left( \mathbb{E}_{\boldsymbol{x} \mid Z_i = 1} \left[ \mathbb{E}\left[Y_{i,T} \mid Z_i = 0, \boldsymbol{X}_i = \boldsymbol{x}\right] \right] 
- \mathbb{E}_{\boldsymbol{x} \mid Z_i = 1} \left[ \mathbb{E}\left[\bar{Y}_{i,0:(T-1)} \mid Z_i = 0, \boldsymbol{X}_i = \boldsymbol{x} \right] \right] \right) \notag  \\
\mathbb{E}\left[\hat{\tau}_{\text{gDiD}}^{\boldsymbol{X},\boldsymbol{Y}^{\mathbf{T}}}\right] 
&= \mathbb{E}[Y_{i,T} \mid Z_i = 1] 
- \mathbb{E}_{(\boldsymbol{x}, \boldsymbol{y}^{\mathbf{T}}) \mid Z_i = 1} \left[ \mathbb{E}\left[Y_{i,T} \mid Z_i = 0, \boldsymbol{X}_i = \boldsymbol{x}, \boldsymbol{Y}^{\mathbf{T}}_i  = \boldsymbol{y}^{\mathbf{T}}_i\right] \right],
\end{aligned}
\]
where $\boldsymbol{Y}^{\mathbf{T}}_i = (Y_{i,0}, \cdots, Y_{i,T-1})$ represents the collection of all pre-treatment outcomes. 

\begin{proof}[Proof of Lemma \ref{Lemma 4.1}]
\cite{ham2024benefits} (their Theorem 5.1) showed that  
\[
\begin{aligned}
\text{Bias}  \left( \hat{\tau}_{\text{gDiD}} \right) &= \vec{\Delta}_{\theta}^{\top} \vec{\delta}_{\theta}  +\vec{\Delta}_{X}^{\top} \vec{\delta}_{x}, \\
\text{Bias}  \left( \hat{\tau}_{\text{gDiD}}^{\boldsymbol{X}} \right) &= 
   \vec{\Delta}_{\theta}^{\top} \left[\vec{\delta}_{\theta} - \Sigma_{\theta X} \Sigma_{X X}^{-1}  \vec{\delta}_{x} \right], \\ 
\text{Bias} \left(\hat{\tau}_{g\text{DiD}}^{\boldsymbol{X}, \boldsymbol{Y^T}}\right) &= \vec{\beta}_{\theta,T}^\top 
\left[
\vec{\delta}_{\theta} - (\Sigma_{\theta X} \ \Sigma_{\theta Y_T})
\left(
\begin{array}{cc}
\Sigma_{XX} & \Sigma_{XY_T} \\
\Sigma_{Y_TX} & \Sigma_{Y_TY_T}
\end{array}
\right)^{-1}
\left(
\begin{array}{c}
\vec{\delta}_x \\
\vec{\beta}_{\theta,0}^\top \vec{\delta}_{\theta} + \vec{\beta}_{x,0}^\top \vec{\delta}_x \\
\vdots \\
\vec{\beta}_{\theta,T-1}^\top \vec{\delta}_{\theta} + \vec{\beta}_{x,T-1}^\top \vec{\delta}_x
\end{array}
\right)
\right]. 
\end{aligned}
\]
As noted in the main text and Appendix \ref{Appendix B.First-order Sampling Approximations}, we take views of these biases to be finite-sample approximations of the limiting bias. We adopt the equality sign, albeit with a slight abuse of notation, to maintain consistency (literal meaning, not in statistical sense) with the variance formulation. 

Denote $\vec{\delta}_{\tilde{\theta}} \triangleq \vec{\delta}_{\theta}  - \Sigma_{\theta X} \Sigma_{X X}^{-1} 
 \vec{\delta}_{x}$ and then the first two bias results immediately follow after rearranging equations. It then suffices to show that we could rewrite $\text{Bias} \left(\hat{\tau}_{g\text{DiD}}^{\boldsymbol{X}, \boldsymbol{Y^T}}\right)$. 
 
Again, similar to the proof in Appendix \ref{Appendix C.1}, the key is to decompose the minus part inside the bracket with proper usage of block matrix inversions by Lemma \ref{Lemma: Block Matrix Inversion}. We first employ the matrix-stacked notation on all $T$-period pre-treatment outcomes so that $Y_T = B_0 + B_\theta \theta + B_X X + \epsilon  $, where $Y_T \in \mathbf{R}^{T}$. We also rewrite that 
$$
\left(
\begin{array}{c}
\vec{\delta}_x \\
\vec{\beta}_{\theta,0}^\top \vec{\delta}_{\theta} + \vec{\beta}_{x,0}^\top \vec{\delta}_x \\
\vdots \\
\vec{\beta}_{\theta,T-1}^\top \vec{\delta}_{\theta} + \vec{\beta}_{x,T-1}^\top \vec{\delta}_x
\end{array}
\right) =
\left(
\begin{array}{c}
\vec{\delta}_x \\
B_\theta \vec{\delta}_{\theta} + B_x\vec{\delta}_{x} 
\end{array}
\right).
$$
We have shown in Appendix \ref{Appendix C.1} that 
\[
\begin{aligned}
(\Sigma_{\theta X} \ \Sigma_{\theta Y_T})
\left(
\begin{array}{cc}
\Sigma_{XX} & \Sigma_{XY_T} \\
\Sigma_{Y_TX} & \Sigma_{Y_TY_T}
\end{array}
\right)^{-1} = (\text{\ding{172}} \quad \text{\ding{173}} ).
\end{aligned}
\]
Multiplying and rearranging terms leads to: 
\[
\begin{aligned}
& (\text{\ding{172}} \quad \text{\ding{173}} )    \left(
\begin{array}{c}
\vec{\delta}_x \\
B_\theta \vec{\delta}_{\theta} + B_x\vec{\delta}_{x} 
\end{array}
\right) \\
= &  \Sigma_{\theta X} \Sigma_{XX}^{-1} \vec{\delta}_x 
+ \left( \Sigma_{\theta X} \Sigma_{XX}^{-1} \Sigma_{X\theta} - \Sigma_{\theta \theta} \right) B_{\theta}^\top S^{-1} (B_{\theta} \Sigma_{\theta X}  + B_X \Sigma_{XX})  \Sigma_{XX}^{-1} \vec{\delta}_x \\
+ & \left(\Sigma_{\theta \theta} - \Sigma_{\theta X} \Sigma_{X X}^{-1} \Sigma_{X \theta} \right) B_\theta^{\top} S^{-1} B_\theta \vec{\delta}_{\theta} + \left(\Sigma_{\theta \theta} - \Sigma_{\theta X} \Sigma_{X X}^{-1} \Sigma_{X \theta} \right) B_\theta^{\top} S^{-1} B_x\vec{\delta}_{x} \\
= &  \Sigma_{\theta X} \Sigma_{XX}^{-1} \vec{\delta}_x + \left( \Sigma_{\theta X} \Sigma_{XX}^{-1} \Sigma_{X\theta} - \Sigma_{\theta \theta} \right) B_{\theta}^\top S^{-1} (B_{\theta} \Sigma_{\theta X})  \Sigma_{XX}^{-1} \vec{\delta}_x  +  \left(\Sigma_{\theta \theta} - \Sigma_{\theta X} \Sigma_{X X}^{-1} \Sigma_{X \theta} \right) B_\theta^{\top} S^{-1} B_\theta \vec{\delta}_{\theta} \\
= & \Sigma_{\theta X} \Sigma_{XX}^{-1} \vec{\delta}_x + \left(\Sigma_{\theta \theta} - \Sigma_{\theta X} \Sigma_{X X}^{-1} \Sigma_{X \theta} \right)  B_\theta^{\top} S^{-1} B_\theta 
\left( \vec{\delta}_\theta -  \Sigma_{\theta X}  \Sigma_{XX}^{-1} \vec{\delta}_x   \right)   \\
= &  \Sigma_{\theta X} \Sigma_{XX}^{-1} \vec{\delta}_x  + \mathbf{r}_{\theta|x} \left( \vec{\delta}_\theta -  \Sigma_{\theta X}  \Sigma_{XX}^{-1} \vec{\delta}_x   \right)
\end{aligned}
\]
where $\mathbf{r}_{\theta|x}$ is defined as previously. 

Finally we combine the terms with the upfront $\vec{\delta}_{\theta}$ and rearranging terms leads to the desired result. 
\end{proof}

\subsection{Proof of Lemma \ref{Lemma 4.2}}\label{Appendix C.3}
\begin{proof}
The proof holds for both non-residualized variance $\sigma_{\theta}$ and  residualized variance $ {\sigma}_{\tilde \theta}$. For notational simplicity, we use $\theta$ throughout the proof.

Since $\theta \in \mathbf{R}$, we have $r_{\theta|x} =
 B_{\theta}^{\top} (B_\theta B_\theta^{\top} \sigma_{\theta}^2+ \Sigma_\epsilon)^{-1} B_\theta \sigma_{\theta}^2 \in \mathbf{R}$, where 
$$B_{\theta} = [\beta_{\theta,0},\cdots, \beta_{\theta,T-1}]^{\top} \in \mathbf{R}^{T \times 1}, \quad  \Sigma_{\epsilon} = \sigma_E^2 \ \mathbf{I}_T \in \mathbf{R}^{T \times T}.$$ 
Note that $1 + \sigma_{\theta}^2 B_\theta^{\top} (\sigma_E^2)^{-1} \mathbf{I}_T B_\theta = 1 + \frac{\sigma_{\theta}^2}{\sigma_E^2}  \| B_\theta \|^2 > 1 \neq 0$, where 
$\| B_\theta \|^2 = B_\theta^\top B_\theta = \sum_{t=0}^{T-1} \beta_{\theta,t}^2$. 
Then applying the Sherman–Morrison formula by Lemma \ref{Lemma: Sherman–Morrison formula} to $(\Sigma_\epsilon + B_\theta B_\theta^{\top} \sigma_{\theta}^2)^{-1},$  and we have 
\[
\begin{aligned}
(\Sigma_\epsilon + B_\theta B_\theta^{\top} \sigma_{\theta}^2)^{-1} &= \frac{1}{\sigma_E^2} \mathbf{I}_T - \frac{\sigma_\theta^2 \cdot (\sigma_E^2)^{-1} \mathbf{I}_T B_\theta B_\theta^\top (\sigma_E^2)^{-1}  \mathbf{I}_T }{1 + \frac{\sigma_{\theta}^2}{\sigma_E^2}  \| B_\theta \|^2 }  
\\
& = \frac{1}{\sigma_E^2} \mathbf{I}_T - \frac{\sigma_\theta^2}{\sigma_E^4 + \sigma_\theta^2 \sigma_E^2 \|B_\theta\|^2} B_\theta B_\theta^\top \\ 
&=  \frac{1}{\sigma_E^2} \left( \mathbf{I}_T - \frac{\sigma_\theta^2}{\sigma_E^2 + \sigma_\theta^2 \|B_\theta\|^2} B_\theta B_\theta^\top \right)    
\end{aligned}
\]
Then 
\[
\begin{aligned}
r_{\theta|x} &=
 B_{\theta}^{\top} (B_\theta B_\theta^{\top} \sigma_{\theta}^2+ \Sigma_\epsilon)^{-1} B_\theta \sigma_{\theta}^2  \\
&= B_{\theta}^{\top} \frac{\sigma_{\theta}^2}{\sigma_E^2} \left( \mathbf{I}_T - \frac{\sigma_\theta^2}{\sigma_E^2 + \sigma_\theta^2 \|B_\theta\|^2} B_\theta B_\theta^\top \right) B_\theta  \\
&= \frac{\sigma_{\theta}^2}{\sigma_E^2} \left(B_{\theta}^{\top} B_\theta - \frac{\sigma_\theta^2}{\sigma_E^2 + \sigma_\theta^2 \|B_\theta\|^2} B_\theta^\top B_\theta B_\theta^\top B_\theta  \right) \\
&= \frac{\sigma_{\theta}^2}{\sigma_E^2} \left( \| B_\theta \|^2  - \frac{\sigma_\theta^2}{\sigma_E^2 + \sigma_\theta^2 \|B_\theta\|^2}  \| B_\theta \|^4 \right) \\
&= \frac{\sigma_{\theta}^2}{\sigma_E^2} \left(  \frac{\| B_\theta \|^2 (\sigma_E^2 + \sigma_\theta^2 \|B_\theta\|^2) - \sigma_\theta^2 \| B_\theta \|^4}{\sigma_E^2 + \sigma_\theta^2 \|B_\theta\|^2}  \right) \\
&= \frac{\sigma_\theta^2 \| B_\theta \|^2}{\sigma_E^2 + \sigma_\theta^2 \|B_\theta\|^2}  
\end{aligned}
\]
Lastly, plugging $\| B_\theta \|^2 = B_\theta^\top B_\theta = \sum_{t=0}^{T-1} \beta_{\theta,t}^2$, we reach to the desired results. 
\end{proof}

\clearpage

\subsection{Proof of Proposition \ref{Prop 4.1}}\label{Appendix C.4}
\begin{proof}

The proof holds for both non-residualized variance $\Sigma_{\theta \theta}$ and  residualized variance $ \Sigma_{\tilde{\theta}\tilde{\theta}}$. Therefore, the result applies to all cases in which the reliability term appears. For clarity, we simply use $\theta$ throughout the proof. 

By positive-definiteness assumption of $\Sigma_{\theta \theta}$, we could invert the matrix. Denote $A := B_\theta \Sigma_{\theta\theta}^{1/2} \in \mathbf{R}^{T \times q}$, and then we could rewrite $\mathbf{r}_{\theta|x} \triangleq  \Sigma_{\theta\theta}^{1/2} A^\top (A A^\top + \sigma_E^2 \mathbf{I}_T)^{-1} A  \Sigma_{\theta\theta}^{-1/2}.$
We also define:
\[
M_{T} := A^\top (A A^\top + \sigma_E^2 \mathbf{I}_T )^{-1} A \in \mathbf{R}^{q \times q}
\quad \text{so that} \quad
\mathbf{r}_{\theta|x} = \Sigma_{\theta\theta}^{1/2} M_{T} \Sigma_{\theta\theta}^{-1/2}
\]
It then suffices to show that $M_T \xrightarrow{\|\cdot\|_2} \mathbf{I}_q \ \text{as } T \to \infty.$  Applying 
the Woodbury matrix identity by Lemma \ref{Lemma: Woodbury Matrix Identity} to the inversion part, we have: 
\[
\begin{aligned}
(\sigma_E^2 \mathbf{I}_T + A A^\top )^{-1} &=  (\sigma_E^2)^{-1}\mathbf{I}_T - (\sigma_E^2)^{-1}\mathbf{I}_T A \left(\mathbf{I}_q^{-1} + A^{\top} (\sigma_E^2)^{-1} \mathbf{I}_T A \right)^{-1} A^{\top} (\sigma_E^2)^{-1} \mathbf{I}_T   \\  
&= (\sigma_E^2)^{-1}\mathbf{I}_T  - (\sigma_E^2)^{-2} A \left(\mathbf{I}_q + (\sigma_E^2)^{-1} A^{\top} A \right)^{-1} A^{\top},   
\end{aligned}
\]
which then implies that 
\[
\begin{aligned}
M_T &= (\sigma_E^2)^{-1}A^{\top}A - (\sigma_E^2)^{-2} A^{\top}A \left(\mathbf{I}_q + (\sigma_E^2)^{-1} A^{\top} A \right)^{-1} A^{\top}A \\
&= (\sigma_E^2)^{-1} A^{\top}A\left[\mathbf{I}_q -   (\sigma_E^2)^{-1} \left(\mathbf{I}_q + (\sigma_E^2)^{-1} A^{\top}  A \right)^{-1} A^{\top}A \right] \\
&= (\sigma_E^2)^{-1} A^{\top}A\left[\mathbf{I}_q -    \left(\sigma_E^2 \mathbf{I}_q +   A^{\top}  A \right)^{-1} A^{\top}A \right] \\
&= (\sigma_E^2)^{-1} A^{\top}A \left[\sigma_E^2 \left(\sigma_E^2 \mathbf{I}_q + A^{\top}A \right)^{-1} \right]
\end{aligned}
\]
The last line follows from the distributive property of an identity matrix that 
\[
\begin{aligned}
\mathbf{I}_q &= \left(\sigma_E^2 \mathbf{I}_q +   A^{\top}  A \right)^{-1}\left(\sigma_E^2 \mathbf{I}_q +   A^{\top}  A \right) \\
& \Rightarrow \mathbf{I}_q  - \left(\sigma_E^2 \mathbf{I}_q +   A^{\top}  A \right)^{-1}  A^{\top}A = \sigma_E^2 \left(\sigma_E^2 \mathbf{I}_q + A^{\top}A \right)^{-1} 
\end{aligned}
\]
Then we have
\[
\begin{aligned}
M_{T} &= A^\top A (A^\top A + \sigma_E^2 \mathbf{I}_q)^{-1} \\
&= (A^\top A + \sigma_E^2 \mathbf{I}_q - \sigma_E^2 \mathbf{I}_q)  (A^\top A + \sigma_E^2 \mathbf{I}_q)^{-1} \\
&= \mathbf{I}_q - \sigma_E^2 \mathbf{I}_q (A^\top A + \sigma_E^2 \mathbf{I}_q)^{-1} \\
&= \mathbf{I}_q - \sigma_E^2 (A^\top A + \sigma_E^2 \mathbf{I}_q)^{-1}  
\end{aligned}
\]
Finally, by assumption, we have 
\[
\frac{1}{T} \sum_{t=0}^{T-1} \vec{\beta}_{\theta,t} \vec{\beta}_{\theta,t}^\top \triangleq \frac{1}{T} B_\theta^\top B_\theta \to Q \succ 0 \quad \text{as } T \to \infty 
\quad \Rightarrow \quad
B_\theta^\top B_\theta =TQ +o(T) 
\]
Then 
\[
\begin{aligned}
A^\top A &= \Sigma_{\theta\theta}^{1/2} B_\theta^\top B_\theta \Sigma_{\theta\theta}^{1/2} \\
&= T \cdot \Sigma_{\theta\theta}^{1/2} Q \Sigma_{\theta\theta}^{1/2} + o(T) 
\end{aligned}
\]
where the second equality follows by that $\Sigma_{\theta\theta}^{1/2}$ does not scale with $T$ and thus the remaining order terms still stay bounded in the spectral norm. 

Then by continuous mapping theorem, in which we assume $Q \succ 0$ and thus $ \Sigma_{\theta\theta}^{1/2} Q \Sigma_{\theta\theta}^{1/2}\succ 0$, so the whole matrix after summation is invertible and continuous at the point, we have 
$$
\sigma_E^2 \left(A^\top A + \sigma_E^2 \mathbf{I}_q\right)^{-1}  = \frac{\sigma_E^2}{T} \left(\Sigma_{\theta\theta}^{1/2} Q \Sigma_{\theta\theta}^{1/2} + \frac{o(T)}{T} + \frac{\sigma_E^2 \mathbf{I}_q}{T}\right)^{-1} \xrightarrow{\|\cdot\|_2} 0 \quad \text{as } T \to \infty. 
$$
Therefore $M_T \xrightarrow{\|\cdot\|_2} \mathbf{I}_q \ \text{as } T \to \infty.$
\end{proof}

\clearpage

\section{Proofs of Theoretical Results in Section \ref{section: Determining What to Match}}  
\noindent We state some useful tricks that will be used out through this section.

\begin{lemma} \label{Lemma: Residualization trick}
Suppose the conditions of Corollary \ref{Corollary 4.1} hold. We have 
$$\hat{\vec{\beta}}_{x,t} \xrightarrow{p} \vec{\beta}_{x,t} +  \Sigma_{XX}^{-1} \Sigma_{X \theta } \vec{\beta}_{\theta,t}.$$
and 
\[
\tilde{Y}_{i,t} \overset{p}{\to} \beta_{0,t} + \tau_i \mathbf{1}(t=T) + \vec{\beta}_{\theta,t}^{\top} (\boldsymbol{\theta}_i - \Sigma_{\theta X} \Sigma_{XX}^{-1} \boldsymbol{X}_i) + \epsilon_{i,t} \triangleq \vec{\beta}_{\theta,t}^{\top} \boldsymbol{\tilde{\theta}}_i + \nu_{i,t},
\]
where we redefine $\boldsymbol{\tilde{\theta}}_i \triangleq \boldsymbol{\theta}_i - \Sigma_{\theta X} \Sigma_{XX}^{-1} \boldsymbol{X}_i$ and $\nu_{i,t} \triangleq \beta_{0,t} + \tau_i \cdot Z_i \mathbf{1}(t=T) + \epsilon_{i,t} $ with $\mathbb{E}[\epsilon_{i,t}] = 0$.
\end{lemma}

\begin{proof}
Utilizing the elementary result from linear regression, we have that $\hat{\vec{\beta}}_{x,t}$ converges in probability to the coefficient of $\boldsymbol{X}_i$ in $\mathbb{E}[Y_{i,t} \mid Z_i = 0, \boldsymbol{X}_i]$. Specifically, we have 
\begin{align*}
    \mathbb{E}[Y_{i,t} \mid Z_i = 0, \boldsymbol{X}_i] &= \beta_{0,t} + \vec{\beta}_{\theta,t}^\top \left(\vec{\mu}_{\theta,0} + \Sigma_{\theta X} \Sigma_{XX}^{-1} (\boldsymbol{X}_i - \vec{\mu}_{x,0}) \right) + \vec{\beta}_{x,t}^\top \boldsymbol{X}_i \\
    &= \beta_{0,t} + 
\vec{\beta}_{\theta,t}^\top  \left(\vec{\mu}_{\theta,0} - \Sigma_{\theta X} \Sigma_{XX}^{-1} \vec{\mu}_{x,0} \right) + \left( \vec{\beta}_{x,t}^\top + \vec{\beta}_{\theta,t}^\top \Sigma_{\theta X} \Sigma_{XX}^{-1} \right) \boldsymbol{X}_i,
\end{align*}
which holds because $ \mathbb{E}[\boldsymbol{\theta}_i \mid Z_i = 0, \boldsymbol{X}_i] = \vec{\mu}_{\theta,0} + \Sigma_{\theta X} \Sigma_{XX}^{-1} (\boldsymbol{X}_i - \vec{\mu}_{x,0}).$ 

The second part follows directly from the residualization of the outcomes in Equation\eqref{Eq: resdiualized outcome} with continuous mapping theorem.
\end{proof}

\subsection{Proofs of Estimation Strategies for Relative (Absolute and Non-Absolute) Bias Reduction} \label{Appendix: Proofs of Estimation Strategies for (Absolute and Non-Absolute) Bias Reduction}
\noindent In this section, we first prove estimation strategies for relative absolute bias reduction in Appendix \ref{Appendix: Proof of Prop: Relative Absolute Bias Reduction Pair 2} and then prove estimation strategies for relative non-absolute bias reduction in Appendix \ref{Appendix: Proofs for Relative (Non-Absolute) Bias Reduction}.

\subsubsection{Proof of Proposition \ref{Prop: Relative Absolute Bias Reduction Pair 2}} \label{Appendix: Proof of Prop: Relative Absolute Bias Reduction Pair 2}
\noindent  \cite{ham2024benefits}'s Theorem 6.2. (proof in their Appendix G.2.) doesn't provide a proof for what we propose in Proposition \ref{Prop: Relative Absolute Bias Reduction Pair 2}, so we complete it here. 
\begin{proof}
Assume $\theta \in \mathbf{R}$ and $T>1$. From Lemma \ref{Lemma 4.1} and Lemma \ref{Lemma 4.2}, we know 
\[
\left| \Bias    \left( \hat{\tau}_{\text{gDiD}}^{\boldsymbol{X}} \right)\right| - \left|\Bias    \left( \hat{\tau}_{g\text{DiD}}^{\boldsymbol{X}, \boldsymbol{Y^T}} \right) \right|  = \underbrace{|{\Delta}_\theta   \delta_{\tilde{\theta}}|}_{\text{first term}}   - \underbrace{|\beta_{\theta,T} \delta_{\tilde{\theta}} (1-r_{\theta|x})|}_{\text{second term}}.
\]
\textbf{Part 1:} We start with consistent estimation of the first term. From Lemma \ref{Lemma: Residualization trick} and Assumption \ref{Assumption: errors}, for the two pre-treatment periods $s \neq q <T$ that we have stable effects of $\theta$, i.e., ${\beta}_{\theta,s} =  {\beta}_{\theta,q}$, from the assumption in Proposition \ref{Prop: Relative Absolute Bias Reduction Pair 2}, we have
\[
{\Var}\left(\tilde{Y}_{i,s} - {\tilde{Y}}_{i,q} \right) = (\beta_{\theta,s} -{\beta}_{\theta,q})^2 \sigma_{\tilde{\theta}}^2 + 2{\sigma}_E^2 =  2{\sigma}_E^2.
\]
Then we have the consistent estimator of the residual error using continuous mapping theorem, i.e., 
$$\hat{\sigma}_E^2 = \frac{1}{2} \widehat{\Var} \left( \tilde{Y}_{i, s} - \tilde{Y}_{i, q} \right) \xrightarrow{p}  {\sigma}_E^2;$$
as well as the consistent estimator of $\beta_{\theta,\cdot}$ in its absolute value, i.e., $\forall t =0,1,\cdots,T$,
\[
\hat{\beta}_{\theta, t} = \sqrt{ \widehat{\Var}(\tilde{Y}_{i,t} \mid Z_i = 0) - \hat{\sigma}_E^2 } \xrightarrow{p} |\beta_{\theta, t}  \sigma_{\tilde{\theta}}|. 
\]
Similarly, let ${\bar \beta}_{\theta, \text{pre}} = \frac{1}{T} \sum_{t=0}^{T-1}{\beta}_{\theta, t}$. We have 
\[
{\Var}\left(\tilde{Y}_{i,T} - \bar{\tilde{Y}}_{i,0:(T-1)} \mid Z_i = 0\right) = (\beta_{\theta,T} - \bar\beta_{\theta,\text{pre}} )^2 \sigma_{\tilde{\theta}}^2 + \frac{T+1}{T}{\sigma}_E^2,
\]
and it follows that 
\[
\hat{\Delta}_\theta = \sqrt{ \widehat{\Var}\left(\tilde{Y}_{i,T} - \bar{\tilde{Y}}_{i,0:(T-1)} \mid Z_i = 0\right) - \frac{T+1}{T} \hat{\sigma}_E^2 } \xrightarrow{p} \left| (\beta_{\theta,T} - \bar\beta_{\theta,\text{pre}} ) \sigma_{\tilde{\theta}} \right|.
\]
Now from Lemma \ref{Lemma: Residualization trick}, by continuous mapping theorem, we can show 
$$\widehat\delta_{\tilde{\theta}} = \frac{
\left|\widehat{\mathbb{E}}\!\left[\bar{\tilde{Y}}_{i,0:(T-1)} \mid Z_i = 1\right]
-
\widehat{\mathbb{E}}\!\left[\bar{\tilde{Y}}_{i,0:(T-1)} \mid Z_i = 0\right]  
\right|}{
\hat{\bar \beta}_{\theta, \text{pre}}
} \xrightarrow{p} \frac{|\bar\beta_{\theta, \text{pre}}  \delta_{\tilde{\theta}} |}{|\bar\beta_{\theta, \text{pre}}  \sigma_{\tilde{\theta}} |}  \underbrace{=}_{|a \cdot b| = |a| \cdot |b|} \frac{|\delta_{\tilde{\theta}}|}{|\sigma_{\tilde{\theta}}|}.$$
Recall $\sigma_{\tilde{\theta}} \ge 0$. Combining terms lead to 
\[
\hat{\Delta}_\theta \hat\delta_{\tilde{\theta}}\xrightarrow{p}   \left| (\beta_{\theta,T} - \bar\beta_{\theta,\text{pre}} ) \sigma_{\tilde{\theta}} \right| \cdot \frac{|\delta_{\tilde{\theta}}|}{|\sigma_{\tilde{\theta}}|} = |(\beta_{\theta,T} - \bar\beta_{\theta,\text{pre}}) \cdot \delta_{\tilde{\theta}} | = |{\Delta}_\theta   \delta_{\tilde{\theta}}|.
\]
\textbf{Part 2:} Now we show the consistent estimation of the second term. By similar arguments as in Part 1, it immediately follows that 
\[
\widehat{\beta}_{\theta,T} \widehat{\delta}_{\tilde{\theta}} \xrightarrow{p} |{\beta}_{\theta,T} {\delta}_{\tilde{\theta}}|.
\]
It then suffices to show the consistent estimation of the reliability term which is always non-negative by Definition \ref{Def: reliabiltiy}. In particular, we have
\[
\hat{r}_{\theta|x} = \frac{\sum_{t=0}^{T-1} \hat{\beta}_{\theta, t}^2}{\sum_{t=0}^{T-1} \hat{\beta}_{\theta, t}^2 + \hat{\sigma}_E^2}
\xrightarrow{p} \frac{\sum_{t=0}^{T-1} {\beta}_{\theta, t}^2 \sigma^2_{\tilde{\theta}}} {\sum_{t=0}^{T-1} {\beta}_{\theta, t}^2 \sigma^2_{\tilde{\theta}} + {\sigma}_E^2}  = r_{\theta|x}. 
\]    
By another continuous mapping theorem, this concludes the proof. 
\end{proof}

\subsubsection{Proofs of Estimation Strategies for Relative (Non-Absolute) Bias Reduction} \label{Appendix: Proofs for Relative (Non-Absolute) Bias Reduction}
\begin{proposition}[\textbf{Estimation Strategies for Relative Bias Reduction}] \label{Prop: Estimation Strategies for Relative Bias Reduction} 
Suppose the conditions of Lemma \ref{Lemma 4.1} hold. Then we have
\begin{align*}
\widehat{\Delta}_{\tau_x}  &\xrightarrow{p} \Bias  \left( \hat{\tau}_{\text{gDiD}} \right) - \Bias    \left( \hat{\tau}_{\text{gDiD}}^{\boldsymbol{X}} \right) \\
\widehat{\Delta}_{\tau_{x,xy}}  &\xrightarrow{p}  \Bias    \left( \hat{\tau}_{\text{gDiD}}^{\boldsymbol{X}} \right) - \Bias    \left( \hat{\tau}_{g\text{DiD}}^{\boldsymbol{X}, \boldsymbol{Y^T}} \right)  
\end{align*}
where
\begin{align*}
\widehat{\Delta}_{\tau_x} &:= \left(\hat{\vec{\beta}}_{x,T} - \frac{\sum_{t=0}^{T-1} \hat{\vec{\beta}}_{x,t}}{T}\right)^{\top} \hat{\vec{\delta}}_x \\
\widehat{\Delta}_{\tau_{x,xy}}  &:=
  \widehat{\Cov}\left( \tilde{Y}_{i, T}, \tilde{Y}_{i, 0:T-1} \mid Z_i = 0 \right) \left(\widehat{\Var}\left( \tilde{Y}_{i, 0:T-1}\right)\right)^{-1} \left(\widehat{\mathbb{E}}\left[ \tilde{Y}_{i, 0:T-1}  \mid Z_i = 1 \right]
- \widehat{\mathbb{E}}\left[ \tilde{Y}_{i, 0:T-1} \mid Z_i = 0 \right] \right) \\
& - \left(\widehat{\mathbb{E}}\left[ \bar{\tilde{Y}}_{i, 0:T-1}  \mid Z_i = 1 \right]
- \widehat{\mathbb{E}}\left[ \bar{\tilde{Y}}_{i, 0:T-1}  \mid Z_i = 0 \right] \right)
\end{align*}
and $\tilde{Y}_{i,t}$ is the aforementioned residualized outcome at period $t$ with $\hat{\vec{\beta}}_{x,t}$ as the corresponding regression coefficients; $\bar{\tilde{Y}}_{i,\cdot}$ is thus an average of the residualized outcomes; $\widehat{\mathbb{E}},\widehat{\Var},\widehat{\Cov}$, $\hat{\vec{\delta}}_x$ are the plug-in estimators of mean, variance, covariance, ${\vec{\delta}}_x$ (defined in Equation \ref{eq: specific PT under LSEM}) respectively.
\end{proposition}
\noindent We proceed this part with general multi-dimensional latent variable as proofs do not require $\boldsymbol{\theta} \in \mathbf{R}.$

\begin{proof}[Proof of $\widehat{\Delta}_{\tau_{x}}$]
From Lemma \ref{Lemma 4.1}, we know that 
\[
\text{Bias}  \left( \hat{\tau}_{\text{gDiD}} \right) - \text{Bias}    \left( \hat{\tau}_{\text{gDiD}}^{\boldsymbol{X}} \right) = \left( \Sigma_{X X}^{-1} \Sigma_{X \theta} \vec{\Delta}_{\theta}   + \vec{\Delta}_{X} \right)^{\top} \vec{\delta}_{x}.
\]
From Lemma \ref{Lemma: Residualization trick}, it follows that  $\hat{\vec{\beta}}_{x,T} - \frac{\sum_{t=0}^{T-1} \hat{\vec{\beta}}_{x,t}}{T} \xrightarrow{p} \vec{{\Delta}}_X + \Sigma_{XX}^{-1} \Sigma_{X \theta} \vec{\Delta}_{\theta}.$ \\
\\
As $\boldsymbol{X}$ are observed covariates, by WLLN, it easily follows that $\hat{\vec{\delta}}_x \xrightarrow{p} \vec{\delta}_{x}$ and then by continuous mapping theorem,  it leads to the desired result that  $$\left(\hat{\vec{\beta}}_{x,T} - \frac{\sum_{t=0}^{T-1} \hat{\vec{\beta}}_{x,t}}{T}\right)^{\top} \hat{\vec{\delta}}_x \xrightarrow{p} \left( \Sigma_{X X}^{-1} \Sigma_{X \theta} \vec{\Delta}_{\theta}   + \vec{\Delta}_{X} \right)^{\top} \vec{\delta}_{x}.$$  
\end{proof}

\begin{proof}[Proof of $\widehat{\Delta}_{\tau_{x,xy}}$]
From Lemma \ref{Lemma 4.1}, we know that 
\[
\text{Bias}    \left( \hat{\tau}_{\text{gDiD}}^{\boldsymbol{X}} \right) - \text{Bias}     \left( \hat{\tau}_{g\text{DiD}}^{\boldsymbol{X}, \boldsymbol{Y^T}} \right) = \underbrace{\vec{\beta}_{\theta,T}^{\top} \mathbf{r}_{\theta|x} \vec{\delta}_{\tilde{\theta}}}_{\text{first term}} - \underbrace{\frac{\sum_{t=0}^{T-1} \vec{\beta}_{\theta,t}}{T} \vec{\delta}_{\tilde{\theta}}}_{\text{second term}}
\]
\textbf{Part 1:} We start with consistent estimation of the second term. From Lemma \ref{Lemma: Residualization trick}, for any pre-intervention periods, i.e., $t \le T-1$, we have $$\widehat{\mathbb{E}}\left[ \tilde{Y}_{i, t} \mid Z_i = 1 \right]
- \widehat{\mathbb{E}}\left[ \tilde{Y}_{i, t} \mid Z_i = 0 \right] \xrightarrow{p} {\mathbb{E}}\left[ \tilde{Y}_{i, t} \mid Z_i = 1 \right]
-{\mathbb{E}}\left[ \tilde{Y}_{i, t} \mid Z_i = 0 \right] 
 = \vec{\beta}_{\theta,t}^{\top} \vec{\delta}_{\tilde{\theta}}.$$
Averaging over all those periods leads to the consistent estimation of the second term, i.e., 
\[
\widehat{\mathbb{E}}\left[ \bar{\tilde{Y}}_{i, 0:T-1}  \mid Z_i = 1 \right]
- \widehat{\mathbb{E}}\left[ \bar{\tilde{Y}}_{i, 0:T-1}  \mid Z_i = 0 \right]  \xrightarrow{p}  \frac{\sum_{t=0}^{T-1} \vec{\beta}_{\theta,t}}{T} \vec{\delta}_{\tilde{\theta}}.
\]
\textbf{Part 2:} Now it remains to show consistent estimation of the first term. In particular,
$$
\widehat{\text{Cov}}\left( \tilde{Y}_{i, T}, \tilde{Y}_{i, 0:T-1} \mid Z_i = 0 \right) \left(\widehat{\text{Var}}\left( \tilde{Y}_{i, 0:T-1}\right)\right)^{-1} \left(\widehat{\mathbb{E}}\left[ \tilde{Y}_{i, 0:T-1}  \mid Z_i = 1 \right]
- \widehat{\mathbb{E}}\left[ \tilde{Y}_{i, 0:T-1} \mid Z_i = 0 \right] \right) \xrightarrow{p} \vec{\beta}_{\theta,t}^{\top} \mathbf{r}_{\theta|x} \vec{\delta}_{\tilde{\theta}}.
$$
For any pre-intervention periods, i.e., $t \le T-1$, by Assumption \ref{Assumption: errors}, we have  $$\widehat{\text{Var}} \left( \tilde{Y}_{i,t} \right)  \xrightarrow{p} \vec{\beta}_{\theta,t}^{\top} \Sigma_{\tilde{\theta}\tilde{\theta}} \vec{\beta}_{\theta,t} + {\sigma}_E^2.$$
Generalizing such $t$-th diagonal entry of the matrix to the full matrix leads to $$\widehat{\text{Var}} \left( \tilde{Y}_{i,0:T-1} \right)  \xrightarrow{p} B_\theta \Sigma_{\tilde{\theta} \tilde{\theta}} B_\theta^{\top}+ {\sigma}_E^2 \mathbf{I}_{T \times T}= B_\theta \Sigma_{\tilde{\theta} \tilde{\theta}} B_\theta^{\top} + \Sigma_\epsilon, $$ where this follows by independence between $\tilde{\theta}$ and $\epsilon_{i,t}$ and $\{\epsilon_{i,t}\}_{t=0}^{T}$ is i.i.d. from Assumption \ref{Assumption: errors}.
Similarly, we can obtain  $\widehat{\text{Cov}}\left( \tilde{Y}_{i, T}, \tilde{Y}_{i, 0:T-1} \mid Z_i = 0 \right)   \xrightarrow{p} \vec{\beta}_{\theta,T}^{\top}  \Sigma_{\tilde{\theta} \tilde{\theta}} B_\theta^{\top}.$

Also, stacking $\tilde{Y}$ (but not averaging) leads to $\widehat{\mathbb{E}}\left( \tilde{Y}_{i, 0:T-1}  \mid Z_i = 1 \right)
- \widehat{\mathbb{E}}\left( \tilde{Y}_{i, 0:T-1} \mid Z_i = 0 \right) \xrightarrow{p} B_{\theta} \vec{\delta}_{\tilde{\theta}}.$ Combining all terms and applying continuous mapping theorem leads to the desired result.
\end{proof}

\clearpage

\subsection{Proof of Proposition \ref{Prop: Estimation Strategies for Relative Variance Reduction}} \label{Appendix: Proof of Estimation Strategies for Relative Variance Reduction} 
\noindent We proceed this part with general multi-dimensional latent variable as proofs do not require $\boldsymbol{\theta} \in \mathbf{R}.$

\begin{proof}[Proof of $\widehat{\Var_{\text{core}}}   \left( \hat{\tau}_{g\text{DiD}}\right)$]
By Lemma \ref{Lemma: Residualization trick}, we have
\[
\hat{\vec{\beta}}_{x,T} - \frac{\sum_{t=0}^{T-1} \hat{\vec{\beta}}_{x,t}}{T} \xrightarrow{p} \vec{{\Delta}}_X + \Sigma_{XX}^{-1} \Sigma_{X \theta} \vec{\Delta}_{\theta},
\]
By WLLN, it follows that $\hat{\Sigma}_{XX} \xrightarrow{p} {\Sigma}_{XX}$. Then by continuous mapping theorem,  
\[
\begin{aligned}
& \left(\hat{\vec{\beta}}_{x,T} - \frac{\sum_{t=0}^{T-1} \hat{\vec{\beta}}_{x,t}}{T}\right)^{\top} \hat{\Sigma}_{XX} \left(\hat{\vec{\beta}}_{x,T} - \frac{\sum_{t=0}^{T-1} \hat{\vec{\beta}}_{x,t}}{T}\right) \xrightarrow{p}  \vec{\Delta}_{X}^{\top} \Sigma_{X X} \vec{\Delta}_{X}  + 2 \vec{\Delta}_{\theta}^{\top}   \Sigma_{\theta X} \vec{\Delta}_{X} +  \vec{\Delta}_{\theta}^{\top} \Sigma_{\theta X} \Sigma_{X X}^{-1} \Sigma_{X \theta} \vec{\Delta}_{\theta}.
\end{aligned}
\]  
Further by Assumption \ref{Assumption: errors}, we have  $$\widehat{\text{Var}} \left( \tilde{Y}_{i,T} - \frac{\sum_{t=0}^{T-1}\tilde{Y}_{i,t}}{T} \mid Z_i = 0  \right)  \xrightarrow{p} \vec{\Delta}_{\theta}^{\top} \Sigma_{\tilde{\theta} \tilde{\theta}} \vec{\Delta}_{\theta} + \frac{(T+1) {\sigma}_E^2}{T}.$$  
Now by definition of $\Sigma_{\tilde{\theta} \tilde{\theta}}$ and by Slutsky's theorem, we obtain the desired result. 
\end{proof}

\begin{proof}[Proof of $\widehat{\Var_{\text{core}}}   \left( \hat{\tau}_{g\text{DiD}}^{\boldsymbol{X}} \right)$]
This immediately follows from the last Proof of $\widehat{\Var_{\text{core}}}   \left( \hat{\tau}_{g\text{DiD}}\right)$.
\end{proof}

\begin{proof}[Proof of $\widehat{\Var_{\text{core}}}   \left( \hat{\tau}_{g\text{DiD}}^{\boldsymbol{X}, \boldsymbol{Y^T}} \right)  $]
By similar arguments as above, we have $$\widehat{\text{Var}}\left( \tilde{Y}_{i, T} \mid Z_i = 0 \right)   \xrightarrow{p} \vec{\beta}_{\theta,T}^{\top}  \Sigma_{\tilde{\theta} \tilde{\theta}} \vec{\beta}_{\theta,T}  + \sigma_E^2. $$
It then suffices to only show the consistent estimation of the remaining part involved with the reliability term. Given in Appendix \ref{Appendix: Proofs for Relative (Non-Absolute) Bias Reduction}, we have shown 
\[
\widehat{\text{Cov}}\left( \tilde{Y}_{i, T}, \tilde{Y}_{i, 0:T-1} \mid Z_i = 0  \right)   \xrightarrow{p} \vec{\beta}_{\theta,T}^{\top}  \Sigma_{\tilde{\theta} \tilde{\theta}} B_\theta^{\top},\ \widehat{\text{Var}} \left( \tilde{Y}_{i,0:T-1}  \right)  \xrightarrow{p} B_\theta \Sigma_{\tilde{\theta} \tilde{\theta}} B_\theta^{\top}+ {\sigma}_E^2 \mathbf{I}_{T \times T}= B_\theta \Sigma_{\tilde{\theta} \tilde{\theta}} B_\theta^{\top} + \Sigma_\epsilon.    
\]
Then the result immediately follows from combining all terms and applying the continuous mapping theorem. 
\end{proof}

\clearpage

\subsection{Proof of Lemma \ref{Lemma: Impossibility Results}} \label{Appendix: Impossibility Results} 

\begin{proof}[Proof of Impossibility of Consistent Estimation of Bias]
We proceed with proof by contradiction. Suppose, for contradiction, that there exists a consistent estimator of the bias of $i$, denoted by $\widehat{\operatorname{Bias}}(i)$, such that
\[
\widehat{\operatorname{Bias}}(i)\xrightarrow{p}\operatorname{Bias}(i).
\]
By definition, $\operatorname{Bias}(i)=\mathbb{E}[i]-\tau,$
and by Theorem \ref{Thm: Consistency of DiD and Matching DiD Estimators}, for any $i \in \mathcal{S}$, we have $i \xrightarrow{p} \mathbb{E}[i].$
Define
\[
\hat{\tau} := i - \widehat{\operatorname{Bias}}(i).
\]
By Slutsky's theorem,
\[
\hat{\tau}
= i - \widehat{\operatorname{Bias}}(i)
\xrightarrow{p}
\mathbb{E}[i] - \operatorname{Bias}(i).
\]
Substituting the definition of bias,
\[
\mathbb{E}[i] - \operatorname{Bias}(i)
= \mathbb{E}[i] - (\mathbb{E}[i] - \tau)
= \tau.
\]
Thus, $\hat{\tau} \xrightarrow{p} \tau.$ That said, $\hat{\tau}$ is a consistent estimator of $\tau$, which contradicts the fact that no consistent estimator of $\tau$ exists because of failure of identification strategies. Therefore, there does not exist a consistent estimator of $\operatorname{Bias}(i)$. In fact, for any $i \in \mathcal{S}$, we can only propose to have (e.g., see specific choices in Proposition \ref{Prop: Estimation Strategies for Relative MSE Reduction under Sign Conditions})
\[
\widehat{\Bias}  \left( i \right) \xrightarrow{p} \Bias \left( i \right) + \tau.
\]    
\end{proof}

\begin{proof}[Proof of Impossibility of Consistent Estimation of Relative MSE Reduction]

By definition of MSE in Corollary \ref{Corollary 4.1}, it suffices to show that we can't consistently estimate the relative bias square reduction of any pair among three estimators. We proceed with proof by contradiction. Suppose that we can have  
\[
\widehat{\text{Bias}}^2  \left( \hat{\tau}_{\text{gDiD}} \right) - \widehat{\text{Bias}}^2  \left( \hat{\tau}_{\text{gDiD}}^{\boldsymbol{X}} \right)  \xrightarrow{p} \text{Bias}^2 \left( \hat{\tau}_{\text{gDiD}} \right) - \text{Bias}^2 \left( \hat{\tau}_{\text{gDiD}}^{\boldsymbol{X}} \right) 
\]
Rewriting the left hand side, by continuous mapping theorem, it is equivalent to have 
\[
\underbrace{\left(\widehat{\text{Bias}}  \left( \hat{\tau}_{\text{gDiD}} \right) - \widehat{\text{Bias}}  \left( \hat{\tau}_{\text{gDiD}}^{\boldsymbol{X}} \right) \right)}_{\text{first term}} \cdot \underbrace{\left(\widehat{\text{Bias}}  \left( \hat{\tau}_{\text{gDiD}} \right) + \widehat{\text{Bias}}  \left( \hat{\tau}_{\text{gDiD}}^{\boldsymbol{X}} \right) \right)}_{\text{second term}}  \xrightarrow{p} \text{Bias}^2 \left( \hat{\tau}_{\text{gDiD}} \right) - \text{Bias}^2 \left( \hat{\tau}_{\text{gDiD}}^{\boldsymbol{X}} \right). 
\]
We have shown in Proposition \ref{Prop: Estimation Strategies for Relative Bias Reduction}  that we can consistently estimate the first term. Then to consistently estimate the relative bias square reduction on the right hand side, we must be able to also consistently estimate the second term. However, by continuous mapping theorem, being able to consistently estimate both the first term and the second term simultaneously implies that we can have 
\begin{align*}
& \left(\widehat{\text{Bias}}  \left( \hat{\tau}_{\text{gDiD}} \right) - \widehat{\text{Bias}}  \left( \hat{\tau}_{\text{gDiD}}^{\boldsymbol{X}} \right) \right) + \left(\widehat{\text{Bias}}  \left( \hat{\tau}_{\text{gDiD}} \right) + \widehat{\text{Bias}}  \left( \hat{\tau}_{\text{gDiD}}^{\boldsymbol{X}} \right) \right)  = 2 \widehat{\text{Bias}}  \left( \hat{\tau}_{\text{gDiD}} \right)  \\
& \xrightarrow{p} \left(\text{Bias} \left( \hat{\tau}_{\text{gDiD}} \right) - \text{Bias} \left( \hat{\tau}_{\text{gDiD}}^{\boldsymbol{X}} \right)\right) + \left(\text{Bias} \left( \hat{\tau}_{\text{gDiD}} \right) + \text{Bias} \left( \hat{\tau}_{\text{gDiD}}^{\boldsymbol{X}} \right)\right) =  2\Bias  \left( \hat{\tau}_{\text{gDiD}} \right).      
\end{align*}
This leads to a contradiction because we have just shown that there does not exist such consistent bias estimator. 
In other words, the consistency for $\widehat{\text{Bias}}  \left( \hat{\tau}_{\text{gDiD}} \right) \xrightarrow{p} \Bias  \left( \hat{\tau}_{\text{gDiD}} \right)$ can not hold. Therefore, we can't have a consistent estimator for the relative bias square reduction as proposed in the beginning. 

Similar arguments hold for any pair among three estimators. This concludes the proof. 
\end{proof}

\subsection{Proof of Lemma \ref{Lemma: Bounds on Bias Square under Sign Conditions}} \label{Appendix: Proof of Lemma: Bounds on Bias Square under Sign Conditions}

\begin{proof}
The proof follows from Lemma \ref{Lemma: Impossibility Results}. For any $i \in \mathcal{S}$, we can only propose to have (e.g., see specific choices in Proposition \ref{Prop: Estimation Strategies for Relative MSE Reduction under Sign Conditions})
\[
\widehat{\Bias}  \left( i \right) \xrightarrow{p} \Bias \left( i \right) + \tau.
\]    
Under the notation that $\widehat{\text{Bias}}^2 (i) = \left(\widehat{\Bias}  \left( i \right) \right)^2$, it follows that 
\begin{align*}
&\widehat{\text{Bias}}^2  \left( \hat{\tau}_{\text{gDiD}} \right) \xrightarrow{p} \text{Bias}^2 \left( \hat{\tau}_{\text{gDiD}} \right) + 2 {\text{Bias}}  \left( \hat{\tau}_{\text{gDiD}} \right) \tau + \tau^2, \\
&\widehat{\text{Bias}}^2  \left( \hat{\tau}_{\text{gDiD}}^{\boldsymbol{X}} \right) \xrightarrow{p} \text{Bias}^2 \left( \hat{\tau}_{\text{gDiD}}^{\boldsymbol{X}} \right) + 2 {\text{Bias}}  \left( \hat{\tau}_{\text{gDiD}}^{\boldsymbol{X}} \right) \tau + \tau^2, \\
& \widehat{\text{Bias}}^2  \left( \hat{\tau}_{\text{gDiD}} \right) - \widehat{\text{Bias}}^2  \left( \hat{\tau}_{\text{gDiD}}^{\boldsymbol{X}} \right) \xrightarrow{p}  \text{Bias}^2 \left( \hat{\tau}_{\text{gDiD}} \right) - \text{Bias}^2 \left( \hat{\tau}_{\text{gDiD}}^{\boldsymbol{X}} \right) + 2 \left( {\text{Bias}}  \left( \hat{\tau}_{\text{gDiD}} \right)  - {\text{Bias}}  \left( \hat{\tau}_{\text{gDiD}}^{\boldsymbol{X}} \right)  \right) \cdot \tau,
\end{align*}
and then given the sign conditions, it corresponds to either the upper or lower bound of the relative bias square terms. 

Similar arguments hold for any pair among three estimators. This concludes the proof. 
\end{proof}

\end{document}